\documentclass[
  reprint,
  amsmath,
  amssymb,
  aps,
  prl,
  superscriptaddress,
  ]{revtex4-2}

\usepackage{dsfont}
\usepackage{graphicx}
\usepackage{dcolumn}
\usepackage{bm}
\usepackage{braket}
\usepackage{amsthm}
\usepackage{aliascnt}
\usepackage{tikz}
\usetikzlibrary{arrows.meta, positioning, fit, backgrounds, calc, shapes.geometric}
\usepackage{quantikz}
\usepackage{booktabs}
\usepackage[normalem]{ulem}
\usepackage[ruled,linesnumbered,vlined]{algorithm2e}
\usepackage{diffcoeff}
\usepackage{newtxtext,newtxmath}
\usepackage{titletoc}

\theoremstyle{plain}
\newtheorem{theorem}{Theorem}
\newtheorem{proposition}{Proposition}
\newaliascnt{corollary}{theorem}
\newtheorem{corollary}[corollary]{Corollary}
\aliascntresetthe{corollary}
\newaliascnt{lemma}{theorem}
\newtheorem{lemma}[lemma]{Lemma}
\aliascntresetthe{lemma}
\theoremstyle{definition}
\newaliascnt{definition}{theorem}
\newtheorem{definition}[definition]{Definition}
\aliascntresetthe{definition}

\makeatletter
\@addtoreset{theorem}{section}
\@addtoreset{proposition}{section}
\makeatother
\newcommand{\mainorthesectioncounter}[1]{\ifnum\value{section}=0 \arabic{#1}\else\thesection.\arabic{#1}\fi}
\newcommand{\mainortheHsectioncounter}[1]{\ifnum\value{section}=0 main.\arabic{#1}\else\arabic{section}.\arabic{#1}\fi}

\usepackage{hyperref}
\hypersetup{
    colorlinks=true,
    citecolor=blue,
    linkcolor=red,
    urlcolor=blue,
    anchorcolor=black,
    linktoc=page
}
\usepackage[capitalise]{cleveref}
\usepackage{simpler-wick}
\usepackage{xcolor}

\newcommand{\Tr}{\mathrm{Tr}}
\newcommand{\sign}[1]{\mathrm{sign}\!\left(#1\right)}
\newcommand{\mH}{\mathcal{H}}
\renewcommand{\Re}{\operatorname{Re}}
\renewcommand{\Im}{\operatorname{Im}}
\newcommand{\muxctrl}[1]{\gate[style={diamond,fill=white,draw=black,minimum size=6pt,inner sep=0pt,scale=0.45,transform shape}]{}\vqw{#1}}

\crefname{theorem}{thm.}{theorems}
\Crefname{theorem}{Thm.}{Theorems}
\crefname{proposition}{prop.}{propositions}
\Crefname{proposition}{Prop.}{Propositions}
\crefname{corollary}{cor.}{corollaries}
\Crefname{corollary}{Cor.}{Corollaries}
\crefname{lemma}{lem.}{lemmas}
\Crefname{lemma}{Lem.}{Lemmas}
\crefname{definition}{def.}{definitions}
\Crefname{definition}{Def.}{Definitions}
\crefname{equation}{eq.}{Eqs.}
\Crefname{equation}{Eq.}{Eqs.}
\crefname{figure}{fig.}{figs.}
\Crefname{figure}{Fig.}{Figs.}
\crefname{appendix}{app.}{appendices}
\Crefname{appendix}{App.}{Appendices}

\begin{document}

\let\oldaddcontentsline\addcontentsline
\renewcommand{\addcontentsline}[3]{}

\title{Quantum Solvers for Nonlinear Matrix Equations in Quantum Chemistry}

\author{Pablo Rodenas-Ruiz}
\affiliation{Department of Chemistry and Chemical Biology, Harvard University, Cambridge, MA, USA}
\affiliation{Institute of Physics, Ecole Polytechnique Fédérale de Lausanne (EPFL), Lausanne, Switzerland}
\author{Andrew Zhao}
\affiliation{Quantum Algorithms and Applications Collaboratory, Sandia National Laboratories, Livermore, CA, USA}
\author{Joonho Lee}
\email[Contact author: ]{joonholee@g.harvard.edu}
\affiliation{Department of Chemistry and Chemical Biology, Harvard University, Cambridge, MA, USA}
\affiliation{Google Quantum AI, Venice, CA, USA}

\date{\today}

\begin{abstract}
We present a quantum algorithm for solving algebraic Riccati equations, with applications to quantum-chemical random-phase approximation (RPA) and higher-order RPA theories. Our method block-encodes stabilizing Riccati solutions via Riesz projectors onto invariant subspaces of an associated non-normal matrix, implemented using contour-integral resolvents and quantum singular value transformations. Applied to $m$-particle, $m$-hole RPA, our algorithm yields a block-encoding of the amplitude solution and estimates the electronic correlation-energy density with it. Under localized-orbital sparsity assumptions, the end-to-end cost scales linearly with system size and polynomially with excitation rank $m$, suggesting an exponential advantage in $m$ over plausible classical local-correlation heuristics. More broadly, this work provides a framework for quantum algorithms for nonlinear matrix equations in quantum chemistry and opens a possible route toward developing quantum algorithms for coupled-cluster theory.
\end{abstract}

\maketitle

{\it Introduction.}
Quantum computation has the potential to achieve speedups for certain tasks over classical approaches~\cite{lloyd_uqs_1996, shor_ptap_1997, grover_fqma_1996, szegedy_qsum_2004}. A canonical example is the Harrow--Hassidim--Lloyd (HHL) algorithm, which prepares a state proportional to the solution of the linear system $A\mathbf x=\mathbf b$~\cite{harrow_qals_2009}.
Under certain assumptions about $A \in \mathbb{C}^{n \times n}$, the algorithm has polylogarithmic scaling in $n$, enabling exponential speedups in matrix dimension compared to dense classical solvers that cost $\Omega(n^2)$ operations.
Due to the ubiquity of linear systems of equations, many scientific applications have proposed using HHL-type solvers~\cite{wiebe_qadf_2012, clader_pqls_2013, berry_hoqa_2014, montanaro_qafe_2016, cai_qcmr_2020}.
A relatively less-explored class of quantum solvers is for nonlinear matrix equations; however, such equations are also important for computational science. 

In this Letter, we develop a quantum solver for a particular instance of nonlinear matrix equations: the Riccati equation~\cite{bittanti_re_1991, lancaster_are_1995}. Riccati equations find applications in optimal control, stability theory, and filtering~\cite{benner_nsls_2008,simoncini_cmlm_2016,bini_nsar_2012}.
In the context of quantum chemistry, they appear directly in ring coupled-cluster doubles (rCCD) and random-phase approximation (RPA) theory~\cite{scuseria_drccd_2008, klopper_sfrc_2011,eshuis_ecmb_2012, rekkedal_cagr_2013}, as well as its $m$-excitation generalizations ($m$-RPA)~\cite{ring_nmbp_1980, suhl_hrpa_1961}.
By applying our quantum solver to this problem, we demonstrate an exponential speedup in the excitation rank $m$ compared to possible classical alternatives, opening a path toward a quantum advantage in solving nonlinear matrix equations.

{\it Riccati equations.}
We begin with the definition of the continuous-time algebraic Riccati equation (CARE)~\cite{lancaster_are_1995},
\begin{equation}
\label{eq:CARE-main}
     XQX - XP  - P^\dagger X- R = 0,
\end{equation}
where $P,Q,R \in \mathbb{C}^{n \times n}$ are given, $X$ is the solution, and $Q,R$ are Hermitian.
A standard way to characterize CARE solutions is through invariant subspaces of an associated matrix. Given
\begin{equation}
    \label{eq:linearization-CARE}
    \mH\begin{pmatrix}
        U\\
        V
    \end{pmatrix} = \begin{pmatrix}
        U\\
        V
    \end{pmatrix} \Lambda, \quad
   \mathrm{where} \quad\mH = \begin{pmatrix}
         P & -Q \\
        - R & -P^\dagger
    \end{pmatrix},
\end{equation}
and if $U$ is invertible, one finds $X = VU^{-1}$ which solves \Cref{eq:CARE-main}. Moreover, $U\Lambda U^{-1}=P-QX$, so the spectrum of $\mH$ associated with the invariant subspace determines the stability of the corresponding CARE solution.

For the CARE, stabilizing and antistabilizing solutions are often the most relevant in applications. They are obtained from $n$-dimensional stable and antistable invariant subspaces of $\mH$, i.e., those associated with eigenvalues satisfying $\Re(\lambda)<0$ and $\Re(\lambda)>0$, respectively.
Since $\mH$ is a Hamiltonian matrix~\footnote{Not to be confused with the Hamiltonian of a quantum system. A matrix $\mH \in \mathbb{C}^{2n \times 2n}$ is said to be \emph{Hamiltonian} if $\mathcal{J}\mH$ is Hermitian, where $\mathcal{J} = \begin{psmallmatrix} 0 & \mathbb{I}\\-\mathbb{I} & 0 \end{psmallmatrix}$.}, its spectrum is symmetric with respect to the imaginary axis. Therefore, if it has no eigenvalues on the imaginary axis, its spectrum splits into equally sized stable and antistable branches, ensuring that the stabilizing and antistabilizing solutions are unique when they exist. There have been detailed studies on the existence and uniqueness of solutions~\cite{bittanti_re_1991, lancaster_are_1995}; the assumptions relevant to this work are discussed in \Cref{app:appA} of the Supplemental Material (SM)~\cite{supplemental}.

Prior work by Liu \emph{et al.}~\cite{liu_qamg_2025} gave a quantum algorithm to solve the CARE under the additional assumptions that $Q$ and $R$ are positive definite, and that $Q^{-1}P = (Q^{-1}P)^\dagger$. In this case, $X$ can be written as a matrix geometric mean,~$X = Q^{-1}P \pm Q^{-1/2} (Q^{1/2} S Q^{1/2})^{1/2} Q^{-1/2}$, where $S = P^\dagger Q^{-1}P + R$. This solution can be block-encoded using the quantum singular-value transformation (QSVT) for matrix inverses, powers, and products~\cite{liu_qamg_2025}. However, their assumptions only encompass a highly structured subset of CAREs, rather than the generic stabilizing regime useful in practice. In particular, it can be verified that their solution is a special case of a stabilizing solution~\cite{supplemental}.
We explore an alternative route to a stabilizing solution without these additional assumptions, allowing us to directly apply the resulting algorithm to RPA equations.

{\it Random phase approximation.}
One of the most accurate and widely used quantum chemistry approaches is the coupled-cluster (CC) method~\cite{bartlett2007coupled}, where the many-body ground state $\ket{\Psi}$ is approximated by $\ket{\Psi} \approx \exp(\hat{T})\ket{\mathrm{HF}}$, with $\hat{T}$ the excitation operator and $\ket{\mathrm{HF}}$ the mean-field Hartree--Fock (HF) reference state.
The excitation operator contains a set of amplitudes, $T$, which is obtained from solving a nonlinear tensor equation in $T$.
A particular simplification of CC with double-excitations (CCD) leads to the rCCD or RPA amplitude equation that takes precisely
the form of a Riccati equation~\cite{scuseria_drccd_2008}.

Following \Cref{eq:CARE-main}, we take
\begin{equation}\label{eq:RPA_matrix-main}
     P = -A, \quad Q = B, \quad R = -B,
\end{equation}
where $A$ and $B$ are matrices defined in terms of the Fock matrix and the antisymmetrized two-electron repulsion integrals, $\braket{ib||aj}$ and $\braket{ab||ij}$. We use $p$ for arbitrary, $i,j$ for occupied, and $a,b$ for unoccupied orbitals, respectively (see the SM for details~\cite{supplemental}). Note that the Hamiltonian matrix $\mH$ is the negative of the typical RPA eigenvalue equation.
With $N$ denoting the system size, the classical cost of finding the solution, $X = T$, generically scales as $\mathcal{O}(N^6)$, although $\mathcal{O}(N^5)$ is possible with slight modifications~\cite{hesselmann2012random} or adiabatic connection~\cite{bleiziffer2012resolution}, and $\mathcal{O}(N^4)$ by neglecting exchange contributions~\cite{bleiziffer2012resolution, eshuis_fcmr_2010}.

The RPA eigenvalue problem given by \Cref{eq:linearization-CARE,eq:RPA_matrix-main} can be generalized by enlarging the excitation manifold beyond the single particle--hole sector~\cite{suhl_hrpa_1961, sawicki_hrpa_1962, rowe_eom_1968}, and can also be formulated in particle--particle and hole--hole excitation spaces~\cite{ring_nmbp_1980,scuseria_ppqr_2013}.
Without loss of generality, we focus on the $m$-particle, $m$-hole ($mpmh$) generalization, which we refer to as $m$-RPA.
The most widely studied member of this hierarchy is $m=2$, the second RPA (SRPA), which augments the usual $1p1h$ space with $2p2h$ configurations~\cite{providencia_vamb_1965, yannouleas_mndi_1983, yannouleas_srpa_1987}.
The $m$-RPA equations can be derived from Rowe's equation-of-motion formalism by truncating the excitation manifold at rank $m$~\cite{rowe_eom_1968}.
With a HF reference, the resulting $m$-RPA problem is defined by
\begin{equation}\label{eq:RPA_AB}
\begin{split}
    A_{\mu_\alpha,\nu_\beta}^{\alpha,\beta}
    &=
    \bra{\mathrm{HF}}
    [K^{(\alpha)}_{\mu_\alpha},
    [\hat{H},K^{(\beta)\dagger}_{\nu_\beta}]]
    \ket{\mathrm{HF}},\\
    B_{\mu_\alpha,\nu_\beta}^{\alpha,\beta}
    &=
    -\bra{\mathrm{HF}}
    [K^{(\alpha)}_{\mu_\alpha},
    [\hat{H},K^{(\beta)}_{\nu_\beta}]]
    \ket{\mathrm{HF}},
\end{split}
\end{equation}
where $\hat{H}$ is the many-body Hamiltonian and $K^{(\alpha)\dagger}_{\mu_\alpha}
= a^\dagger_{a_1}\cdots a^\dagger_{a_\alpha}
a_{i_\alpha}\cdots a_{i_1}$ is a rank-$\alpha$ excitation operator written in the particle/hole basis, with $a_k$ and $i_k$ labeling particle and hole orbitals, respectively.
Since the excitation manifold contains configurations up to rank $m$, the dimensions of $A$ and $B$ scale as $\mathcal{O}(N^{2m})$.
A detailed derivation of the $m$-RPA equations is provided in \Cref{app:appD} of the SM~\cite{supplemental}.

The correlation energy is given by $E_c = \frac{1}{4}\Tr(BT)$, where ${T}={V}{U}^{-1}$ is the solution to the associated CARE.
Dense linear-algebra methods for solving $m$-RPA without exploiting the sparsity of the input matrices scale as $\mathcal{O}(N^{6m})$.
This steep scaling of $m$-RPA, together with the rather unsatisfying numerical performance at $m=2$~\cite{peng_rsrp_2014}, has limited the broad adoption of higher-order RPA in the quantum chemistry community.
While approaches to improve the accuracy of higher-order RPA, especially SRPA, via the subtraction method~\cite{tselyaev_smsc_2013, gambacurta_smsr_2015, gambacurta_nbdh_2025} are promising,
we believe that the development of quantum algorithms to remove the scaling wall could lead to a wider community interest in these methods in the future.

We now describe our quantum algorithm for computing a stabilizing solution to the CARE. After establishing the general framework, we analyze the end-to-end cost of applying it to the correlation energy calculation in $m$-RPA, under physically motivated sparsity assumptions on the inputs $A$ and $B$.

{\it Riesz projectors.}
Invariant-subspace methods provide a robust classical route to solving CAREs.
A prototypical example is the matrix-sign-function method, which first embeds the CARE into its associated linear invariant-subspace problem, and then constructs projectors onto the stable and antistable subspaces that determine the Riccati solutions~\cite{roberts_lmrs_1980, byers_msfm_1997, bai_msfc_1998}.
Under the assumption that $\mH$ has no purely imaginary eigenvalues, the matrix sign function maps eigenvalues in the right half-plane $\mathbb{C}_{>}\coloneqq\{z\in\mathbb{C}:\Re(z)>0\}$ to $+1$, and eigenvalues in the left half-plane $\mathbb{C}_{<}\coloneqq\{z\in\mathbb{C}:\Re(z)<0\}$ to $-1$. Hence, it defines stable and antistable spectral projectors,
\begin{equation}
\Pi_s=\frac{1}{2}\left(\mathbb{I}-\sign{\mH}\right), \qquad \Pi_a=\frac{1}{2}\left(\mathbb{I}+\sign{\mH}\right).
\end{equation}
We will be interested in the stabilizing solution $X_s$, which is characterized by the fact that the stable subspace admits a basis of the block form $\begin{pmatrix}{\mathbb{I}}& {X_s}\end{pmatrix}^\mathsf{T}$. Since this subspace is the kernel of $\Pi_a$, the block matrix is annihilated by $\Pi_a$. Writing $\Pi_{a}=\begin{pmatrix}\Pi_1 & \Pi_2\end{pmatrix}$ where $\Pi_1,\Pi_2\in\mathbb{C}^{2n\times n}$, we have
\begin{equation}
\Pi_{a}\begin{pmatrix}\mathbb{I}\\X_s\end{pmatrix}=\Pi_1+\Pi_2X_s = 0,
\end{equation}
and the stabilizing CARE solution is obtained by solving the overdetermined system $\Pi_2X_s=-\Pi_1$.

These projectors can equivalently be written as contour integrals of the resolvent, namely Riesz projectors. This equivalence is the starting point of our approach: rather than approximating the matrix sign function by a polynomial, which is challenging to apply to the spectrum of a non-normal matrix such as $\mH$~\cite{low_qevt_2026}, we instead construct block-encodings of the relevant spectral projectors through this contour-integral representation. This allows us to subsequently handle the problem with standard QSVT techniques~\cite{gilyen_qsvt_2019}.
Our high-level strategy is to block-encode the Riesz projectors $\Pi_a$ and $\Pi_s$, constructed from their Cauchy integral representation using an exponentially convergent trapezoid rule for periodic contours~\cite{trefethen_ectr_2014}.
We then express the CARE solution $X$ in terms of these projectors, yielding a query-efficient circuit that block-encodes $X$.

{\it Block-encoding the stabilizing solution.}
A basic requirement for stabilizing solutions to exist is that $\mH$ has no eigenvalues on the imaginary axis. This is generally true under standard assumptions for the coefficient matrices (see \Cref{app:appA} of the SM for further details~\cite{supplemental}). In that case, the spectrum has the following split:
\begin{equation*}
    \sigma(\mH) = \Sigma_s\cup\Sigma_a, \quad \Sigma_s \subset \mathbb{C}_<, \quad \Sigma_a \subset \mathbb{C}_>, \quad \sigma(\mH)\cap i\mathbb{R} = \emptyset,
\end{equation*}
and there is a nonzero gap $\delta \coloneqq \min_{\lambda\in\sigma(\mH)}|{\Re(\lambda)}|$ to the imaginary axis. Then the matrix sign function is well-defined and yields projectors $\Pi_a$ and $\Pi_s$. Equivalently, these projectors admit the contour-integral representation
\begin{equation}
    \Pi_{a} = \frac{1}{2\pi i}\oint_{\Gamma_a} (z\mathbb{I}-\mH)^{-1}dz,
    \label{eq:riesz-projector}
\end{equation}
where $(z\mathbb{I}-\mH)^{-1}$ is the resolvent of $\mH$ and $\Gamma_a$ is a positively oriented closed contour that encloses only the antistable spectrum $\Sigma_a$. For the projector to be well-defined, the contour must lie in the resolvent set $\varrho(\mH)$, the set of points where $z\mathbb{I}-\mH$ is nonsingular. The stable Riesz projector is defined analogously with a contour $\Gamma_s$ enclosing $\Sigma_s$.

Recall that we write $\Pi_a = (\Pi_1 \quad \Pi_2)$. When the stabilizing solution $X_s$ exists, $\Pi_2$ has full column rank and so we can find $X_s=-\Pi_2^+\Pi_1$ where $\Pi_2^+=(\Pi_2^\dagger \Pi_2)^{-1}\Pi_2^\dagger$ is the Moore--Penrose pseudoinverse satisfying $\Pi_2^+\Pi_2=\mathbb{I}_n$.
Therefore, we aim to block-encode $\Pi_a$, extract its sub-blocks $\Pi_1$ and $\Pi_2$, and apply QSVT to block-encode $\Pi_2^+$. The solution $X_s$ is then constructed by a simple product of these block-encoded matrices. A high-level workflow of this procedure is presented in \Cref{fig:flowchart}.
\begin{figure*}[t]
    \centering
    \resizebox{\textwidth}{!}{\definecolor{CB1}{RGB}{8,48,107}
\definecolor{CB2}{RGB}{33,83,155}
\definecolor{CB3}{RGB}{66,133,211}
\definecolor{CB4}{RGB}{158,202,235}
\definecolor{CB5}{RGB}{222,235,247}
\definecolor{CBOut}{RGB}{210,229,244}
\definecolor{CBInFill}{RGB}{236,243,249}
\definecolor{CBInDraw}{RGB}{72,118,178}
\tikzset{
    box/.style={
        rectangle, draw=CB2, fill=CB5, thick, rounded corners=4pt,
        minimum width=34mm, minimum height=14mm, text width=32mm,
        align=center, inner sep=5pt, font=\normalsize
    },
    nohyph/.style={
        execute at begin node={\hyphenpenalty=10000\exhyphenpenalty=10000}
    },
    inoutin/.style={
        rectangle, draw=CBInDraw, fill=CBInFill, text=black, thick,
        rounded corners=4pt, minimum width=34mm, minimum height=14mm,
        text width=32mm, align=center, inner sep=5pt, font=\normalsize
    },
    inoutout/.style={
        rectangle, draw=CB2, fill=CBOut, text=black, thick,
        rounded corners=4pt, minimum width=34mm, minimum height=14mm,
        text width=32mm, align=center, inner sep=5pt, font=\normalsize
    },
    arr/.style={->, >={Latex[length=2.5mm,width=1.5mm]}, thick, draw=CB2},
    handoff/.style={
        ->, >={Latex[length=2.5mm,width=1.5mm]}, thick,
        draw=CB2, rounded corners=1.5mm
    },
    rowtitle/.style={font=\bfseries\normalsize, text=black},
    boxtitle/.style={font=\bfseries, text=black},
    stepbadge/.style={
        circle, fill=CB1, text=white, font=\bfseries\scriptsize,
        inner sep=0pt, minimum size=4.5mm, draw=white, line width=0.5pt
    }
}
\begin{tikzpicture}[node distance=5mm and 5.5mm]
\node[inoutin] (P1)
  {{\textbf{Input}}\\[2pt]
   $U_{\mH}$: $(\alpha_{\mH},\varepsilon_{\mH})$-block-encoding\\[-1pt]
   of $\mH$};
\node[box, right=of P1] (P2)
  {{\textbf{Contour quadrature}}\\[2pt]
   Choose $\Gamma_a$ and use the trapezoid rule\\[-1pt]
   on $M$ nodes for $\Pi_a$};
\node[box, right=of P2] (P3)
  {{\textbf{Resolvents}}\\[2pt]
   LCU and QSVT inversion\\[-1pt]
   for $(z_k\mathbb{I}-\mH)^{-1}$};
\node[box, nohyph, right=of P3] (P4)
  {{\textbf{LCU over nodes}}\\[2pt]
   Combine the\\[-1pt]
   resolvents across\\[-1pt]
   quadrature points};
\node[inoutout, right=of P4] (P5)
  {{\textbf{Output}}\\[2pt]
   $U_{\Pi_a}$: $(\alpha_{\Pi},\varepsilon_{\Pi})$-block-encoding\\[-1pt]
   of $\Pi_a$};
\draw[arr] (P1) -- (P2);
\draw[arr] (P2) -- (P3);
\draw[arr] (P3) -- (P4);
\draw[arr] (P4) -- (P5);
\node[stepbadge] at (P2.north east) {1};
\node[stepbadge] at (P3.north east) {2};
\node[stepbadge] at (P4.north east) {3};
\node[rowtitle] (Ptitle) at ($ (P1.north west)!0.5!(P5.north east) + (0,6mm) $)
  {Block-encoding the spectral projector $\Pi_a$};
\node[inoutin, below=16mm of P1] (C1)
  {{\textbf{Input}}\\[2pt]
   $U_{\Pi_a}$: block-encoding of $\Pi_a=(\Pi_1\;\Pi_2)$\\[-1pt]
   with $\Pi_2$ full column rank};
\node[box, right=of C1] (C2)
  {{\textbf{Extract blocks}}\\[2pt]
   Form $\Pi_1=\Pi_a \begin{psmallmatrix}\mathbb{I}\\0\end{psmallmatrix}$ and $\Pi_2=\Pi_a \begin{psmallmatrix}0\\\mathbb{I}\end{psmallmatrix}$\\[-1pt]
   using rectangular block-encodings};
\node[box, right=of C2] (C3)
  {{\textbf{Pseudoinverse of $\Pi_2$}}\\[2pt]
   Apply QSVT to \\[-1pt]
   obtain a block- \\[-1pt]
   encoding of $\Pi_2^{+}$};
\node[box, right=of C3] (C4)
   {{\textbf{CARE solution}}\\[2pt]
   Block-encoding product gives $X_s=-\Pi_2^{+}\Pi_1$};
\node[inoutout, right=of C4] (C5)
  {{\textbf{Output}}\\[2pt]
   $U_{X_s}$: $(\alpha_X,\varepsilon_X)$-block-encoding\\[-1pt]
   of $X_s$};
\draw[arr] (C1) -- (C2);
\draw[arr] (C2) -- (C3);
\draw[arr] (C3) -- (C4);
\draw[arr] (C4) -- (C5);
\node[stepbadge] at (C2.north east) {4};
\node[stepbadge] at (C3.north east) {5};
\node[stepbadge] at (C4.north east) {6};
\node[rowtitle] (Ctitle) at ($ (C1.north west)!0.5!(C5.north east) + (0,6mm) $)
  {Extracting the stabilizing CARE solution $X_s$};
\coordinate (h1) at ($(P5.south)+(0,-6mm)$);
\coordinate (h2) at ($(C1.north)+(0,6mm)$);
\coordinate (hmidR) at ($(h2 |- h1)$);
\draw[handoff] (P5.south) -- (h1) -- (hmidR) -- (C1.north);
\end{tikzpicture}}
    \caption{Workflow of the quantum CARE solver algorithm.}
    \label{fig:flowchart}
\end{figure*}
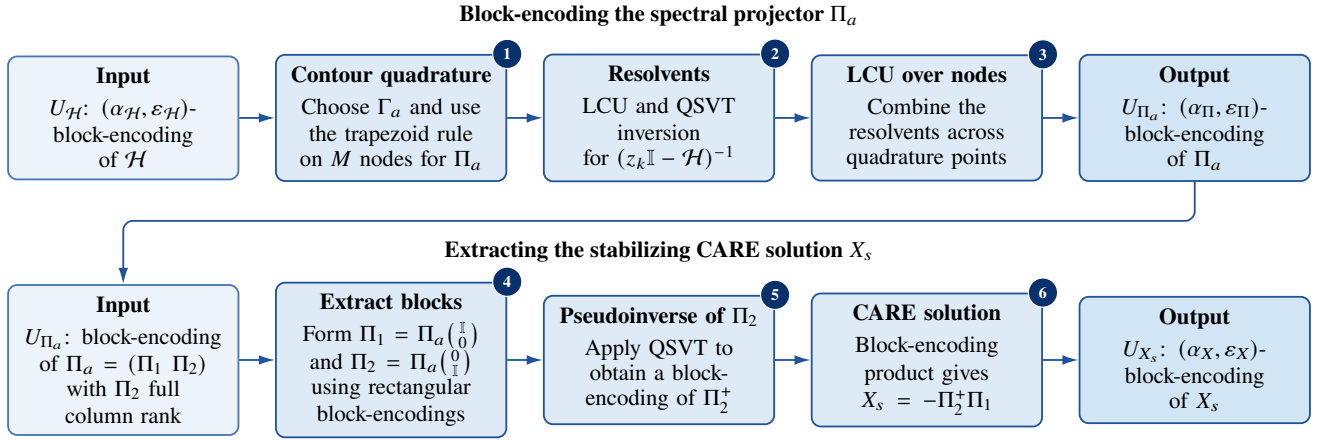
To approximate the Riesz projector, we specialize to periodic contours and apply the trapezoid rule. Let the contour $\Gamma_a\subset\varrho(\mH)$ be parametrized by a $2\pi$-periodic, $C^1$ curve $\gamma(\theta)$ that encloses $\Sigma_a$ and excludes $\Sigma_s$. Then
\begin{equation}
\label{eq:riesz-periodic-main}
\Pi_a=\frac{1}{2\pi i} \int_0^{2\pi} \gamma'(\theta) \left(\gamma(\theta)\mathbb{I}-\mH\right)^{-1} d\theta.
\end{equation}
With $h(\theta)\coloneqq\frac{1}{i}\gamma'(\theta)(\gamma(\theta)\mathbb{I}-\mH)^{-1}$, the $M$-node trapezoid rule at $\theta_k=2\pi k/M$ gives
\begin{equation}
\label{eq:trap-periodic-main}
\Pi_{a,M}=\frac{1}{M}\sum_{k=0}^{M-1} h(\theta_k),
\end{equation}
which converges exponentially for analytic periodic integrands, assuming that $\gamma$ extends holomorphically to a strip around the real axis.

\begin{proposition}[Informal]
\label{prop:informal-trapezoid}
Let $\Gamma_a\subset \varrho(\mH)$ be a periodic contour enclosing $\Sigma_a$ and excluding $\Sigma_s$. Assume its parametrization $\gamma$ is analytic in a strip of width $\eta$ about the real axis. Then the $M$-node trapezoid approximation $\Pi_{a,M}$ to $\Pi_a$ achieves spectral-norm error at most $\varepsilon_\mathrm{trap}>0$, for
\begin{equation}
M=\mathcal{O}\mathopen{}\left(\frac{1}{\eta} \log\frac{\gamma_\pm}{\varepsilon_\mathrm{trap}} \right),
\end{equation}
where $\gamma_\pm=\max_{\theta\in[0,2\pi]}\|h(\theta\pm i\eta)\|$.
\end{proposition}

The formal statement with all technical details and proofs is provided in \Cref{app:appB} of the SM~\cite{supplemental}. This exponential convergence in $M$ matches related bounds for approximating functions of matrices~\cite{trefethen_ectr_2014}. Beyond the target error $\varepsilon_\mathrm{trap}$, the main parameters controlling the quadrature complexity are the strip width $\eta$, which quantifies the analyticity region available to the parametrization, and the contour-resolvent bound $\gamma_\pm$. $\gamma_\pm$ depends on the derivative $|\gamma'|$ and the resolvent norm along the shifted contours. For normal matrices, the resolvent norm scales with the inverse of the distance to the spectrum, but for non-normal matrices it may become large even far from the spectrum due to pseudospectral effects. Although we do not cover them here, conformal maps may be used to make the analyticity region larger, enabling more efficient quadratures with fewer nodes~\cite{hale_calr_2008, trefethen_ectr_2014}.

We propose a smoothed-semicircle parametrization for $\Gamma_a$,
\begin{equation}
    \label{eq:smooth-semicircle}
    \gamma_{\mathrm{SD}}(\theta) = z_0 + \frac{1}{2}\left(R\cos(\theta) + \sqrt{\omega^2+R^2\cos^2(\theta)} \right)+ iR\sin(\theta).
\end{equation}
For this contour, the maximal analyticity strip is determined by the branch points of the square-root term, giving $\chi=\operatorname{arsinh}(\omega/R)$, and one may choose any admissible $0<\eta<\chi$ for which the shifted contours remain in the resolvent set. A parameter choice achieving the desired spectral separation is $z_0 = 0$, $R = \mathcal{O}(\alpha_\mH)$, and $0<\omega<2\delta$. Moreover, the contour satisfies $|\gamma'_{\mathrm{SD}}(\theta)|\leq R$ and has length smaller than $2\pi R$, so its size remains controlled even when the gap $\delta$ is small. Having established the contour, it is straightforward to block-encode the contour integral following similar approaches to prior art~\cite{takahira_qamf_2020, takahira_qabb_2022,jiang_cibq_2026}.

We say that a unitary $U_A$ is an $(\alpha, \varepsilon)$-block-encoding of a matrix $A$ if
\begin{equation}
    U_A = \begin{pmatrix}
        \widetilde{A}/\alpha & * \\
        * & *
    \end{pmatrix}, \quad \|\widetilde{A} - A\| \leq \varepsilon.
\end{equation}
Suppose we have access to an $(\alpha_\mH, \varepsilon_\mH)$-block-encoding of the CARE matrix $\mH$. To implement $\Pi_{a,M}$, we block-encode the resolvents $(\gamma(\theta_k)\mathbb{I} - \mH)^{-1}$ using QSVT matrix inversion, and then combine them by an LCU whose coefficients are proportional to the quadrature weights. This process yields a block-encoding of $\Pi_a$, summarized below. We assume the smoothed-semicircle contour parametrization.

\begin{theorem}[Informal]
\label{thm:informal-riesz-be}
Let $U_\mH$ be an $(\alpha_\mH, \varepsilon_\mH)$-block-encoding of $\mH$. For the smoothed-semicircle contour $\gamma_\mathrm{SD}$ with $z_0=0$, $\omega=\delta$ and $R=2\alpha_\mH$, the antistable Riesz projector $\Pi_a$ admits an $(\alpha_\Pi, \varepsilon_\Pi+\varepsilon_{\mathrm{trap}})$-block-encoding $U_{\Pi_a}$, with $\alpha_\Pi=\mathcal{O}(M_\gamma\alpha_\mH)$ where $M_\gamma=\max_{\theta\in[0,2\pi]}\|(\gamma_{\mathrm{SD}}(\theta)\mathbb{I}-\mH)^{-1}\|$, and $M$ is chosen as in \Cref{prop:informal-trapezoid} to achieve error $\varepsilon_{\mathrm{trap}}$. It can be implemented with
\begin{equation}
m_\Pi = \mathcal{O}\mathopen{}\Bigg(M_\gamma \alpha_\mH \log\frac{1}{\varepsilon_\mathrm{pol}^{(\Pi)}}\Bigg)
\end{equation}
queries to controlled $U_\mH,U_\mH^\dagger$ and $\widetilde{\mathcal{O}}(m_\Pi M)$ additional one- and two-qubit gates. The space overhead is $\mathcal{O}(\log M)$ and $\varepsilon_\mathrm{pol}^{(\Pi)}$ is the target precision for approximating the resolvents.
\end{theorem}

The complete result and its proof can be found in \Cref{app:appC} of the SM~\cite{supplemental}. Note that we use $\widetilde{\mathcal{O}}(\cdot)$ to indicate asymptotic scaling suppressing polylogarithmic factors. The first row of \Cref{fig:flowchart} shows the steps for building $U_{\Pi_a}$. The main overheads for block-encoding the Riesz projector are the upper bound on the resolvent norm $M_\gamma$ and the subnormalization $\alpha_\mH$. Additionally, $\mathcal{O}(m_\Pi M)$ gates are required for node coefficients, where $M$ scales as $\alpha_\mH/\delta$ in the small-gap regime and logarithmically with $1/\varepsilon_\mathrm{trap}$. $U_\mH$ can be constructed from block-encodings of its blocks $P,Q,R$ (recall \Cref{eq:linearization-CARE}) with the same query complexity. The gate complexity of these individual blocks is dependent on the structure of those matrices; for example, we will use sparsity to reduce encoding costs.

Once we have block-encodings of $\Pi_a$, the CARE solution is constructed by identifying the rectangular sub-blocks $\Pi_1$ and $\Pi_2$, building $\Pi_2^+$ via QSVT, and multiplying to obtain a block-encoding of $X_s = -\Pi_2^+\Pi_1$.

\begin{theorem}[Informal]
\label{thm:informal-care-be}
Assume there exists a stabilizing solution $X_s$ of the CARE. Then, this solution satisfies $X_s=-\Pi_2^+\Pi_1$ for $\Pi_a=(\Pi_1\ \Pi_2)$, and there exists $\kappa_2 \geq 1$ such that the singular values of $\Pi_2/\alpha_\Pi$ lie in $[1/\kappa_2,1]$. Given the block-encoding of $\Pi_a$ from \Cref{thm:informal-riesz-be}, there exists a unitary $U_{X_s}$ which is an $(\alpha_X, \varepsilon_X)$-block-encoding of $X_s$, with $\alpha_X=8\kappa_2/3$ and $\varepsilon_X = \widetilde{\mathcal{O}}\mathopen{}\left( \kappa_2^2 \sqrt{\varepsilon_\Pi / \alpha_\Pi} + \kappa_2 \varepsilon^{(+)}_{\mathrm{pol}} \right)$. It uses
\begin{equation}
m_+=\mathcal{O}\mathopen{}\Bigg(\kappa_2\log\frac{1}{\varepsilon_\mathrm{pol}^{(+)}}\Bigg)
\end{equation}
queries to $U_{\Pi_a}$ and $U_{\Pi_a}^\dagger$ for the pseudoinverse, hence $\mathcal{O}(m_+ m_\Pi)$ queries to controlled $U_\mH,U_\mH^\dagger$ and $\widetilde{\mathcal{O}}(m_+ m_\Pi M)$ additional one- and two-qubit gates.
\end{theorem}

The proof of this statement can be found in \Cref{app:appC} of the SM~\cite{supplemental}. The second row of \Cref{fig:flowchart} shows the extraction steps, given $U_{\Pi_a}$ as input.
The main cost driver for block-encoding $X_s$ is the $\widetilde{\mathcal{O}}(M_\gamma \alpha_\mH \kappa_2)$ queries of $U_\mH$. Note that $\kappa_2$ must be at least $\alpha_\Pi/\sigma_{\min}(\Pi_2)$, and so acts as a ``supernormalized'' condition number of $\Pi_2$.

With the solution $X_s$ in hand as a block-encoding, one task of interest is to compute traces involving it. For example, in our chemical RPA application we want to calculate $\Tr(BX_s)$ given a sparse matrix $B$. This can be achieved using a modified Hadamard test, combined with amplitude estimation~\cite{brassard_qaae_2002,aaronson_qacs_2020}, to obtain Heisenberg scaling in the target precision $\epsilon$.

\begin{theorem}[Informal]\label{thm:informal-be-trace}
    Let $B,X\in\mathbb{C}^{n\times n}$ and let $U_{X}$ be an $(\alpha_X, \varepsilon_X)$-block-encoding of $X$. There exist states $\ket{\eta_B}$, $\ket{\chi_B} \in \mathbb{C}^{n} \otimes \mathbb{C}^{n}$, whose sparsity depends on $B$, which satisfy $\Tr(BX)=\Lambda_B\bra{\chi_B}(X\otimes\mathbb{I})\ket{\eta_B}$. Here, $\Lambda_B=\sum_\mu\sqrt{\sum_\nu |B_{\mu\nu}|^2}$. Furthermore, if $\varepsilon_X\leq \epsilon/(2\Lambda_B)$ then $\Tr(BX)$ can be estimated to additive error $\epsilon$ with failure probability at most $p_f$ using
    \begin{equation}
        \mathcal{O}\mathopen{}\left(\frac{\alpha_X\Lambda_B}{\epsilon}\log\frac{1}{p_f}\right)
    \end{equation}
    queries to controlled $U_X$, the state-preparation circuits for $\ket{\eta_B}$ and $\ket{\chi_B}$, and their inverses. If $B$ has row sparsity $s$ and $r$ nonzero rows, then $\Lambda_B \leq r\,\sqrt{s}\,\|B\|_\mathrm{max}$.
\end{theorem}

{\it Block-encoding $m$-RPA excitation amplitudes.}
We now analyze the cost of our main application of the quantum CARE solver: estimating the $m$-RPA correlation energy. Our goal is to obtain linear scaling in volume $V$ while achieving quantum speedups in the excitation rank $m$ over classical local-correlation methods~{\cite{saebo_ltec_1993,schutz_losl_1999, schutz_losl_2001, pavosevic_ssir_2016, nagy_olsl_2017}}. Our analysis of the physical setup follows that of Ref.~\cite{chen_frqs_2025}.
We work with a localized orbital basis, with $R_c$ denoting the ``correlation length'' of the system, $D$ the physical dimension, and $V$ the physical volume. In our cost analyses, we omit polylogarithmic overheads due to, e.g., finite-precision arithmetic or quantum read-only memory (QROM) lookup.

First, we analyze the block-encoding cost of $\mH$, denoted by $C^\mathrm{BE}_{\mH}$. From \Cref{eq:RPA_AB}, the matrix $A$ has nonzero blocks only between excitation ranks satisfying $|\alpha-\beta|\le 2$, while $B$ only has one nonzero block, $B^{1,1}_{ia,jb}=\braket{ab||ij}$.
A block-encoding of $\mH$ can be constructed from block-encodings of $A$, $A^\dagger$, and $B$ via a four-term LCU, giving $C_\mH^{\mathrm{BE}}=\mathcal{O}(C_A^{\mathrm{BE}}+C_B^{\mathrm{BE}})$. To implement these, we use a localized-orbital integral-loading model~\cite{berry2019qubitization,chen_frqs_2025}: the one- and two-electron integrals in the localized basis are precomputed, thresholded according to a correlation cutoff $R_c$, and loaded into QROM. The $m$-RPA matrices $A$ and $B$ are not stored explicitly. Instead, sparse-access oracles for $A$ and $B$ compute each requested nonzero entry from the QROM-loaded localized integrals, with only polylogarithmic overhead for QROM lookup and quantum arithmetic~\cite{babbush_eesq_2018, berry2019qubitization,vedral_qnea_1996,haner_oqca_2018}.

For $A$, the one-body part of the commutator requires only local one-electron integrals, contributing $\mathcal{O}(VR_c^D)$ data. The off-diagonal two-body part is generated from the retained local two-electron integrals; there are $\mathcal{O}(VR_c^D)$ choices for one local transition-density pair and $\mathcal{O}(R_c^D)$ nearby choices for the second pair, yielding $\mathcal{O}(VR_c^{2D})$ total entries.
The commutator structure in $A$ retains only connected matrix elements, reducing the number of nonzero elements by $\mathcal{O}(V)$.
The long-range Coulomb tail due to $\braket{pq|pq}$ contributes only to diagonal or one-body-like elements of $A$; rather than loading long-range Coulomb integrals, a separate oracle handles these terms explicitly with overhead $\mathcal{O}(VR_c^D)$. The long-range Coulomb interactions in other off-diagonal elements involving interactions between transition densities decay algebraically as ${\sim}1/R^6$. We neglect these and assume they can be treated by an efficient classical post-processing method widely adopted in local-correlation theory~{\cite{schwilk_secm_2017,hetzer_madp_1998, werner_cmad_2016}}.
Therefore, the cost of block-encoding $A$ is $C_A^{\mathrm{BE}}=\mathcal{O}(VR_c^{2D})$.
For $B$, only the $(1,1)$ block is nonzero, with entries given by the retained local pair--pair integrals (we truncate diagonal long-range Coulomb contributions). The same counting as for $A$ gives row sparsity $s_B=\mathcal{O}(R_c^D)$ and total sparsity $\mathcal{O}(VR_c^{2D})$. Hence the block-encoding cost for $B$ also scales as $C_B^{\mathrm{BE}}=\mathcal{O}(VR_c^{2D})$, and so $C_\mH^{\mathrm{BE}}=\mathcal{O}(VR_c^{2D})$.

The block-encoding normalization for $\mH$ is $\alpha_{\mH}=\mathcal{O}(s_{\mH}\|\mH\|_{\max})$, where $s_\mH \leq s_A + s_B$ is the row sparsity of $\mH$ and $\|\mH\|_{\max} = \mathcal{O}(1)$. For the $A$ block, the relevant sparsity is the maximum number of column excitation strings $(\beta,\nu_\beta)$ connected to a fixed row excitation string $(\alpha,\mu_\alpha)$. Since the Hamiltonian $\hat{H}$ contains at most two-body operators, $A^{\alpha,\beta}=0$ for $|\alpha-\beta|>2$, and only $\mathcal{O}(1)$ $\beta$-rank blocks contribute for fixed $\alpha$. 
For a rank-$\alpha$ excitation string, the one-body term of $\hat{H}$ can change one of $\mathcal{O}(\alpha)$ particle/hole labels into $\mathcal{O}(R_c^D)$ local orbital choices. The two-body part can change at most two labels, giving $\mathcal{O}(\alpha^2)$ label choices and $\mathcal{O}(R_c^{2D})$ local orbital choices. Taking $\alpha \leq m$, the maximum row sparsity over all rank-$\alpha$ blocks is $s_A=\mathcal{O}(m^2R_c^{2D})$.
Since $s_B=\mathcal{O}(R_c^D)$ for the $B$ block, the $A$ block dominates in the two-electron scattering regime, giving $s_{\mH}^{\mathrm{2e}}=\mathcal{O}(m^2R_c^{2D})$. However, following the assumption of Ref.~\cite{chen_frqs_2025}, if single-electron scattering processes account for the dominant correlation effects, we can drop the quadratic terms and consequently $s_{\mH}^{\mathrm{1e}}=\mathcal{O}(mR_c^D)$.

Now we examine the cost of block-encoding the Riesz projector. \Cref{thm:informal-riesz-be} shows that $U_{\Pi_a}$ requires $m_\Pi=\widetilde{\mathcal{O}}(M_\gamma\alpha_\mH)$ queries to $U_\mH$, up to logarithmic factors. Therefore, the cost is $C_{\Pi}^{\mathrm{BE}}=\widetilde{\mathcal{O}}(C_\mH^{\mathrm{BE}}M_\gamma\alpha_\mH)$. Additionally, we treat the $M$-dependent gate cost for multiplexing quadrature coefficients as asymptotically comparable to $C_{\mH}^{\mathrm{BE}}$. Using the integral-loading model as before, with $\alpha_\mH=\mathcal{O}(s_\mH)$, we arrive at $C_{\Pi}^{\mathrm{BE}}=\widetilde{\mathcal{O}}(VR_c^{2D}M_\gamma s_\mH)$. In the two-electron scattering regime, this is $C_{\Pi}^{\mathrm{BE,2e}}=\widetilde{\mathcal{O}}(VM_\gamma m^2R_c^{4D})$, while the single-electron scattering regime has $C_{\Pi}^{\mathrm{BE,1e}}=\widetilde{\mathcal{O}}(VM_\gamma m R_c^{3D})$.
Finally, we apply \Cref{thm:informal-care-be} to construct the block-encoding $U_T$ of the stabilizing solution $X_s = T$, with subnormalization $\alpha_T = \mathcal{O}(\kappa_2)$. This makes $\widetilde{\mathcal{O}}(\kappa_2)$ queries to $U_{\Pi_a}$, and hence the cumulative cost is $C_{T}^{\mathrm{BE}} = \widetilde{\mathcal{O}}(C_{\Pi}^{\mathrm{BE}} \kappa_2)$. We will discuss the scaling of $\kappa_2$ later.

{\it Correlation energy estimation.}
Since $B^{\alpha,\beta}$ is nonzero only in the $(\alpha,\beta)=(1,1)$ block, the relevant observable is $E_c^{(m)}/V=\frac{1}{4V}\Tr(BT)$, where $B \equiv B^{1,1}$ and $T \equiv T_{1,1}^{(m)}$. Then, by \Cref{thm:informal-be-trace} this energy density can be estimated to additive error $\epsilon$ using $\mathcal{O}(\alpha_T \Lambda_B/(V\epsilon))$ queries to $U_T$, and state-preparation circuits for $\ket{\chi_B}$ and $\ket{\eta_B}$. Therefore, the end-to-end cost for estimating $E_c^{(m)}/V$ is $C_E = \mathcal{O}[(C_T^{\mathrm{BE}} + C_{\chi}^{\mathrm{prep}} + C_{\eta}^{\mathrm{prep}}) \kappa_2 \Lambda_B / (V\epsilon)]$.

We now insert the localized-orbital sparsity model into this cost. The $B$ block has $r_B=\mathcal{O}(VR_c^D)$ nonzero rows and row sparsity $s_B=\mathcal{O}(R_c^D)$. Together with $\|B\|_{\max} = \mathcal{O}(1)$, \Cref{thm:informal-be-trace} gives $\Lambda_B\le r_B\sqrt{s_B}\|B\|_{\max}=\mathcal{O}(VR_c^{3D/2})$. The corresponding sparse state-preparation costs are $C_{\chi}^{\mathrm{prep}} + C_{\eta}^{\mathrm{prep}} = \widetilde{\mathcal{O}}(VR_c^{2D})$, which we show in \Cref{app:appE} of the SM~\cite{supplemental}. This is subleading compared to the cost of producing $U_T$, which is $C_T^{\mathrm{BE}}=\widetilde{\mathcal{O}}(VM_\gamma\kappa_2s_\mH R_c^{2D})$. Thus, the asymptotic gate complexity is $C_E=\widetilde{\mathcal{O}}(VM_\gamma\kappa_2^2s_\mH R_c^{7D/2}/\epsilon)$.

In \Cref{app:appC} of the SM~\cite{supplemental}, we show that $\kappa_2 = \mathcal{O}(M_\gamma s_\mH \|T\|)$. We assume that $\|T\| = \mathcal{O}(1)$, which is well justified for weakly correlated systems for which $m$-RPA can be successfully applied. With this, we obtain
\begin{equation*}
    C_E^{\mathrm{2e}} = \widetilde{\mathcal{O}}\mathopen{}\left(\frac{V M_\gamma^3 m^6 R_c^{19D/2}}{\epsilon}\right), \quad 
    C_E^{\mathrm{1e}} = \widetilde{\mathcal{O}}\mathopen{}\left(\frac{V M_\gamma^3 m^3 R_c^{13D/2}}{\epsilon}\right)
\end{equation*}
for the two- and single-electron scattering regimes, respectively.

{\it Roadmap to quantum advantage.}
To our knowledge, there is currently no classical heuristic that uses the Riesz projector and orbital locality to achieve $\mathcal{O}(V)$ scaling for estimating the $m$-RPA correlation energy per volume.
However, one can imagine algorithms similar to CC-like equation solvers that also exploit sparsity. 
We expect that such an algorithm would scale no better than linearized coupled-cluster~\cite{saebo_ltec_1993,schutz_losl_2001,chen_frqs_2025}, which is $\mathcal{O}(VR_c^{(2m+1)D})$ in the single-electron scattering regime. Note that this is the explicit dependence on volume and correlation length; other factors, such as the number of iterations or condition number, may implicitly scale with $R_c$ as well.
Already for $m=3$, $D=3$, our bound of $\mathcal{O}(R_c^{19.5})$ strictly beats this more optimistic classical scaling of $\mathcal{O}(R_c^{21})$, and we anticipate that tighter analyses may improve our bounds further.
We leave a more detailed analysis and comparison for future studies.

{\it Conclusions.}
In this Letter, we have presented a quantum algorithm for solving the Riccati equation and discussed its potential application in quantum chemistry. Our algorithm block-encodes a stabilizing solution to the Riccati equation into an efficient quantum circuit, and subsequently estimates the solution's trace inner product with a sparse coefficient matrix. In the context of $m$-RPA, this quantity is precisely the electronic correlation energy. Our method achieves linear scaling with the physical volume $V$, and we expect that the observed exponential speedup in excitation rank $m$ will persist, even as classical methods develop.
Our work also motivates a research direction to extend to quantum algorithms for nonlinear \emph{tensor} equations, namely the coupled-cluster equations~\cite{bartlett2007coupled, baskaran_ahhl_2023, chen_frqs_2025}. Indeed, CCSD(T) is often regarded as the gold standard of quantum chemistry; however, its steep $\mathcal{O}(N^7)$ scaling hinders its applicability to very large systems. A quantum algorithm to handle high-order CC theory with a similar exponential speedup in the excitation rank could have profound impacts on its scalability and accuracy.

{\it Note added.}
During the final stages of preparing this manuscript, we became aware of independent work~\cite{wang2026sign} that proposes an alternative scheme to block-encode stabilizing solutions of CAREs, but does not provide applications to quantum chemistry.

{\it Acknowledgments.}
We thank Ryan Babbush, Matthias Degroote, Matthew Hagan, Bill Huggins, Nick Rubin, and Rolando Somma for helpful discussions.
This work was supported by the U.S. Department of Energy, Office of Science, Accelerated Research in Quantum Computing Centers, Quantum Utility through Advanced Computational Quantum Algorithms, grant no. DE-SC0025572.
P.R.-R. received the support of a fellowship from ``la Caixa'' Foundation (ID 100010434), with fellowship code LCF/BQ/EU23/12010085.
A.Z. was supported by the U.S. Department of Energy, Office of Fusion Energy Sciences, ``Foundations for quantum simulation of warm dense matter'' project, and the National Nuclear Security Administration’s Advanced Simulation and Computing Program.

This article has been authored by an employee of National Technology \& Engineering Solutions of Sandia, LLC under Contract No.\ DE-NA0003525 with the U.S. Department of Energy (DOE). The employee owns all right, title and interest in and to the article and is solely responsible for its contents. The United States Government retains and the publisher, by accepting the article for publication, acknowledges that the United States Government retains a non-exclusive, paid-up, irrevocable, world-wide license to publish or reproduce the published form of this article or allow others to do so, for United States Government purposes. The DOE will provide public access to these results of federally sponsored research in accordance with the DOE Public Access Plan \url{https://www.energy.gov/downloads/doe-public-access-plan}.

\bibliography{references}

\clearpage
\onecolumngrid

\appendix
\crefalias{section}{appendix}
\crefalias{subsection}{appendix}
\setcounter{secnumdepth}{2}

\renewcommand{\thesection}{\Alph{section}}
\renewcommand{\thesubsection}{\Roman{subsection}}
\makeatletter
\renewcommand{\p@subsection}{\thesection.}
\makeatother

\makeatletter
\@addtoreset{equation}{section}
\@addtoreset{figure}{section}
\@addtoreset{table}{section}
\makeatother

\renewcommand{\theequation}{\Alph{section}\arabic{equation}}
\renewcommand{\thefigure}{\Alph{section}\arabic{figure}}
\renewcommand{\thetable}{\Alph{section}\arabic{table}}

\begin{center}
{\large\textbf{Supplemental Material for ``Quantum Solvers for Nonlinear Matrix Equations in Quantum Chemistry''}}
\end{center}
\let\addcontentsline\oldaddcontentsline
\tableofcontents
\setlength{\parskip}{0.7em}

\section{CARE background and solution extraction}
\label{app:appA}

In this appendix, we review the necessary background for solving algebraic Riccati equations.

\subsection{The invariant-subspace formulation\label{app:appA-CARE}}
Consider the continuous-time algebraic Riccati equation (CARE),
      \begin{equation}
        XQX - XP  - P^\dagger X- R = 0.
        \label{eq:CARE}
    \end{equation}
The matrices $P,Q,R$ can be defined over $\mathbb{C}^{n\times n}$ or $\mathbb{R}^{n\times n}$, and $Q, R$ are Hermitian. The solution set is defined in $\mathbb{C}^{n\times n}$ or $\mathbb{R}^{n\times n}$, respectively. We are particularly interested in Hermitian or real symmetric solutions of the equation. The CARE arises in continuous-time linear systems, for example in infinite-horizon optimal control and filtering. Its discrete-time analogue is the discrete-time algebraic Riccati equation (DARE). Although we do not treat the DARE directly, the methods developed in this work apply with little modification, after redefining the stability region as the unit disk. 

We briefly introduce the concept of invariant subspaces. A subspace is a subset of $\mathbb{C}^n$ (or $\mathbb{R}^n$) that is also closed under addition and scalar multiplication. Let $A\in\mathbb{C}^{n\times n}$ and $\mathcal{V}\subseteq \mathbb{C}^n$ be an $m$-dimensional subspace. Define $A\mathcal{V} = \{Ax\; | \;x\in \mathcal{V}\}$. We call $\mathcal{V}$ an invariant subspace for $A$ if $A\mathcal{V} \subseteq \mathcal{V}$. If the columns of a full rank matrix $V\in\mathbb{C}^{n\times m}$ span $\mathcal{V}$, then $\mathcal{V}$ is invariant for $A$ iff there exists $\Lambda \in\mathbb{C}^{m\times m}$ such that 
\begin{equation}
    \label{eq:invariant-matrix}
    AV = V\Lambda.
\end{equation}
Additionally, the spectrum of $\Lambda$ is a subset of the spectrum of $A$~\cite{bini_nsar_2012}. Of special relevance for the purposes of this work is the invariant subspace formulation of the CARE, as given in \Cref{eq:linearization-CARE}. Equivalently, a matrix $X\in\mathbb{C}^{n\times n}$ solves the CARE if and only if
\begin{equation}
    \label{eq:graph-CARE}
    \mH\begin{pmatrix}
        \mathbb{I}_n\\
        X
    \end{pmatrix} = \begin{pmatrix}
        \mathbb{I}_n\\
        X
    \end{pmatrix} (P-QX), \qquad
   \mathrm{where} \quad\mH = \begin{pmatrix}
         P & -Q \\
        - R & -P^\dagger
    \end{pmatrix}.
\end{equation}
Thus, the columns of the matrix $\begin{psmallmatrix} {\mathbb{I}} \\X \end{psmallmatrix}$ span an invariant subspace of $\mH$, and the eigenvalues of $P-QX$ are a subset of the eigenvalues of $\mH$, as we described in \Cref{eq:invariant-matrix}. The subspace spanned by $\begin{psmallmatrix} \mathbb{I} \\X \end{psmallmatrix}$ is called a graph subspace, and there is a one-to-one correspondence between the solutions of the CARE and the $n$-dimensional graph invariant subspaces of $\mH$~\cite{bini_nsar_2012}. 

We are often interested in the stabilizing and antistabilizing solutions. A Hermitian solution $X$ of the CARE is called stabilizing if $P-QX$ is stable ($\sigma(P-QX)\subseteq\mathbb{C}_<$), and antistabilizing if $P-QX$ is antistable ($\sigma(P-QX)\subseteq\mathbb{C}_>$). If these solutions exist, they are associated with an $n$-dimensional stable/antistable graph invariant subspace of $\mH$, which corresponds to the $n$ eigenvalues with $\Re(\lambda)<0$ and the $n$ eigenvalues with $\Re(\lambda)>0$, respectively. However, this correspondence does not extend to the eigenvalues associated with the invariant subspace. Repeated eigenvalues of $\mH$ may give rise to different invariant subspaces and consequently to different CARE solutions.

The matrix $\mH\in\mathbb{C}^{2n\times 2n}$ is Hamiltonian in the sense that $\mathcal{J}\mH$ is Hermitian, where $\mathcal{J}$ is the symplectic form $\mathcal{J}=\begin{psmallmatrix}0 & {\mathbb{I}_n}\\ -{\mathbb{I}_n} & 0\end{psmallmatrix}$. Because $\mH$ is Hamiltonian, its spectrum is symmetric with respect to the imaginary axis: non-imaginary eigenvalues form pairs $(\lambda, -\lambda^*)$. Additionally, when $\mH$ has no purely imaginary eigenvalues, the spectrum splits cleanly into $n$ eigenvalues in the left half-plane $\mathbb{C}_{<}$ and $n$ in the right half-plane $\mathbb{C}_{>}$~\cite{bini_nsar_2012}. The separation between these two sets can be quantified by the spectral gap
\begin{equation}
\delta = \min_{\lambda\in\sigma(\mH)}|{\Re(\lambda)}|>0. 
\end{equation}
As we will see, this spectral separation is the basis for the matrix sign function method. Furthermore, if $\sigma(\mH)\cap i\mathbb{R}=\emptyset$ and a stabilizing solution exists, then $X$ is unique and Hermitian, by Theorem 2.17 in~\cite{bini_nsar_2012}. There is a similar result in the antistabilizing case. Therefore, solving the CARE is reduced to finding the correct invariant subspace of $\mH$. Under our assumptions, this task is achievable for the stable and antistable solutions via the matrix sign function method.

Even though we will not discuss them in detail, there are several control-theoretic assumptions which ensure that $\mH$ has no imaginary eigenvalues~\cite{bittanti_re_1991}; here we present two of them. A pair $(P,Q)$ is stabilizable iff there exists $K$ such that $P-QK$ is stable, and $(P,R)$ is detectable iff $(P^\dagger, R^\dagger)$ is stabilizable. Assuming that $Q\succeq 0, R\succeq 0$, we have two main results: (1) The CARE has a unique positive semidefinite (Hermitian) solution, and this solution is stabilizing, if and only if $(P,Q)$ is stabilizable and $(P,R)$ is detectable. (2) The CARE has a unique negative semidefinite (Hermitian) solution, and this solution is antistabilizing, if and only if $(-P,Q)$ is stabilizable and $(P,-R)$ is detectable. Additionally, under these conditions $\sigma(\mH)\cap i\mathbb{R}=\emptyset$. There are other conditions for existence of stabilizing (equivalently of antistabilizing) solutions, with varying characteristics; for example, we can have a stabilizing solution that is not the unique positive semidefinite solution. For our purposes, it will be enough to assume $\sigma(\mH)\cap i\mathbb{R}=\emptyset$ and that there exists a stabilizing/antistabilizing $X$.

In this context, we note that the prior quantum algorithm for solving Riccati equations by Liu \emph{et al.}~assumes $Q\succ0$, $R\succ0$, and $Q^{-1}P=(Q^{-1}P)^\dagger$. This is a special case of a stabilizing solution:~since $Q$ is invertible, $(P,Q)$ is stabilizable in the sense above. For example, $K=Q^{-1}(P+\mathbb{I})$ gives $P-QK=-\mathbb{I}$. Similarly, $R\succ0$ implies that $(P,R)$ is detectable. Hence the control-theoretic criterion gives a unique positive-semidefinite stabilizing solution and $\sigma(\mH)\cap i\mathbb{R}=\emptyset$. Moreover, Liu \emph{et al.}'s solution in the plus branch is $X_s=Q^{-1}P+G$ with $G\coloneqq Q^{-1/2}(Q^{1/2}SQ^{1/2})^{1/2}Q^{-1/2}\succ0$, so the plus branch satisfies $P-QX_s=-QG$. Since $QG$ is similar to $Q^{1/2}GQ^{1/2}\succ0$, $\sigma(P-QX_s)\subset \mathbb{C}_<$ and $X_s$ is the stabilizing solution.

\subsection{The matrix sign function method and Riesz projectors}

The matrix sign function is a highly versatile tool that can be used for solving matrix equations such as Lyapunov, Sylvester, generalized Sylvester, algebraic Riccati, and generalized algebraic Riccati equations~\cite{benner_ssse_2005, benner_fsgs_2021,roberts_lmrs_1980, gardiner_gmsf_1986}. It can also be used for computing invariant subspaces, for computing the matrix polar decomposition, matrix square roots, and even matrix $p^\mathrm{th}$-roots~\cite{byers_msfm_1997, higham_sims_1997, bini_ampr_2005}.

We introduce the matrix sign function following the exposition in~\cite{higham_fm_2008}. The sign function is defined for complex numbers $z \in \mathbb{C}\setminus i\mathbb{R}$ by:
\begin{equation}
\label{eq:complex-sign}
\sign{z} =
\begin{cases}
 + 1, \quad \Re(z) > 0\\
 -1, \quad \Re(z)<0
\end{cases}.
\end{equation}
Consider the Jordan canonical form of a matrix $A$ with no eigenvalues on the imaginary axis
\begin{equation} 
A = VJV^{-1} = 
V\begin{pmatrix}
J_> & 0\\
0& J_<
\end{pmatrix}V^{-1},
\end{equation}
arranged such that the Jordan blocks are separated into stable and antistable eigenvalues. Using the definition of the function of a matrix via the Jordan canonical form, we have
\begin{equation} 
\label{eq:sign-jordan}
\sign{A} =  V\begin{pmatrix}
\mathbb{I}_p & 0\\
0& -\mathbb{I}_q
\end{pmatrix}V^{-1}.
\end{equation}
The change of basis that separates $A$ into its stable and antistable subspaces also separates $\sign{A}$ into its stable and antistable subspaces. 

We introduce a lemma which will be used in the following discussion:
\begin{lemma}[Lemma 1.7 \cite{bini_nsar_2012}]
    \label{lem:func-matrix}
    Let $\Omega$ be an open subset of $\mathbb{C}$ and $f(z):\Omega\mapsto\mathbb{C}$ an analytic function. For $A\in\mathbb{C}^{n\times n}$ with eigenvalues $\lambda_1,\ldots,\lambda_n\in\Omega$, the eigenvalues of $f(A)$ are equal to $f(\lambda_i)$. Moreover, if $A$ is such that $AV = V\Lambda$, where $V\in\mathbb{C}^{n\times m}$ and $\Lambda\in\mathbb{C}^{m\times m}$, then
    \[
    f(A)V = Vf(\Lambda).
    \]
\end{lemma}
Therefore, applying a function such as $\sign{\cdot}$ preserves the invariant subspaces. We now show how to apply the matrix sign function method to the CARE.

\begin{theorem}[Matrix sign function for the stabilizing solution of the CARE, Theorem 3.9~\cite{bini_nsar_2012}]
    \label{thm:sign-method-care}
    Assume there exists a stabilizing solution $X_s$ of the CARE \Cref{eq:CARE}, and let $\mH$ be the Hamiltonian matrix of the CARE. Write $\sign{\mH} + \mathbb{I} = \begin{pmatrix}
            S_1 & S_2
        \end{pmatrix}$ where $S_1,S_2\in \mathbb{C}^{2n\times n}$. Then $X_s$ is the unique solution of the overdetermined system
    \begin{equation}
    \label{eq:sign-system-stab}
        S_2X_s = -S_1.
    \end{equation}
\end{theorem}
\begin{proof}
    By the assumption that there exists a stabilizing solution $X_s$, $\mH$ has no imaginary eigenvalues, since its spectrum is symmetric with respect to the imaginary axis and $X_s$ is associated with $n$ eigenvalues with $\Re(\lambda)<0$. As $X_s$ solves the CARE, there is an invariant-subspace formulation by \Cref{eq:graph-CARE}. By \Cref{lem:func-matrix}, applying $f(z)$ preserves the invariant subspace,
    \[
    f(\mH)\begin{pmatrix} \mathbb{I}_n\\ X_s\end{pmatrix}
    =
    \begin{pmatrix} \mathbb{I}_n\\ X_s\end{pmatrix} f(P-QX_s).
    \]
Choose $f(z)=\sign{z}+1$. Since $X_s$ is stabilizing, $\sigma(P-QX_s)\subset\mathbb{C}_<$, and $\sign{P-QX_s}=-\mathbb{I}_n$, which implies $f(P-QX_s)=0$. Hence the right side of the equation is zero:
\[
(\sign{\mH}+\mathbb{I})\begin{pmatrix} \mathbb{I}_n\\ X_s\end{pmatrix}=0 \rightarrow \begin{pmatrix}
    S_1 & S_2
\end{pmatrix}\begin{pmatrix} \mathbb{I}_n\\ X_s\end{pmatrix} = \begin{pmatrix} S_1 + S_2X_s\end{pmatrix} = 0.
\]
We obtain $S_1 + S_2 X_s = 0$, which is \Cref{eq:sign-system-stab}. On the other hand, compute
\[
(\sign{\mH}+\mathbb{I})\begin{pmatrix} \mathbb{I}_n & 0\\ X_s & \mathbb{I}_n\end{pmatrix}
= \begin{pmatrix} 0 & S_2\end{pmatrix}.
\]
Since the left-hand side has rank $n$ ($\sign{\mH}+\mathbb{I}$ has rank equal to the dimension of the antistable invariant subspace, $n$), $S_2$ has full column rank and the solution to the system in \Cref{eq:sign-system-stab} is unique. Note that this was expected due to the uniqueness of the stable and antistable spaces for $\mH$ with no imaginary eigenvalues, as discussed above.
\end{proof}
Equivalently, the antistabilizing solution $X_a$ can be obtained under similar assumptions by computing $\sign{\mH} - \mathbb{I} = \begin{pmatrix} S_1 & S_2 \end{pmatrix}$ and solving the overdetermined system $S_2X_a = -S_1$. Even though the assumption that $\mH$ has no purely imaginary eigenvalues is fundamental for the existence of stabilizing solutions and to apply the matrix sign function, there are CARE instances in which this will not be true. In some cases, it is possible to apply a shifting procedure to move selected eigenvalues away from the imaginary axis and to compute other non-stabilizing solutions~\cite{bini_nsar_2012}. The shift has to preserve the relevant invariant subspace structure of the CARE. Shifting can also improve convergence by increasing the distance of the spectrum to the imaginary axis. However, this strategy requires knowledge of the eigenvalues and their respective eigenvectors, and large shifts may increase the matrix norm.

For the quantum implementation it is convenient to work with Riesz projectors. Even though both viewpoints are equivalent in the split by the imaginary axis relevant to stabilizing solutions, the Riesz projector brings more freedom for selecting the invariant subspace, yielding other potential applications.

We call the operator $R(z;A) = (z\mathbb{I}-A)^{-1}$ the resolvent. The resolvent set $\varrho(A)$ is the set of points for which $z\mathbb{I}-A$ is invertible (i.e., the resolvent exists). $R(z;A)$ is holomorphic on the resolvent set~\cite{gohberg_rpfc_1990, campbell_laca_2013}. The spectrum of $A$ is defined as the complement of the resolvent set, $\sigma(A) = \mathbb{C}\setminus \varrho(A)$, which coincides with the set of eigenvalues $\sigma(A)$ for matrices. Consider a subset $\tau\subset \sigma(A)$ of the spectrum of $A$. If $\tau$ is enclosed by a contour $\Gamma_\tau\subset\varrho(A)$, which has $\sigma(A)\setminus\tau$ in its exterior, then
\begin{equation}
    \Pi_\tau = \frac{1}{2\pi i}\oint_{\Gamma_\tau} (z\mathbb{I}_n-A)^{-1}dz
\end{equation}
is a projector onto the spectral subspace associated with $\tau$, called a Riesz projector. Cauchy's theorem implies that $\Pi_\tau$ is independent of the specific contour, as long as it separates $\tau$ and the rest of $\sigma(A)$. Let $\Gamma_s \subset\varrho(\mH)$ and $\Gamma_a \subset\varrho(\mH)$ be positively oriented closed contours enclosing $\Sigma_s\subset\mathbb{C}_<$ and $\Sigma_a\subset\mathbb{C}_>$ respectively, and excluding the complementary spectral set. Then, we can define the associated Riesz projectors,
\begin{equation}
\label{eq:riesz-halfplane}
 \Pi_s=\frac{1}{2\pi i}\oint_{\Gamma_s}(z\mathbb{I}-\mH)^{-1}dz, \qquad
   \Pi_a=\frac{1}{2\pi i}\oint_{\Gamma_a}(z\mathbb{I}-\mH)^{-1}dz,
\end{equation}
which are projectors on the stable and antistable subspaces of $\mH$. Importantly, when $\mH$ has no purely imaginary eigenvalues, there is a key relationship between matrix sign function and Riesz projectors~\cite{roberts_lmrs_1980, byers_msfm_1997, bai_msfc_1998}: 
    \begin{equation}
    \label{eq:sign-riesz}
     \sign{\mH}=\Pi_a-\Pi_s
     =2\Pi_a-\mathbb{I}, \qquad \Pi_s=\frac{1}{2}\left(\mathbb{I}-\sign{\mH}\right), \qquad
       \Pi_a=\frac{1}{2}\left(\mathbb{I}+\sign{\mH}\right).
\end{equation}
This equivalence allows us to rewrite the matrix sign function in terms of spectral projectors, and the Riesz formulation provides a precise way of computing these projectors (we could choose others, though, such as computing $\sign{\mH}$ iteratively). Now we establish the correspondence between the matrix sign function method for the CARE and the Riesz projector formulation, which is straightforward.

\begin{corollary}[Riesz projectors for solving the CARE]\label{cor:riesz-care-solution}
   Let $\mH$ be the Hamiltonian matrix of the CARE, with $\sigma(\mH)\cap i\mathbb{R}=\emptyset$. Let $\Pi_s$, $\Pi_a$ be the Riesz projectors on the stable and antistable subspaces of $\mH$.

   If there exists a stabilizing solution $X_s$, it is the unique solution of the overdetermined system:
   \begin{equation}
       \Pi_a = (\Pi_{a,1} \quad \Pi_{a,2}),\qquad \Pi_{a,2}X_s = -\Pi_{a,1}.
   \end{equation}

   If there exists an antistabilizing solution $X_a$, it is the unique solution of the overdetermined system:
   \begin{equation}
       \Pi_s = (\Pi_{s,1} \quad \Pi_{s,2}), \qquad \Pi_{s,2}X_a = -\Pi_{s,1}.
   \end{equation}
\end{corollary}

\begin{proof}
Since $\sigma(\mH)\cap i\mathbb{R}=\emptyset$ by hypothesis, the projectors $\Pi_a$ and $\Pi_s$ in \Cref{eq:riesz-halfplane} are the desired half-plane projectors. Start from the solutions of the CARE via the matrix sign function, $S_2X_s = -S_1$ or $S_2X_a = -S_1$, where $\sign{\mH} \pm \mathbb{I} = \begin{pmatrix}S_1 & S_2\end{pmatrix}$, with $+$ sign for $X_s$ and $-$ sign for $X_a$. As $\sign{\mH} + \mathbb{I}\propto\Pi_a$ and $\sign{\mH} - \mathbb{I}\propto\Pi_s$, the relationships between Riesz projectors and the matrix sign function in \Cref{eq:sign-riesz} yield the desired equivalences.
\end{proof}
Interestingly, the antistable projector is used to obtain the stable solution and the stable projector to obtain the antistable solution. The reason for this is that the desired invariant subspace is defined by $\mathrm{span}\begin{psmallmatrix} {\mathbb{I}_n} \\ X\end{psmallmatrix}$ (\Cref{eq:graph-CARE}), while the complementary projector has this subspace as its kernel. Therefore, applying the projector on the complementary subspace gives the equation $\Pi\begin{psmallmatrix} {\mathbb{I}_n} \\ X\end{psmallmatrix} = 0$, and by writing $\Pi = \begin{pmatrix} \Pi_{1} & \Pi_{2}\end{pmatrix}$ we can find $X$.

\section{Contour quadrature for Riesz projectors}\label{app:appB}

The trapezoid rule is exponentially convergent for approximating Cauchy integrals over $2\pi$-periodic contours when the integrand is analytic on the contour and in a strip of nonzero width around it~\cite{trefethen_ectr_2014}. In the context of quantum computing, Takahira \emph{et al.}~derive a bound for trapezoid approximations of $f(A)$ on offset circles and Jiang \emph{et al.}~do it for arbitrary contours~\cite{takahira_qabb_2022, jiang_cibq_2026}. Nonetheless, we are interested in the case of spectral projectors, where the contour encloses only a subset of the eigenvalues. For this reason, we cannot directly use their techniques for bounding the trapezoid approximation, and we must develop new proofs.

Assume that the contour $\Gamma_a\subset \varrho(\mH)$ enclosing $\Sigma_a$ and excluding $\Sigma_s$ is parametrized by a $2\pi$-periodic, $C^1$ curve $\gamma(\theta)$. For this parametrization, the Riesz projector becomes:
\begin{equation}
\label{eq:riesz-periodic}
\Pi_a=\frac{1}{2\pi i} \int_0^{2\pi} \gamma'(\theta) \left(\gamma(\theta)\mathbb{I}-\mH\right)^{-1} d\theta = \frac{1}{2\pi}\int_{0}^{2\pi}h(\theta)\,d\theta,
\end{equation}
where we defined $h(\theta) = \frac{1}{i} \gamma'(\theta) \left(\gamma(\theta)\mathbb{I}-\mH\right)^{-1}$. The trapezoid rule for $\Pi_a$ over $M$ equally spaced points $\theta_k = \frac{2\pi}{M}k$ is:
\begin{equation}
    \label{eq:trap-periodic}
    \Pi_{a,M} = \frac{1}{M}\sum_{k=0}^{M-1}h(\theta_k).
\end{equation}

Further, for the trapezoid rule over $2\pi$-periodic contours to work we require a region of analyticity around the contour~\cite{trefethen_ectr_2014}. Equivalently, the condition is that $\gamma$ extends holomorphically to the strip
\[
S_{\chi_-,\chi_+}=\{\theta\in\mathbb{C}:-\chi_-<\mathrm{Im}(\theta)<\chi_+\},
\]
with $\gamma(\theta)\in\varrho(\mH)$ for all $\theta\in S_{\chi_-,\chi_+}$. We now proceed to show the exponential convergence of the trapezoid rule for Riesz projectors over this general periodic contour.
\begin{proposition}[Exponential convergence of the trapezoid rule on a periodic contour]
    \label{prop:error-trapezoid-periodic}
    Assume the contour $\Gamma_a\subset\varrho(\mH)$ is parametrized by a $C^1$, positively oriented $2\pi$-periodic curve $\gamma:\mathbb{R}\mapsto\mathbb{C}$ that extends holomorphically to the strip $S_{\chi_-,\chi_+}$, and that $\gamma(\theta)\in\varrho(\mH)$ for all $\theta\in S_{\chi_-,\chi_+}$. Let $0<\eta_-<\chi_-$, $0<\eta_+<\chi_+$. Define the Riesz projector $\Pi_a$ and its trapezoid rule approximation $\Pi_{a,M}$ as in \Cref{eq:riesz-periodic,eq:trap-periodic}. Then,
    \begin{equation}
    \label{eq:error-trapezoid-periodic}
    \|\Pi_a-\Pi_{a,M}\| \leq \gamma_-\frac{1}{e^{\eta_- M}-1}
    + \gamma_+\frac{1}{e^{\eta_+ M}-1},
    \end{equation}
    where $\gamma_\pm=\max_{\varphi\in[0,2\pi]}\|\gamma'(\varphi\pm i\eta_\pm)(\gamma(\varphi\pm i\eta_\pm)\mathbb{I}-\mH)^{-1}\|$. In particular, to guarantee $\|\Pi_a-\Pi_{a,M}\|  \leq \varepsilon$, it suffices to take
    \begin{equation}
    \label{eq:number-quadrature-periodic}
    M\geq \max\left( \frac{\log\left(\frac{2\gamma_-}{\varepsilon}+1\right)}{\eta_-},
    \frac{\log\left(\frac{2\gamma_+}{\varepsilon}+1\right)}{\eta_+}
    \right).
    \end{equation}
\end{proposition}
\begin{proof}

For unit vectors $x,y\in\mathbb{C}^{2n}$, define the scalar function $s_{xy}(\theta) = x^\dagger h(\theta)y$. Since $\gamma$ is holomorphic in $S_{\chi_-,\chi_+}$, $\gamma'$ is holomorphic there too. Since $\gamma(\theta)\in\varrho(\mH)$ for $\theta\in S_{\chi_-,\chi_+}$, the resolvent $(\gamma(\theta)\mathbb{I}-\mH)^{-1}$ is holomorphic in the strip as well. The resolvent being holomorphic implies that it is entrywise holomorphic in this region~\cite{kato_ptlo_1995,campbell_laca_2013}. Therefore $s_{xy}$ is $2\pi$-periodic and holomorphic on $S_{\chi_-,\chi_+}$, and we can directly follow the trapezoid rule for scalar functions on periodic intervals~\cite{trefethen_ectr_2014}. $s_{xy}$ has a uniformly and absolutely convergent Fourier series on the substrips of $S_{\chi_-,\chi_+}$:
    \[
    s_{xy}(z) = \sum_{n=-\infty}^\infty c_n(x,y)e^{inz}, \qquad
    c_n(x,y) = \frac{1}{2\pi}\int^{2\pi}_0 s_{xy}(\varphi)e^{-in\varphi}d\varphi.
    \]
   Using \Cref{eq:riesz-periodic,eq:trap-periodic} we obtain
   \[
    x^\dagger \Pi_a y
    = \frac{1}{2\pi}\int_0^{2\pi}
    s_{xy}(\theta)d\theta, \qquad x^\dagger \Pi_{a,M} y
    = \frac{1}{M}\sum_{k=0}^{M-1}
    s_{xy}(\theta_k) = \sum_{n=-\infty}^{\infty} c_n(x,y)\left(\frac{1}{M}\sum_{k=0}^{M-1}e^{in\theta_k}\right).
    \]
    It can be easily verified that $x^\dagger \Pi_a y = c_0(x,y)$. Additionally, the sum $\frac{1}{M}\sum_{k=0}^{M-1} e^{in\theta_k}$ is $1$ when $n$ is a multiple of $M$, and $0$ otherwise. Therefore, $x^\dagger \Pi_{a,M} y =\sum_{q=-\infty}^{\infty}c_{qM}(x,y)$, and
    \[
    x^\dagger(\Pi_a-\Pi_{a,M})y
 = -\sum_{q\in\mathbb{Z}\setminus\{0\}}c_{qM}(x,y).
    \]
    Taking the absolute value,
    \begin{align*}
        |x^\dagger(\Pi_a - \Pi_{a,M})y| \leq \sum_{q=1}^\infty |c_{-qM}(x,y)|+\sum_{q=1}^\infty |c_{qM}(x,y)|.
    \end{align*}
    It remains to bound the Fourier coefficients. For $n>0$, consider a rectangle with vertices $0,2\pi, 2\pi-i\eta_-, -i\eta_-$. Since $s_{xy}(\varphi)e^{-in\varphi}$ is $2\pi$ periodic and holomorphic, we can apply Cauchy's theorem with the rectangle as a contour:
    \[
    \int_0^{2\pi}s_{xy}(\varphi)e^{-in\varphi}d\varphi - \int_0^{2\pi}s_{xy}(\varphi-i\eta_-)e^{-in(\varphi-i\eta_-)}d\varphi + \mathrm{sides} = 0.
    \]
    The vertical sides cancel due to periodicity. Therefore,
    \[
    c_n(x,y) = \frac{1}{2\pi}\int_0^{2\pi}s_{xy}(\varphi-i\eta_-)e^{-in(\varphi-i\eta_-)}d\varphi.
    \]
    We achieved an integral that yields the same coefficient but on a downward-shifted contour by $\eta_-$ (instead of $\chi_-$, to have a safety margin within the holomorphy region). Taking the absolute value and the triangle inequality for integrals,
    \[
    |c_n(x,y)| \leq e^{-n\eta_-}\frac{1}{2\pi}\int_0^{2\pi}|s_{xy}(\varphi-i\eta_-)|d\varphi\leq e^{-n\eta_-}\max_{\varphi\in[0,2\pi]}|s_{xy}(\varphi-i\eta_-)|.
    \]
    The downward shifting allowed us to obtain an exponential decay of the coefficients in terms of the lower strip width $\eta_-$. For $n<0$, we will need an upward shift, which can go up to $i\eta_+$. Similarly,
    \[|c_{-n}(x,y)| \leq e^{-n\eta_+}\frac{1}{2\pi}\int_0^{2\pi}|s_{xy}(\varphi+i\eta_+)|d\varphi\leq e^{-n\eta_+}\max_{\varphi\in[0,2\pi]}|s_{xy}(\varphi+i\eta_+)|.\]
    Now we need to bound the scalar function $|s_{xy}(\theta)|$. The Cauchy-Schwarz inequality with $\|x\|_2 = \|y\|_2 = 1$ gives:
    \begin{align*}
        |s_{xy}(\theta)| = |x^\dagger h(\theta)y| \leq \|x^\dagger\|_2\|h(\theta)y\|_2 \leq \| h(\theta) \| = |\gamma'(\theta)|\|(\gamma(\theta)\mathbb{I}-\mH)^{-1}\|.
    \end{align*}
    In the second inequality we used $\|Mx\|_2\leq \|M\|\|x\|_2$, which is true since $\|\cdot\|$ is the spectral norm, induced by the Euclidean vector norm $\|\cdot\|_2$~\cite{horn_ma_2012}. Hence,
    \[
    |c_n(x,y)|\leq \gamma_- e^{-n\eta_-},\qquad
     |c_{-n}(x,y)|\leq \gamma_+ e^{-n\eta_+}.
    \]
    Substituting these bounds into the error inequality,
    \begin{align*}
        |x^\dagger(\Pi_a - \Pi_{a,M})y| &\leq \sum_{q=1}^\infty |c_{-qM}(x,y)|+\sum_{q=1}^\infty |c_{qM}(x,y)| \leq \sum_{q=1}^\infty \gamma_+ e^{-qM\eta_+} +\sum_{q=1}^\infty  \gamma_- e^{-qM\eta_-} = \gamma_-\frac{1}{e^{\eta_- M}-1}
    + \gamma_+\frac{1}{e^{\eta_+ M}-1}.
    \end{align*}
    In the last step, we summed the geometric series. Finally, we now use the representation of the spectral norm,
    \[
    \|A\| = \max_{\|x\|_2=1,\|y\|_2=1}|x^\dagger A y|,
    \] to obtain the desired result \Cref{eq:error-trapezoid-periodic}. To estimate the required number of quadrature points for a given $\varepsilon$, it is enough to require that each term on the right-hand side of \Cref{eq:error-trapezoid-periodic} is bounded by $\varepsilon/2$:
    \[
    \gamma_\pm \frac{1}{e^{\eta_\pm M}-1}\le \frac{\varepsilon}{2} \rightarrow M \geq \frac{\ln\left(\frac{2\gamma_\pm}{\varepsilon}+ 1\right)}{\eta_\pm}.
    \]
    Thus, it suffices to take the maximum between both cases.
\end{proof}

We see that $M=\mathcal{O}\mathopen{}\left(\ln\left(\frac{\gamma_{\pm}}{\varepsilon}\right)/\eta_\pm\right)$, so the convergence is exponential in $\varepsilon$, in agreement with related bounds by Trefethen and Weideman~\cite{trefethen_ectr_2014}. The number of nodes is also dependent on the resolvent norm through $\gamma_\pm$, and the analytic strip width through $\eta_\pm$. Therefore, it is expected that for thin analyticity regions, where the strip width is small and the resolvent norms can grow uncontrollably, the number of quadrature points will be much larger. Conformal maps can be used to make the analyticity region larger, enabling more efficient quadratures with fewer points $M$~\cite{hale_calr_2008, trefethen_ectr_2014}. Another option that simplifies contour choice would be to apply a Möbius/Cayley transform $\mH_\circ = (\mH - \mathbb{I})(\mH+\mathbb{I})^{-1}$, which maps eigenvalues in $\Sigma_a$ to the unit disk. These maps are the underlying mechanism for some transformations between the CARE and the DARE, since the stability region of the DARE is the unit disk~\cite{xu_tdtc_2007}. However, we do not pursue this option as the block-encoding of $\mH_\circ$ itself would require matrix inversion, introducing multiplicative overhead on top of the inversions already required to block-encode the resolvent~$(z\mathbb{I} - \mH_\circ)^{-1}$.

\subsection{Smoothed semicircle contour}
Now we specialize our results to a specific contour that encloses the eigenvalues in the right half-plane and excludes the eigenvalues in the left half-plane. One simple possibility is an offset circle, chosen large enough to achieve the desired spectral separation. However, its inconvenient perimeter growth in edge cases increases the cost of the block-encoding.

Let $r_\sigma(\mH)=\max\{|\lambda|:\lambda\in\sigma(\mH)\}$ denote the spectral radius of $\mH$. Since $r_\sigma(\mH) \leq \|\mH\| \leq \alpha_\mH$, a circle of radius $R \geq \alpha_\mH$ centered at the origin will enclose $\Sigma_s$ and $\Sigma_a$. Taking $R=\alpha_\mH$, we can build a D-shaped contour (a semicircle) that only encloses the eigenvalues $\Sigma_a$ and whose length grows only as $\mathcal{O}(\alpha_\mH)$, even when the distance to the imaginary axis $\delta$ is small. Centered at $z_0$, the semicircle in the right half-plane can be parametrized by:
\[
\gamma_{\mathrm{D}}(\theta) = z_0 + R\max(0,\cos(\theta)) + iR\sin(\theta).
\]
However, the maximum function is not smooth. There exist multiple smoothings; we choose the following:
\begin{equation}
    \label{eq:smooth-semicircle-app}
    \gamma_{\mathrm{SD}}(\theta) = z_0 + \frac{1}{2}\left(R\cos(\theta) + \sqrt{\omega^2+R^2\cos^2(\theta)} \right)+ iR\sin(\theta), \qquad \theta\in\mathbb{R}.
\end{equation}
Here, $\omega$ is a smoothing parameter, and when $\omega = 0$, we recover $\gamma_{\mathrm{D}}$. Now we study under which parameters the smoothed semicircle encloses the region of interest.
\begin{proposition}[Spectral separation by the smoothed semicircle contour]
\label{prop:smooth-semicircle}
Let $\|\mH\|\leq \alpha_\mH$ and assume that the spectrum of $\mH$ has no purely imaginary eigenvalues, and that $\delta=\min_{\lambda\in\sigma(\mH)} |{\Re\lambda}|>0$. Define the $2\pi$-periodic, $C^\infty$ curve $\gamma_{\mathrm{SD}}$ as in \Cref{eq:smooth-semicircle-app}, and choose curve parameters satisfying
\[
R\geq \alpha_{\mH},\qquad z_0\geq 0,\qquad \omega>0,\qquad  0<z_0 + \frac{\omega}{2}<\delta.
\]
Then the contour $\Gamma_{\mathrm{SD}}$ given by $\gamma_{\mathrm{SD}}(\theta), \,\theta\in[0,2\pi]$ is simple, closed and positively oriented. Moreover, $\Sigma_a$ lies in its interior, $\Sigma_s$ lies in its exterior, and $ \Gamma_{\mathrm{SD}}\subset \varrho(\mH)$. Hence $\Gamma_{\mathrm{SD}}$ is a valid choice for the contour $\Gamma_a$ of the Riesz projector $\Pi_a$. Moreover,
\begin{equation}
\label{eq:smooth-SD-derivative}
|\gamma'_{\mathrm{SD}}(\theta)|\leq R \qquad
\forall\,\theta\in\mathbb{R},
\end{equation}
and the length of the contour satisfies $\mathcal{L}_{\mathrm{SD}} \leq 2\pi R$. In particular, taking $R=\alpha_{\mH}$ gives $\mathcal{L}_{\mathrm{SD}} = \mathcal{O}(\alpha_{\mH})$.
\end{proposition}
\begin{proof}
Since $\omega>0$, the square root in \Cref{eq:smooth-semicircle-app} is never zero for $\theta\in\mathbb{R}$. Then, $\gamma_{\mathrm{SD}}$ is $C^{\infty}$ and $2\pi$-periodic. To see that the curve does not intersect itself (the contour is simple), start by splitting $\gamma_{\mathrm{SD}}(\theta)$ in its real and imaginary parts, $\gamma_{\mathrm{SD}}(\theta) = x(\theta) + iy(\theta)$. Since $y(\theta) = R\sin(\theta)$, we have $y(\theta)\in[-R,R]$.

For each height $h\in(-R,R)$, there are two values of $\theta\in[0,2\pi]$ such that $y(\theta) = R\sin\theta = h$. For these values, $R\cos(\theta) =\pm\sqrt{R^2 - h^2}$. Thus, the real coordinate $x(\theta)$ can take two values at height $h$,
\[
x_-(h) = z_0 + \frac{1}{2}\left(-\sqrt{R^2-h^2 } + \sqrt{\omega^2+ R^2-h^2}\right) \qquad x_+(h) = z_0 + \frac{1}{2}\left(\sqrt{R^2-h^2 } + \sqrt{\omega^2+ R^2-h^2}\right).
\]

Clearly, $x_-(h)<x_+(h)$ for all $h\in(-R,R)$. At $h = \pm R$, we have
\[
x_-(\pm R) = x_+(\pm R) = z_0 +\frac{\omega}{2},
\]
and the branches meet at the top and the bottom. Thus, the curve has no self-intersections, and since
$\gamma_{\mathrm{SD}}(0) =\gamma_{\mathrm{SD}}(2\pi)=z_0 +\frac{1}{2}\left( R+\sqrt{\omega^2+R^2}\right)$, $\Gamma_\mathrm{SD}$ is simple and closed. It can be verified that on the right branch $x_+(h)$, which has $\cos\theta>0$, the imaginary part $y(\theta)$ is increasing with $\theta$, while on the left branch $x_-(h)$, $y(\theta)$ is decreasing. Therefore, the region enclosed by the curve remains to the left when traveling on it, and the curve is positively oriented.

We now prove spectral separation under a suitable parameter choice. Since every eigenvalue satisfies $|\lambda|\leq \|\mH\|\leq \alpha_\mH \leq R$, we write $\lambda = c+id$ with $|d|<R$, because $|\Re\lambda|=|c|\geq \delta>0$. First consider the antistable eigenvalues $\lambda\in\Sigma_a$, so that $c\geq \delta$. We can now use the contour branches defined previously, and compare their values at the eigenvalue height $d$, i.e., $x_-(d), x_+(d)$, with the corresponding real part of $\lambda$, which should lie between these branches. For $\omega>0$, the left boundary satisfies
\[
x_-(d) = z_0 + \frac{1}{2}\left(-\sqrt{R^2-d^2} + \sqrt{\omega^2+ R^2-d^2}\right)\leq z_0 + \frac{1}{2}\left(-\sqrt{R^2-d^2} + \omega + \sqrt{R^2-d^2}\right)\leq z_0 + \frac{\omega}{2}.
\]
Under the parameter assumptions, $z_0 +\omega/2<\delta\leq c$, which ensures that $\lambda$ is to the right of $x_-(d)$. For the right branch, given $z_0\geq 0, \omega>0$ we immediately have
\[
x_+(d) = z_0 + \frac{1}{2}\left(\sqrt{R^2-d^2 } + \sqrt{\omega^2+ R^2-d^2}\right)>\sqrt{R^2-d^2}.
\]
Now note that $|\lambda| = \sqrt{c^2+ d^2}\leq R$, so $c\leq \sqrt{R^2-d^2}$, and finally $c<x_+(d)$. In conclusion, combining the two branch inequalities we have
\[
x_-(d)< \delta \leq c < x_+(d),
\]
under the stated parameter assumptions. Since the imaginary part satisfies $|d|<R$, every eigenvalue in $\Sigma_a$ lies inside the contour $\Gamma_\mathrm{SD}$.

For stable eigenvalues $\lambda = c+id\in\Sigma_s$, we have $c\leq -\delta<0$. Since $z_0\geq 0, \omega>0$, we have $x_-(d)>0$ and the contour lies in the right half-plane, so every eigenvalue in $\Sigma_s$ lies strictly outside $\Gamma_\mathrm{SD}$. As the inclusions and exclusions are strict, the contour does not intersect the spectrum, so the resolvent is defined along the contour, $\Gamma_{\mathrm{SD}}\subset \varrho(\mH)$.

It remains to prove the length bound. The length of the contour in the complex plane is given by~\cite{conway_focv_1994}:
\[
\mathcal{L}_{\mathrm{SD}} = \int_0^{2\pi} |\gamma'_{\mathrm{SD}}(\theta)|d\theta.
\]
First we differentiate \Cref{eq:smooth-semicircle-app}:
\[
\gamma'_{\mathrm{SD}}(\theta) = -\frac{R\sin\theta}{2}
\left(1+\frac{R\cos\theta}{\sqrt{\omega^2+R^2\cos^2\theta}}\right)
+iR\cos\theta.
\]
Since
\[
\left|\frac{R\cos(\theta)}{\sqrt{\omega^2+R^2\cos^2\theta}}\right|\le 1,
\]
the term inside parentheses in the derivative has absolute value bounded by $2$. Hence $|\Re(\gamma'_{\mathrm{SD}}(\theta))|\le R|\sin(\theta)|$. Therefore,
\[
|\gamma'_{\mathrm{SD}}(\theta)|^2 = \Re(\gamma'_{\mathrm{SD}}(\theta))^2 + \Im(\gamma'_{\mathrm{SD}}(\theta))^2\leq R^2\sin^2(\theta) + R^2 \cos^2(\theta) = R^2,
\]
and we obtain $|\gamma'_{\mathrm{SD}}(\theta)|\leq R$ for all $\theta\in \mathbb{R}$. Using the contour formula,
\[
\mathcal{L}_{\mathrm{SD}} = \int_0^{2\pi} |\gamma'_{\mathrm{SD}}(\theta)|d\theta \leq  \max_{\theta\in[0,2\pi]} |\gamma'_{\mathrm{SD}}(\theta)| \int_0^{2\pi} d\theta = 2\pi R.
\]
The length of the contour therefore satisfies $\mathcal{L}_{\mathrm{SD}}\leq 2\pi R$. In particular, taking $R=\Theta(\alpha_\mH)$ gives $\mathcal{L}_{\mathrm{SD}}= \mathcal{O}(\alpha_\mH)$.
\end{proof}
Importantly, under the assumptions in the proposition, choosing $R=\Theta(\alpha_\mH)$ gives a smooth semicircle contour of length $\mathcal{O}(\alpha_\mH)$ independently of how close to the imaginary axis the spectrum of $\mH$ is. Note that $R\geq \alpha_\mH$ under our assumptions. Additionally, we have also found $|\gamma'_{\mathrm{SD}}(\theta)|\leq R$, which will be important when estimating the query complexity of the block-encoding. In \Cref{fig:comparison-contours-H}, we compare the smoothed semicircle contour with an offset semicircle, which grows as $\mathcal{O}(\alpha^2/\delta)$ when $\delta$ is small. We observe that the smoothed semicircle has a smaller radius, and this improvement becomes more pronounced for smaller $\delta$.
\begin{figure}[!htbp]
    \centering
\includegraphics[width=0.8\linewidth]{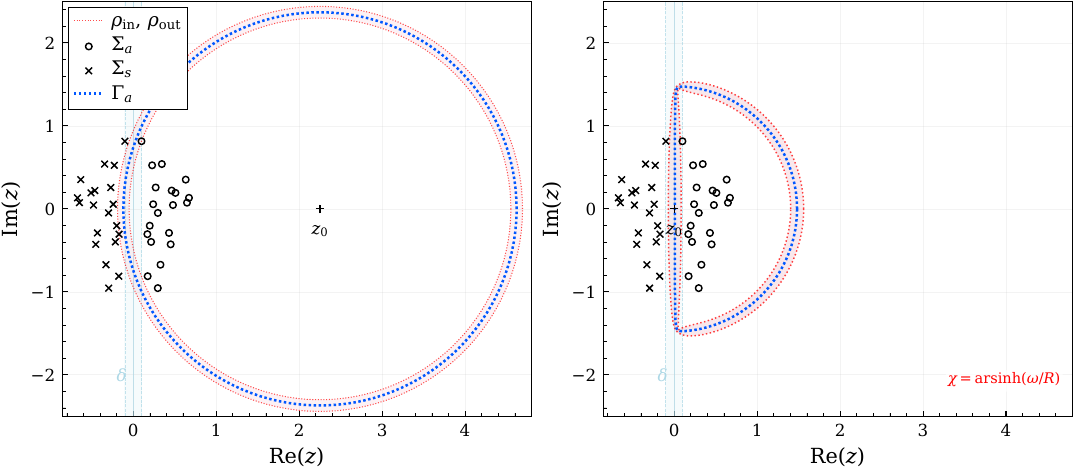}
\caption{Spectrum of a random complex CARE Hamiltonian $\mH$ without imaginary eigenvalues, with the antistable region enclosed by a contour $\Gamma_a$ parametrized by an offset-circle contour $\gamma_\circ$ or by a smoothed semicircle $\gamma_{\mathrm{SD}}$. The light-blue shaded region marks the spectral gap, here $\delta = 0.10$; the light-red shaded region shows the corresponding analyticity region (an annulus on the left, the smoothed-semicircle strip image on the right).}
    \label{fig:comparison-contours-H}
 \end{figure}
The only missing requirement to apply \Cref{prop:error-trapezoid-periodic} is that $\gamma_{\mathrm{SD}}$ extends holomorphically to a strip $S_{\chi_-,\chi_+}$ and that $\gamma_{\mathrm{SD}}(\theta)\in\varrho(\mH)$ for all $\theta\in S_{\chi_-,\chi_+}$. For simplicity we take $\chi = \chi_- = \chi_+$ and find the following result.
\begin{proposition}[Holomorphic extension of the smoothed semicircle]
\label{prop:smooth-semicircle-strip}
Let $\gamma_{\mathrm{SD}}(\theta)$ be the smoothed semicircle under the curve parameters of \Cref{prop:smooth-semicircle}. Define
\begin{equation}
\label{eq:chi-smooth-semicircle}
\chi = \operatorname{arsinh}\left(\frac{\omega}{R}\right)>0.
\end{equation}
Then $\gamma_{\mathrm{SD}}$ admits a holomorphic $2\pi$-periodic extension to the strip $S_{\chi,\chi} = \{z\in\mathbb{C}: |\mathrm{Im}( z)|<\chi\}$ given by
\begin{equation}
\label{eq:holo-extension-semicircle}
\widehat\gamma_{\mathrm{SD}}(z)=z_0+\frac{1}{2}\left(R\cos z+\sqrt{\omega^2+R^2\cos^2 z}\right)+iR\sin z,
\end{equation}
where the square root is taken on the principal branch. Hence, $\widehat\gamma_{\mathrm{SD}}(\theta) = \gamma_{\mathrm{SD}}(\theta),\, \theta\in\mathbb{R}$ and $\widehat\gamma_{\mathrm{SD}}$ is holomorphic on $S_{\chi,\chi}$.
\end{proposition}
\begin{proof}
    The smoothed semicircle contour is
    \[
        \gamma_{\mathrm{SD}}(\theta) = z_0 + \frac{1}{2}\left(R\cos(\theta) + \sqrt{\omega^2+R^2\cos^2(\theta)} \right)+ iR\sin(\theta).
    \]
    The functions $\cos z$ and $\sin z$ are entire, so the only complicated term is the square root. Define the radicand $g(z) ={\omega^2+R^2\cos^2(z)}$, which is entire because $\cos(z)$ is entire. We need to determine where $g(z)$ has its zeros. First,
    \[
    g(z) = 0\implies \cos(z) = \pm i\frac{\omega}{R}.
    \]
    At these points, the square root will not be holomorphic. We now identify the corresponding values of $z$. Taking $z = x+iy$, $\cos(z) = \cos x \cosh y - i \sin x \sinh y$. Thus $g(z) =0$ implies
    \[
    \cos x \cosh y = 0, \qquad \sin x \sinh y = \mp \frac{\omega}{R}.
    \]
    As $\cosh y$ is positive for real $y$, the first equation gives $\cos x = 0$, hence $x = \pi/2 + k\pi,\, k\in\mathbb{Z}$. For these values of $x$, $|{\sin x}| = 1$ and the second equation gives
    \[
    |{\sinh y}| = \frac{\omega}{R} \implies |y|= \operatorname{arsinh}\left(\frac{\omega}{R}\right) = \chi,
    \]
    because $\sinh y$ is odd and increasing for $y$ real. Therefore, the analyticity strip is characterized by the zeros of $g(z)$, which can occur at the boundaries $|\mathrm{Im}(z)|=\chi$. Hence $g(z)\neq 0$ for all $z\in S_{\chi,\chi}$.

    Additionally, we need to be careful with the branch cuts of the square root. We choose the principal square root, so that for $z= r e^{i\phi}$, $\sqrt{z} = \sqrt{r}e^{i\phi/2}$, where $\phi \in(-\pi, \pi)$. The principal square root has its branch cut at $(-\infty,0]$, so it is holomorphic on $\mathbb{C}\setminus(-\infty,0]$~\cite{conway_focv_1994}. If we show that $g(z),\, z\in S_{\chi,\chi}$ has no values in $(-\infty,0]$, we may take the principal square root to define the holomorphic extension.

    We have seen that $g(z)\neq 0$ for $z\in S_{\chi,\chi}$. It remains to check the interval $(-\infty,0)$, in which $\mathrm{Im}(g(z))$ is zero. We obtain
    \[
    \mathrm{Im}(g(z)) = -2R^2 \cos x \cosh y \sin x \sinh y = -R^2 \sin 2x \cosh y \sinh y.
    \]
    Hence, $\mathrm{Im}(g(z)) = 0$ only if $\sin 2x =0$ or $\sinh y =0$. Equivalently, $x = k\pi /2$ or $y = 0$. Now, we study $g(z)$ under these conditions.

    For $y=0$,
    \[g(z) = \omega^2 + R^2 \cos^2(x)>0.
    \]
    For $x = k\pi$,
    \[
    g(z) = \omega^2 +R^2 \cosh^2 y>0.
    \]
    For $x = \pi/2 + k\pi$,
    \[
    g(z) = \omega^2 -R^2 \sinh^2 y>0,
    \]
    since for $z\in S_{\chi,\chi}$ we have $|y|< \chi$, and therefore $|\sinh y| < \omega/R$. In conclusion, whenever $g(z)$ is real on $S_{\chi,\chi}$, it is strictly positive,
    \[
    g(z)\cap (-\infty, 0] =\emptyset, \qquad z\in S_{\chi,\chi}.
    \]
    It follows that $\sqrt{ \omega^2 +R^2 \cos^2(z)}$ is holomorphic on $S_{\chi,\chi}$. On the real axis, $g(\theta)>0$ and the principal square root agrees with the positive square root, so $\widehat\gamma_{\mathrm{SD}}(\theta) = \gamma_{\mathrm{SD}}(\theta), \, \theta\in\mathbb{R}$. Finally, $\widehat\gamma_{\mathrm{SD}}$ is $2\pi$-periodic because $\cos z$ and $\sin z$ are $2\pi$-periodic.
\end{proof}
For the trapezoid-rule proposition, we may take any symmetric substrip
\[
S_{\eta,\eta},\qquad 0<\eta<\chi=\operatorname{arsinh}(\omega/R),
\]
provided that the image of this substrip remains in $\varrho(\mH)$. We now show that such an admissible $\eta$ always exists under a concrete choice of smoothed-semicircle parameters.

\begin{lemma}[Admissible strip width]
\label{lem:eta-scaling}
    Let $\|\mH\|\leq\alpha_\mH$ and $\delta=\min_{\lambda\in\sigma(\mH)}|\Re\lambda|>0$. For the smoothed semicircle with $z_0=0$, $\omega = \delta$ and $R = 2\alpha_\mH$, set $\eta_\star = \chi/16$, for $\chi = \operatorname{arsinh}(\delta/R)$. Then,  $\widehat\gamma_\mathrm{SD}(z)\in\varrho(\mH)$ for
$z\in S_{\eta_\star, \eta_\star}$, and $\eta_\star=\Theta(\delta/\alpha_\mH)$.
\end{lemma}
\begin{proof}
    By \Cref{prop:smooth-semicircle-strip}, $\widehat\gamma_{\mathrm{SD}}$ is holomorphic in the strip $S_{\chi, \chi}$. We first show that the real contour has distance at least $\delta/2$ from $\sigma(\mH)$. Since $|\lambda|\leq \|\mH\|\leq \alpha_\mH = R/2$ for all $\lambda\in \sigma(\mH)$, and $\delta\leq \|\mH\|$, we have $R/2\geq \delta$.
    On the left branch of the contour, where $\cos\theta\leq 0$, write $u = -R\cos\theta\geq 0$, so that
    \[
    0\leq\Re(\gamma_{\mathrm{SD}}(\theta))=\frac{1}{2}\left(\sqrt{\delta^2+u^2}-u\right)\leq \delta/2.
    \]
    Hence stable eigenvalues, with $\Re\lambda\leq-\delta$, are at distance at least $\delta$ from this branch, and antistable eigenvalues, with $\Re\lambda\geq \delta$, are at distance at least $\delta/2$. For the right branch, take $v = R\cos\theta\geq 0$, and since $\Re(\gamma_\mathrm{SD}(\theta))\geq v$, we have:
    \[
        |\gamma_{\mathrm{SD}}(\theta)|^2\geq v^2+R^2\sin^2\theta=R^2 .
    \]
    As a consequence, for every $\lambda\in\sigma(\mH)$ and for the right branch of the contour:
    \[
     |\gamma_{\mathrm{SD}}(\theta)-\lambda| \geq|\gamma_{\mathrm{SD}}(\theta)|-|\lambda|
        \geq R/2
        \geq \delta,
    \]
    Combining the distances for both branches, it follows that
    \[
    \mathrm{dist}(\gamma_{\mathrm{SD}}(\theta),\sigma(\mH))\geq \delta/2, \qquad \theta\in[0,2\pi].
    \]
    Now it remains to obtain the distance between $\sigma(\mH)$ and the analytic extension of the contour $\widehat\gamma_{\mathrm{SD}}$. Take $\rho = \delta/R\leq 1/2$ under the assumed $\delta$ and $R$, and note that $\chi = \operatorname{arsinh}(\rho)$. Write $\cos z = a+ib$, with $a=\cos x\cosh y$ and $b=-\sin x\sinh y$. If $z=x+iy$ with $|y|\leq\chi/2$, then, using the fact that $\sinh$ is strictly increasing on the reals,
    \[
    b^2 = \sin^2(x)\sinh^2(y)
    \leq \sinh^2(\chi/2) =\frac{\sqrt{1+\sinh^2\chi}-1}{2}=\frac{\sqrt{1+\rho^2}-1}{2} \leq \rho^2/2 .
    \]
    From the previous inequality, we can bound:
    \[
        |\rho^2+\cos^2 z|^2-|\cos^2 z|^2
        =\rho^2\bigl(\rho^2+2(a^2-b^2)\bigr)
        \geq 0,
    \]
    which implies that
\[
\left|\frac{R\cos z}{\sqrt{\delta^2+R^2\cos^2 z}}\right|
=\left|\frac{\cos z}{\sqrt{\rho^2+\cos^2 z}}
\right|\leq 1 .
\]
Taking the derivative of \Cref{eq:smooth-semicircle-app},
\[
\widehat\gamma'_{\mathrm{SD}}(z) = -\frac{R\sin z}{2}\left(1+\frac{R\cos z}{\sqrt{\omega^2+R^2\cos^2 z}}\right)
+iR\cos z,
\]
we can bound it in $|\Im z|\leq\chi/2$, using the previous results:
\[
|\widehat\gamma_{\mathrm{SD}}'(z)|\leq R|\sin z|+R|\cos z| \leq 4R.
\]
Now let $|\tau|\leq \eta_\star=\chi/16<\chi/2$. Then, we are in the region $|\Im z|\leq\chi/2$, and hence
\[
|\widehat\gamma_{\mathrm{SD}}(\theta+i\tau)-\gamma_{\mathrm{SD}}(\theta)|\leq 4R|\tau|\leq4R\eta_\star=R\chi/4<R\rho/4 = \delta/4.
\]
Combining this with the real contour gap of $\delta/2$, we conclude that the shifted contour remains at distance at least $\delta/4$ from $\sigma(\mH)$, so it is inside $\varrho(\mH)$. Finally, $\operatorname{arsinh}(x) = \Theta(x)$ for $0<x\leq 1/2$, and $\delta/R\leq 1/2$, $R=2\alpha_\mH$, so $\eta_\star = \Theta(\delta/\alpha_\mH)$.
\end{proof}

\section{Block-encoding Riesz projectors and CARE solutions}
\label{app:appC}
We now introduce the details of our quantum algorithm, which block-encodes the stabilizing solution of the CARE into a quantum circuit. First, we formally define what we refer to as a block-encoding.

\subsection{Block-encoding Riesz projectors}

\begin{definition}[Block-encoding]
\label{def:block-encoding}
Let $A\in\mathbb{C}^{n\times n}$, with $n=2^q, q\in\mathbb{N}$ without loss of generality by padding. Let $a\in\mathbb{N}$, $\varepsilon, \alpha \in \mathbb{R}^+$. The $(q+a)$-qubit unitary $U_A$ is an $(\alpha,a,\varepsilon)$-block-encoding of $A$ if
\begin{equation}
\left\| \alpha(\bra{0}^{\otimes a}\otimes \mathbb{I}_n)U_A(\ket{0}^{\otimes a}\otimes \mathbb{I}_n) - A \right\| \le \varepsilon,
\end{equation}
where $\|\cdot\|$ denotes the spectral norm, i.e., $\|M\|=\sigma_{\max}(M)$. $\alpha\ge 1$ is a sub-normalization factor such that $\|A\|/\alpha\leq 1$. 
\end{definition}
In the main text we suppress the ancilla-qubit count $a$ and refer to $(\alpha,\varepsilon)$-block-encodings for conciseness. This definition is the formal one we use throughout the SM.

Having established a suitable contour, we can now construct a block-encoding of the contour integral by techniques similar to those in prior work~\cite{takahira_qamf_2020, takahira_qabb_2022, jiang_cibq_2026}. We present a method to block-encode the trapezoidal approximation of the Riesz projector $\Pi_a$ on $M$ points of an arbitrary periodic contour
\begin{equation}
\Pi_{a,M} = \frac{1}{M}\sum_{k=0}^{M-1} w_k (z_k\mathbb{I}-\mH)^{-1},\quad
 z_k=\gamma(\theta_k),\quad
 w_k=\frac{1}{i}\gamma'(\theta_k),\quad
 \theta_k=\frac{2\pi k}{M}.
 \label{eq:general-trapezoid}
\end{equation}
The approximation depends on the contour through $z_k$ and $w_k$. We assume access to an $(\alpha_\mH, a_\mH, \varepsilon_\mH)$-block-encoding of $\mH$. To implement $\Pi_{a,M}$, we need to block-encode resolvents $(z_k\mathbb{I}-\mH)^{-1}$. Note that, with access to a block-encoding of $z_k\mathbb{I} -\mH$, this can be done by QSVT matrix inversion. Then, an LCU can be applied to the resulting resolvent block-encodings, yielding a block-encoding of $\Pi_{a,M}$, and hence of $\Pi_a$ (with controlled error). We provide the main LCU and QSVT results in terms of block-encodings, which will be useful as building blocks of our algorithm~\cite{childs_hslc_2012, berry_hsno_2015,childs_qasl_2017,lin_qasc_2026}.

\begin{lemma}[LCU with positive coefficients~\cite{lin_qasc_2026}]
    \label{lem:lcu}
    Let $U_0, \ldots, U_{K-1}$ be $q$-qubit unitaries, let $c_0, \ldots, c_{K-1}\geq0$, and $\|c\|_1 = \sum_{k=0}^{K-1} c_k $.
    Define unitaries $V_{\mathrm{PREP}}$ and $U_\mathrm{SEL}$ by:
    \[
    V_{\mathrm{PREP}}\ket{0^a} =\frac{1}{\sqrt{\|c\|_1}}
    \sum_{i=0}^{K-1}\sqrt{c_i}\ket{i}, \qquad
    U_{\mathrm{SEL}}= \sum_{i=0}^{K-1} \ket{i}\bra{i}\otimes U_i,\]
    where $V_\mathrm{PREP}$ acts on $a =\lceil \log K\rceil$ qubits and $U_\mathrm{SEL}$ acts on $q+a$ qubits. Define
    \begin{equation}
W=\left(V_{\mathrm{PREP}}^\dagger\otimes \mathbb{I}\right)U_{\mathrm{SEL}}\left(V_{\mathrm{PREP}}\otimes \mathbb{I}\right).
    \end{equation}
    Then $W$ is a $(\|c\|_1,\lceil \log K\rceil, 0)$-block-encoding of $\sum_{k=0}^{K-1} c_k U_k$. If $d_k=c_k e^{i\phi_k}$ are complex coefficients, the same conclusion holds after replacing $U_k$ by $e^{i\phi_k}U_k$.
\end{lemma}
\begin{proof}
    Take $a=\lceil \log K\rceil$ for conciseness. By the definition of an $(\|c\|_1,a,0)$-block-encoding, it suffices to verify that
\[
\|c\|_1(\bra{0^a}\otimes \mathbb{I})W(\ket{0^a}\otimes \mathbb{I})
=\sum_{i=0}^{K-1} c_i U_i.
\]
Substituting the definition of $W$, we obtain
\[
\|c\|_1(\bra{0^a}\otimes \mathbb{I})(V_{\mathrm{PREP}}^\dagger\otimes \mathbb{I})U_{\mathrm{SEL}}(V_{\mathrm{PREP}}\otimes \mathbb{I})(\ket{0^a}\otimes \mathbb{I}).
\]
Using the definitions of $V_\mathrm{PREP}$ and $U_\mathrm{SEL}$, and the orthogonality of the basis, this becomes
\[
\frac{\|c\|_1}{\|c\|_1}\sum_{j=0}^{K-1}\sum_{i=0}^{K-1}\sqrt{c_j}\sqrt{c_i}(\bra{j}\otimes \mathbb{I})U_{\mathrm{SEL}}(\ket{i}\otimes \mathbb{I}) = \sum_{i=0}^{K-1} c_i U_i.
\]
Thus $W$ is a $(\|c\|_1, a, 0)$-block-encoding of $\sum_{i=0}^{K-1} c_i U_i$.
\end{proof}

\begin{lemma}[Block-encoding of an inverse via QSVT~\cite{gilyen_qsvt_2019}]
    \label{lem:inverse-qsvt}
    Let $U_A$ be an $(\alpha,a,\varepsilon_A)$-block-encoding of an invertible matrix $A\in\mathbb{C}^{n\times n}$. Assume the singular values of $A/\alpha$ lie in $[1/\kappa,1]$, and set $\delta = \frac{1}{2\kappa}\in(0,\frac{1}{2})$. Then, there exists a quantum circuit that is an
    \[\left(\frac{8\kappa}{3\alpha}, a+1, \frac{8\kappa}{3\alpha} \left( \varepsilon_\mathrm{pol}+ 4m\sqrt{\frac{\varepsilon_A}{\alpha}}\right)\right)\text{-block-encoding of $A^{-1}$},\]
    where $m=\mathcal{O}\mathopen{}\left(\kappa\log\left(\frac{1}{\varepsilon_\mathrm{pol}}\right)\right)$.
    It can be implemented using an additional ancilla qubit, $m$ calls to $U_A$ and $U_A^\dagger$, and $\mathcal{O}(m(a+1))$ one- and two-qubit gates.
\end{lemma}

\begin{proof}
Let $A=W\Sigma V^\dagger$ be the singular value decomposition of $A$. Since the singular values of $A/\alpha$ lie in $[1/\kappa,1]$, all singular values of $A$ are nonzero, and $A^{-1}=V\Sigma^{-1}W^\dagger$. By Corollary 69 of~\cite{gilyen_qsvt_2019} with $\delta = 1/2\kappa$, there exists an odd real polynomial $P$ of degree $m= \mathrm{deg}(P) = \mathcal{O}\mathopen{}\left(\kappa\log \frac{1}{\varepsilon_\mathrm{pol}}\right)$ such that
    \[
  |P(x)|\leq 1, \quad \forall x\in[-1,1], \qquad
\left|P(x)-\frac{3}{8\kappa x}\right|\leq
\varepsilon_{\mathrm{pol}},  \,\, \forall x\in [1/\kappa,1].
\]
Since $U_A^\dagger$ is an $(\alpha, a, \varepsilon_A)$-block-encoding of $A^\dagger$, we apply QSVT for real odd polynomials. Let $\widetilde A$ denote the matrix extracted from the block-encoding of $A$, so $\|\widetilde A-A\|\leq \varepsilon_A$. Applying Corollary 18 in Ref.~\cite{gilyen_qsvt_2019} to $U_A^\dagger$ gives a block-encoding of $P^\mathrm{SV}(\widetilde A^\dagger/\alpha)$, i.e., the polynomial transformation applied to the singular values of $\widetilde A^\dagger/\alpha$. The construction uses one ancilla qubit. Assuming no QSVT implementation error, we compare the obtained polynomial with $\beta f^\mathrm{SV}(A^\dagger/\alpha)$, where $f(x)=1/x$ and $\beta = \frac{3}{8\kappa}$:
\begin{align*}
\left\|P^\mathrm{SV}(\widetilde A^\dagger/\alpha) - \beta f^\mathrm{SV}(A^\dagger/\alpha)\right\|
&\leq \left\|P^\mathrm{SV}(\widetilde A^\dagger/\alpha) - P^\mathrm{SV}(A^\dagger/\alpha)\right\| \\
& + \left\|P^\mathrm{SV}(A^\dagger/\alpha) - \beta f^\mathrm{SV}(A^\dagger/\alpha)\right\|.
\end{align*}
The first term (robustness term) can be bounded by Lemma 22 in~\cite{gilyen_qsvt_2019},
\[
\left\|{P^{\mathrm{SV}}}({\widetilde A^\dagger}/{\alpha})- {P^{\mathrm{SV}}}({ A^\dagger}/{\alpha})\right\| \leq 4 m\sqrt{\|\widetilde A^\dagger/\alpha -{A^\dagger}/{\alpha} \|} \leq 4m\sqrt{\frac{\varepsilon_A}{\alpha}}.
\]
The second term is bounded by writing the singular value decompositions and using the definition of the spectral norm:
\[
\left\| P^{\mathrm{SV}}({A^\dagger}/{\alpha})-\beta f^{\mathrm{SV}}({A^\dagger}/{\alpha})\right\| = \left\|V\left(P(\Sigma/\alpha)-\beta f(\Sigma/\alpha)\right)W^\dagger\right\| \leq \left\|P(\Sigma/\alpha)-\frac{3}{8\kappa} f(\Sigma/\alpha)\right\| \leq \varepsilon_{\mathrm{pol}}.
\]
Therefore, QSVT yields a $\left(\frac{1}{\beta}, a+1, \frac{1}{\beta}\left(\varepsilon_\mathrm{pol}+4m\sqrt{\frac{\varepsilon_A}{\alpha}}\right)\right)$-block-encoding of $f^\mathrm{SV}(A^\dagger/\alpha)$. Since
\[
f^{\mathrm{SV}}(A^\dagger/\alpha) = Vf(\Sigma/\alpha)W^\dagger = V(\alpha \Sigma^{-1})W^\dagger = \alpha A^{-1},
\]
and $\beta = \frac{3}{8\kappa}$, the circuit is an $\left(\frac{8\kappa}{3\alpha}, a+1, \frac{8\kappa}{3\alpha} \left( \varepsilon_\mathrm{pol}+4m\sqrt{\frac{\varepsilon_A}{\alpha}}\right)\right)$-
block-encoding of $A^{-1}$. Lemma 19 in Ref.~\cite{gilyen_qsvt_2019} gives the stated gate count.
\end{proof}
We are now ready to state our results for block-encoding Riesz projectors. The construction below applies to an arbitrary Riesz projector once the quadrature nodes and weights are fixed. However, explicit bounds on $\varepsilon_\mathrm{trap}(M)$ require additional assumptions on the contour.
\begin{theorem}[Block-encoding of a Riesz projector]
    \label{thm:block-encoding-riesz}
Let $U_\mH$ be an $(\alpha_\mH, a_\mH, \varepsilon_\mH)$-block-encoding of $\mH$ and $\gamma:\mathbb{R}\mapsto \varrho(\mH)$ be a $C^1$, $2\pi$-periodic parametrization of the contour $\Gamma_a$ associated with the Riesz projector $\Pi_a$. Define the $M$-point trapezoid approximation $\Pi_{a,M}$ as in \Cref{eq:general-trapezoid}, and assume known bounds
\begin{equation}
    M_\gamma \geq \|(\gamma(\theta)\mathbb{I}-\mH)^{-1}\|, \qquad d_\gamma\geq |\gamma'(\theta)|, \qquad b_\gamma \geq |\gamma(\theta)|, \qquad \forall\theta\in[0,2\pi].
\end{equation}
Further assume that $\|\Pi_{a,M}-\Pi_a\|\leq \varepsilon_\mathrm{trap}(M)$, with $\varepsilon_\mathrm{trap}(M)$ given by \Cref{prop:error-trapezoid-periodic} when its hypotheses hold. Then, there exists a unitary $U_{\Pi_a}$ that is an $(\alpha_\Pi, a_\mH+\lceil \log M \rceil +2, \varepsilon_\Pi+\varepsilon_\mathrm{trap} (M))$-block-encoding of $\Pi_{a}$, where
\[
\alpha_\Pi \leq \frac{8}{3}M_\gamma d_\gamma\left(1+ \frac{b_\gamma}{\alpha_\mH}\right), \qquad \varepsilon_\Pi \leq\alpha_\Pi
\left(\varepsilon_{\mathrm{pol}}+4m\sqrt{\frac{\varepsilon_\mH}{\alpha_\mH}}\right),
\]
with $\kappa= M_\gamma(\alpha_\mH+b_\gamma)$ and
\[
m =\mathcal{O}\mathopen{}\left( \kappa\log\frac{1}{\varepsilon_{\mathrm{pol}}}\right).
\]
Moreover, $U_{\Pi_a}$ can be implemented with $\mathcal{O}(m)$ queries to controlled $U_\mH,U_\mH^\dagger$, and $\mathcal{O}(m(a_\mH+M)+M)$ additional one- and two-qubit gates.
\end{theorem}

\begin{proof}
    We start by constructing a block-encoding of $z_k\mathbb{I} - \mH$. Write $z_k = |z_k|e^{i\phi_k}$ and set $\alpha_k = \alpha_\mH +|z_k|$. Since $\|z_k\mathbb{I}-\mH\|\leq |z_k| + \alpha_\mH$, this is a valid block-encoding normalization. Using \Cref{lem:lcu} for $|z_k|e^{i\phi_k}\mathbb{I} - \mH$, we define
    \[
    V_{\mathrm{PREP}_k}\ket{0} =\sqrt{\frac{|z_k|}{\alpha_k}}\ket{0}+\sqrt{\frac{\alpha_\mH}{\alpha_k}}\ket{1}, \qquad
    U_{\mathrm{SEL}_k}= \ket{0}\bra{0}\otimes e^{i\phi_k}\mathbb{I} + \ket{1}\bra{1}\otimes (-U_\mH).
\]
Then, $W_k= (V_{\mathrm{PREP}_k}^\dagger\otimes \mathbb{I})U_{\mathrm{SEL}_k}(V_{\mathrm{PREP}_k}\otimes \mathbb{I})$ is an $(\alpha_k,a_\mH+1,\varepsilon_\mH)$-block-encoding of $z_k\mathbb{I}-\mH$. Now we build the resolvents $(z_k\mathbb{I} - \mH)^{-1}$. To apply the same QSVT polynomial to each resolvent, we need to choose a uniform $\kappa$ so that the singular values of $(z_k\mathbb{I}-\mH)/\alpha_k$ lie in $[1/\kappa,1]$, for all $k$. Using the definitions of $M_\gamma$ and $b_\gamma$,
\[
\frac{\alpha_k}{\sigma_\mathrm{min}(z_k\mathbb{I}-\mH)} \leq  \|(z_k\mathbb{I} - \mH)^{-1}\|(\alpha_\mH +|z_k|) \leq M_\gamma (\alpha_\mH + b_\gamma) = \kappa.
\]
Therefore, $\kappa = M_\gamma (\alpha_\mH + b_\gamma)$ is a valid QSVT condition number bound for every node $k$. Applying \Cref{lem:inverse-qsvt} to $W_k$, we obtain a unitary $V_k$ that is an
\[
\left(\frac{8\kappa}{3\alpha_k},a_\mH+2,\frac{8\kappa}{3\alpha_k}\left(\varepsilon_{\mathrm{pol}}+4m\sqrt{\frac{\varepsilon_\mH}{\alpha_k}}\right)\right)\text{-block-encoding of $(z_k\mathbb{I}-\mH)^{-1}$,}
\]
where
\[
m =\mathcal{O}\mathopen{}\left( \kappa\log\frac1{\varepsilon_{\mathrm{pol}}}\right) = \mathcal{O}\mathopen{}\left( M_\gamma (\alpha_\mH + b_\gamma)\log\frac1{\varepsilon_{\mathrm{pol}}}\right).
\]
The final step is a second application of the LCU lemma to block-encode the discrete sum. Write $w_k = |w_k|e^{i\varphi_k}$ and define $\tilde V_k = e^{i\varphi_k}V_k$, which is unitary. Then, we have
\[
U_{\mathrm{SEL}} = \sum_{k=0}^{M-1}\ket{k}\bra{k}\otimes \tilde V_k, \qquad V_{\mathrm{PREP}}\ket{0^{\lceil\log M\rceil}}
=\frac{1}{\sqrt{\alpha_\Pi}}
\sum_{k=0}^{M-1}\sqrt{c_k}\,\ket{k}, \qquad c_k = \frac{8\kappa}{3M\alpha_k}|w_k|, \]
and the unitary $U_{\Pi_{a,M}}= (V_{\mathrm{PREP}}^\dagger\otimes \mathbb{I})U_{\mathrm{SEL}}(V_{\mathrm{PREP}}\otimes \mathbb{I})$ is an $(\alpha_\Pi, a_\mH+\lceil \log M \rceil+2, \varepsilon_\Pi)$-block-encoding of $\Pi_{a,M}$. Moreover,
\[
\alpha_\Pi = \sum_{k=0}^{M-1} c_k  = \frac{8\kappa}{3M}\sum_{k=0}^{M-1} \frac{|w_k|}{\alpha_k}\leq \frac{8M_\gamma (\alpha_\mH + b_\gamma)}{3\alpha_\mH}\frac{1}{M}\sum_{k=0}^{M-1} {|w_k|}\leq \frac{8}{3}M_\gamma d_\gamma\left(1+ \frac{b_\gamma}{\alpha_\mH}\right).
\]
By error propagation,
\begin{align*}
\varepsilon_\Pi
&\leq \frac{1}{M}\sum_{k=0}^{M-1}|w_k|\frac{8\kappa}{3\alpha_k}\left(\varepsilon_{\mathrm{pol}}+4m\sqrt{\frac{\varepsilon_\mH}{\alpha_k}}\right) \\
&\leq \frac{1}{M}\sum_{k=0}^{M-1}|w_k|\frac{8\kappa}{3\alpha_k}\left(\varepsilon_{\mathrm{pol}}+4m\sqrt{\frac{\varepsilon_\mH}{\alpha_\mH}}\right) \\
&= \alpha_\Pi\left(\varepsilon_{\mathrm{pol}}+4m\sqrt{\frac{\varepsilon_\mH}{\alpha_\mH}}\right).
\end{align*}
Finally, by the definition of block-encoding, with $a_\Pi = a_\mH+\lceil \log M \rceil+2$, we have:
\[
\left\|\alpha_\Pi(\bra{0}\otimes\mathbb{I})U_{\Pi_{a,M}}(\ket{0}\otimes\mathbb{I})-\Pi_a
\right\|
\le
\left\|
\alpha_\Pi(\bra{0}\otimes\mathbb{I})U_{\Pi_{a,M}}(\ket{0}\otimes\mathbb{I})-\Pi_{a,M}
\right\|+\|\Pi_{a,M}-\Pi_a\|\leq \varepsilon_\Pi + \varepsilon_\mathrm{trap}(M).\]
Hence $U_{\Pi_{a,M}}$ is also an $(\alpha_\Pi,a_\mH+\lceil\log M\rceil+2,\varepsilon_\Pi+\varepsilon_{\mathrm{trap}}(M))$-block-encoding of $\Pi_a$. We rename this unitary as $U_{\Pi_a}$.

$U_{\Pi_a}$ makes $\mathcal{O}(m)$ uses of $U_\mH, U_\mH^\dagger$. Additionally, QSVT uses $\mathcal{O}(m a_\mH)$ single and two-qubit gates. This number is independent of $M$ because the outer LCU applies the same QSVT circuit for every branch $k$, as shown in \Cref{fig:UPia-circuits-stacked}. The $k$ dependence appears in the $(z_k\mathbb{I}-\mH)$ block-encoding and in the LCU coefficients $c_k$.
Each use of the block-encoding $W_k$ for $z_k\mathbb{I}-\mH$ requires loading the node-dependent angle and phase determined by $z_k$. Assuming no structure for the list of nodes, these are implemented as multiplexed rotations controlled on the $\lceil \log M \rceil$ control qubits from the LCU, at cost $\mathcal{O}(M)$ per QSVT step, giving a total of $\mathcal{O}(mM)$ gates. Finally, the LCU unitaries $V_\mathrm{PREP}$ and $V_\mathrm{PREP}^\dagger$, together with the controlled phases for $e^{i\varphi_k}V_k$, cost $\mathcal{O}(M)$ gates using generic state preparation.
\end{proof}

\begin{figure}[!ht]
    \centering
    \resizebox{0.52\textwidth}{!}{%
    \begin{tabular}{c@{\hspace{0.3em}}c@{\hspace{0.3em}}c}
    \begin{quantikz}[column sep=0.48cm, row sep=0.33cm]
        \lstick{$\ket{0}^{\,a_\mH+\lceil\log M\rceil+2}$}
            & \gate[2]{U_{\Pi_a}}
            & \qw \\
        \lstick{$\ket{\psi}$}
            &
            & \qw
    \end{quantikz}
    &
    $\displaystyle =$
    &
    \begin{quantikz}[column sep=0.38cm, row sep=0.28cm]
        \lstick{$\ket{0}^{\,\lceil\log M\rceil}$}
            & \gate{V_\mathrm{PREP}}
            & \muxctrl{1}\gategroup[4,steps=1,style={dashed,draw=red!70!black,fill=red!10,rounded corners=2pt,inner xsep=2pt,inner ysep=4pt},background,label style={label position=above,anchor=south,yshift=0.06cm}]{$U_\mathrm{SEL}=\sum_k\ket{k}\!\bra{k}\otimes\widetilde V_k$}
            & \gate{V_\mathrm{PREP}^{\dagger}}
            & \qw \\
        \lstick{$\ket{0}_{s}$}
            & \qw
            & \gate[3]{\widetilde V_k}
            & \qw
            & \qw \\
        \lstick{$\ket{0}^{\,a_\mH+1}$}
            & \qw
            &
            & \qw
            & \qw \\
        \lstick{$\ket{\psi}$}
            & \qw
            &
            & \qw
            & \qw
    \end{quantikz}
    \end{tabular}
    }

    \vspace{0.6em}
    \makebox[\textwidth][c]{%
    \(\displaystyle =\)\hspace{0.25em}%
    \resizebox{0.94\textwidth}{!}{%
    \begin{quantikz}[column sep=0.075cm, row sep=0.17cm]
        \lstick{$\ket{0}^{\,\lceil\log M\rceil}$}
            & \gate{V_\mathrm{PREP}}
            & \qw\gategroup[5,steps=26,style={dashed,draw=red!60!black,fill=red!8,rounded corners=2pt,inner xsep=-2pt,inner ysep=5pt},background,label style={label position=above,anchor=south,yshift=-0.08cm}]{}
            & \qw & \qw & \qw
            & \qw
            & \muxctrl{2} & \muxctrl{2} & \qw & \muxctrl{2}
            & \qw
            & \qw & \qw & \qw
            & \qw
            & \muxctrl{2} & \muxctrl{2} & \qw & \muxctrl{2}
            & \qw
            & \qw & \cdots & \qw
            & \qw & \qw & \qw & \qw
            & \gate{V_\mathrm{PREP}^{\dagger}} & \qw \\
        \lstick{$\ket{0}_{s}$}
            & \qw
            & \gate{H} & \targ{} & \gate{e^{-i\phi_0 Z}} & \targ{}
            & \qw
            & \qw & \qw & \qw & \qw
            & \qw
            & \targ{} & \gate{e^{-i\phi_1 Z}} & \targ{}
            & \qw
            & \qw & \qw & \qw & \qw
            & \qw
            & \qw & \cdots & \qw
            & \targ{} & \gate{e^{-i\phi_{m_\Pi} Z}} & \targ{} & \gate{H}
            & \qw & \qw \\
        \lstick{$\ket{0}$}
            & \qw
            & \qw & \octrl{-1} & \qw & \octrl{-1}
            & \qw
            & \gate{R_y(2\vartheta_k)}\gategroup[3,steps=4,style={dashed,draw=blue,fill=blue!10,rounded corners=2pt,inner xsep=3pt,inner ysep=2pt},background,label style={label position=below,anchor=north,xshift=-0.2cm,yshift=0.3cm}]{$W_k$}
            & \gate{e^{-i\varphi_k Z}} & \ctrl{1} & \gate{R_y(-2\vartheta_k)}
            & \qw
            & \octrl{-1} & \qw & \octrl{-1}
            & \qw
            & \gate{R_y(2\vartheta_k)}\gategroup[3,steps=4,style={dashed,draw=blue,fill=blue!10,rounded corners=2pt,inner xsep=3pt,inner ysep=2pt},background,label style={label position=below,anchor=north,xshift=-0.3cm,yshift=0.45cm}]{$W_k^\dagger$}
            & \gate{e^{i\varphi_k Z}} & \ctrl{1} & \gate{R_y(-2\vartheta_k)}
            & \qw
            & \qw & \cdots & \qw
            & \octrl{-1} & \qw & \octrl{-1} & \qw
            & \qw & \qw \\
        \lstick{$\ket{0}^{\,a_\mH}$}
            & \qw
            & \qw & \octrl{-2} & \qw & \octrl{-2}
            & \qw
            & \qw & \qw & \gate[2]{U_\mH} & \qw
            & \qw
            & \octrl{-2} & \qw & \octrl{-2}
            & \qw
            & \qw & \qw & \gate[2]{U_\mH^\dagger} & \qw
            & \qw
            & \qw & \cdots & \qw
            & \octrl{-2} & \qw & \octrl{-2} & \qw
            & \qw & \qw \\
        \lstick{$\ket{\psi}$}
            & \qw
            & \qw & \qw & \qw & \qw
            & \qw
            & \qw & \qw & & \qw
            & \qw
            & \qw & \qw & \qw
            & \qw
            & \qw & \qw & & \qw
            & \qw
            & \qw & \cdots & \qw
            & \qw & \qw & \qw & \qw
            & \qw & \qw
    \end{quantikz}
    }
    }
\caption{Circuit representations of $U_{\Pi_a}$ at increasing granularity. \textbf{Top (left):} The simplest circuit for $U_{\Pi_a}$ as an
    $(\alpha_\Pi,a_\mH+\lceil\log M\rceil+2,\varepsilon_\Pi+\varepsilon_{\mathrm{trap}}(M))$-block-encoding of
    $\Pi_a$. \textbf{Top (right):} The outer LCU for preparing the sum of weighted resolvents. \textbf{Bottom:} The $U_\mathrm{SEL}$ block expanded into the multiplexed QSVT calls in each $\widetilde{V}_k$, with the diamonds indicating $\ket{k}$-multiplexed control over the $k$-dependent rotation parameters $\vartheta_k,\varphi_k$ inside $W_k,W_k^\dagger$; the $U_\mH$ block-encoding is shared across all $M$ contour points and is not separately controlled by $k$. Within each $W_k$,
    $\cos\vartheta_k = \sqrt{|z_k|/\alpha_k}$ and
    $\varphi_k = \frac{1}{2}(\pi - \theta_k)$ where $z_k = |z_k|e^{i\theta_k}$ are the quadrature nodes.}
    \label{fig:UPia-circuits-stacked}
\end{figure}
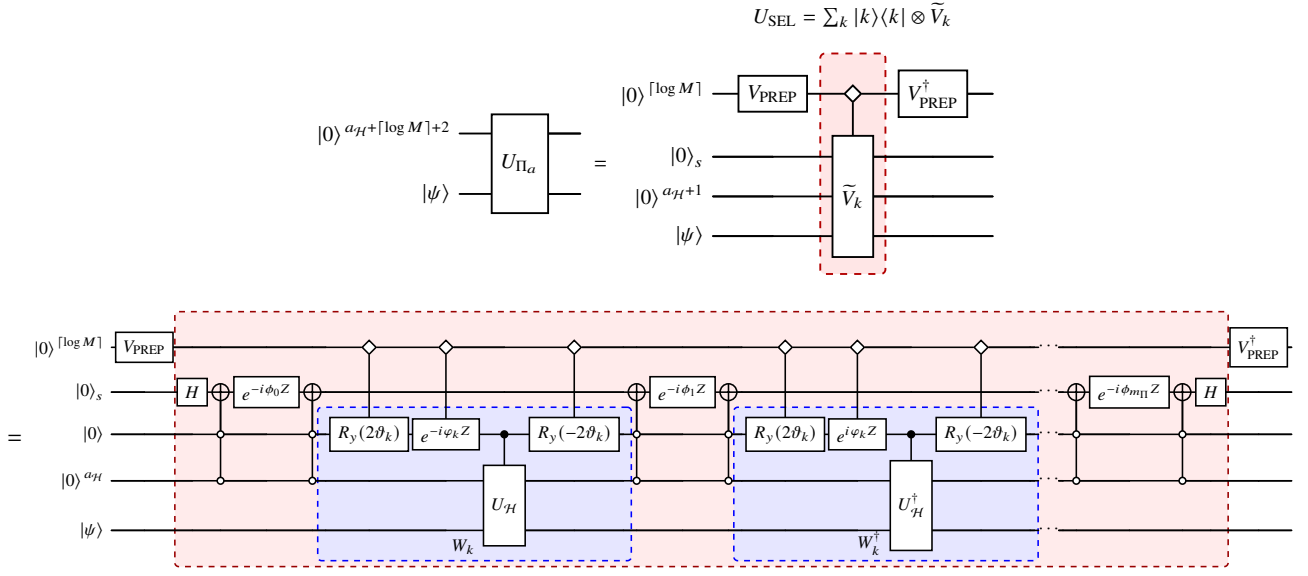

An important feature of this block-encoding construction is that it is algorithmically agnostic to the matrix $\mH$ and the particular contour. Note that we did not specify any conditions for $\mH$ in this case---it could be any matrix. Once the quadrature nodes $z_k$ and weights $w_k$ are fixed, the same routine applies, which makes it highly adaptable. The contour choice enters mainly through the discretization error $\varepsilon_{\mathrm{trap}}(M)$, the number of quadrature points $M$, and the constants $M_\gamma$, $d_\gamma$, and $b_\gamma$. A dominant cost of the implementation is querying $U_\mH$ for building the resolvents, although the number of queries is independent of $M$. However, it does grow with $M_\gamma$, and with the size of the contour through $b_\gamma\geq|\gamma(\theta)|$. At this point we observe the benefit in choosing an efficient contour that has tighter bounds on the constants and also holds the potential to reduce $M$. In the non-normal case, $M_\gamma$ can be large even far from the spectrum, which would affect both the number of queries and the sub-normalization.

Linear dependence in $M$ appears in the ancilla and gate overhead, but $M$ only scales logarithmically with the target precision $\varepsilon_\mathrm{trap}$. The $M$ dependence may still be important under large resolvent norms or small analyticity strips. In these cases, the cost of loading the pairs $\{z_k,w_k\}_{k=0}^{M-1}$ and the $k$-dependent operations could dominate. However, we emphasize that we provided a conservative, structure-agnostic gate-count estimate, which comes from treating the coefficients as arbitrary classical data. If additional structure is available, such as if the coefficients are specified by an efficiently integrable smooth function, then one may prepare the corresponding coefficient state using quantum arithmetic methods, leading to a cost polynomial in $\log M$ rather than linear in $M$ \cite{vedral_qnea_1996, grover_cstc_2002, haner_oqca_2018}.

\subsection{Specialization to the smoothed semicircle contour}
We can now establish the parameters of the Riesz projector block-encoding under the smoothed semicircle.
\begin{corollary}[Smoothed-semicircle antistable Riesz projector]
\label{cor:block-encoding-riesz-smooth}
    Assume $\sigma(\mH)\cap i\mathbb{R}=\emptyset$, with gap $\delta=\min_{\lambda\in\sigma(\mH)}|\Re\lambda|>0$. Assume the setting of \Cref{thm:block-encoding-riesz} and choose the smoothed semicircle contour with $z_0=0$, $\omega=\delta$ and $R=2\alpha_\mH$. Let $M_\gamma = \max_{\theta \in [0,2\pi]}\|(\gamma_\mathrm{SD}(\theta)\mathbb{I} - \mH)^{-1}\|$. Then $U_{\Pi_a}$ can be chosen as an $(\alpha_\Pi,a_\Pi,\bar\varepsilon_\Pi)$-block-encoding of $\Pi_a$ with
\begin{equation}
\label{eq:be-scaling-semicircle}
 \alpha_\Pi=\mathcal{O}(M_\gamma\alpha_\mH), \qquad a_\Pi=a_\mH+\lceil\log M\rceil+2, \qquad  \bar\varepsilon_\Pi=\varepsilon_\Pi+\varepsilon_{\mathrm{trap}},
\end{equation}
where
\begin{equation}
    \varepsilon_\Pi\leq \alpha_\Pi\left(\varepsilon_{\mathrm{pol}}^{(\Pi)}+4m_\Pi\sqrt{\frac{\varepsilon_\mH}{\alpha_\mH}}\right),\qquad
 m_\Pi=\mathcal{O}\left(M_\gamma\alpha_\mH\log\frac{1}{\varepsilon_{\mathrm{pol}}^{(\Pi)}}\right).
\end{equation}
Moreover, using the admissible strip width from \Cref{prop:smooth-semicircle-strip} and \Cref{lem:eta-scaling}, the number of quadrature points is
\begin{equation}
\label{eq:quadrature-semicircle}
M=\mathcal{O}\mathopen{}\left(\frac{1}{\eta}\log\frac{\max\{\gamma_-, \gamma_+\}}{\varepsilon_\mathrm{trap}}\right), \qquad
 0<\eta<\operatorname{arsinh}(\omega/R),
\end{equation}
The implementation uses $\mathcal{O}(m_\Pi)$ queries to controlled $U_\mH,U_\mH^\dagger$ and $\mathcal{O}(m_\Pi(a_\mH+M)+M)$ additional one- and two-qubit gates.
\end{corollary}
\begin{proof}
For $\gamma_{\mathrm{SD}}$, \Cref{prop:smooth-semicircle} gives $d_\gamma\leq R$. Moreover, we can bound $|\gamma_{\mathrm{SD}}(\theta)|\leq z_0+\omega/2+R = b_\gamma$, and the choice $z_0 = 0$, $\omega = \delta$, $R=2\alpha_\mH$ yields $b_\gamma=\mathcal{O}(\alpha_\mH)$ since $\delta$ is always smaller than $\alpha_\mH$. By \Cref{thm:block-encoding-riesz}, these bounds imply $\alpha_\Pi=\mathcal{O}(M_\gamma\alpha_\mH)$ and $m=\mathcal{O}(M_\gamma\alpha_\mH\log(1/\varepsilon_{\mathrm{pol}}))$. The quadrature bound follows from \Cref{prop:error-trapezoid-periodic,prop:smooth-semicircle-strip}, with the substrip provided by \Cref{lem:eta-scaling}. The query and gate counts come directly from \Cref{thm:block-encoding-riesz}.
\end{proof}

Even though $M$ does not enter the query count to $U_\mH$ directly, it enters the gate count and number of ancilla qubits. In particular, it can grow when the analyticity region becomes small, for example if there is a small spectral gap $\delta$. By \Cref{lem:eta-scaling}, an admissible $\eta_\star=\Theta(\delta/\alpha_\mH)$ exists for the chosen parameters, so that
\[
M=\mathcal{O}\mathopen{}\left(\frac{\alpha_\mH}{\delta}\log\frac{\max\{\gamma_-, \gamma_+\}}{\varepsilon_\mathrm{trap}}\right).
\]

\subsection{Block-encoding the stabilizing solution of the CARE}
Recall that by \Cref{cor:riesz-care-solution} and under the relevant assumptions, the stabilizing solution of the CARE can be represented as:
\[
\Pi_a = (\Pi_{1} \quad \Pi_{2}),\qquad \Pi_{2}X_s = -\Pi_{1},
\]
where $\Pi_a$ is the Riesz projector on the antistable subspace of $\mH$ and $\Pi_{1},\Pi_{2} \in\mathbb{C}^{2n\times n}$ (we dropped the $a$ subscript from the corollary for conciseness). We have built the machinery to block-encode $\Pi_a$, and the remaining steps are to extract the rectangular blocks, block-encode the Moore--Penrose pseudoinverse $\Pi_{2}^+$, and multiply block-encodings to block-encode the solution $X_s = -\Pi_{2}^{+}\Pi_{1}$.

The intermediate matrices are rectangular. We therefore use a rectangular version of block-encodings, allowing for different left and right subspace sizes.

\begin{definition}[Rectangular block-encoding]
\label{def:rect-block-encoding}
Let $A\in\mathbb{C}^{n_L\times n_R}$. We say that $U_A$ is an $(\alpha,a_L,a_R,\varepsilon)$-block-encoding of $A$ if
\begin{equation}
\left\|\alpha(\bra{0^{a_L}}\otimes\mathbb{I}_{n_L})U_A(\ket{0^{a_R}}\otimes\mathbb{I}_{n_R})-A\right\|\leq\varepsilon.
\end{equation}
\end{definition}

When $n_L=n_R$ and $a_L=a_R=a$, this reduces to the usual $(\alpha,a,\varepsilon)$-block-encoding of a square matrix. The standard QSVT construction applies directly for odd polynomials of such rectangular encodings.

We can now build rectangular block-encodings for the rectangular blocks of $\Pi_a$.

\begin{lemma}[Rectangular block-encodings of the blocks of $\Pi_a$]
\label{lem:rectangular-be-Pia}
Let $\Pi_a\in\mathbb{C}^{2n\times 2n}$ with column blocks $\Pi_a=\begin{pmatrix}\Pi_1 & \Pi_2\end{pmatrix},\,\Pi_1,\Pi_2\in\mathbb{C}^{2n\times n}$. If $U_{\Pi_a}$ is an $(\alpha_\Pi,a_\Pi,\bar\varepsilon_\Pi)$-block-encoding of $\Pi_a$, then
\[
U_{\Pi_1}=U_{\Pi_a},
\qquad
U_{\Pi_2}=U_{\Pi_a}
\left(\mathbb{I}_{2^{a_\Pi}}\otimes X\otimes\mathbb{I}_n\right)
\]
are $(\alpha_\Pi,a_\Pi,a_\Pi+1,\bar\varepsilon_\Pi)$-block-encodings of $\Pi_1$ and $\Pi_2$, respectively.
\end{lemma}
\begin{proof}
First, define $E_1 = \ket{0}\otimes \mathbb{I}_n$ and $E_2 = \ket{1}\otimes \mathbb{I}_n$. Then $\Pi_j = \Pi_a E_j$ for each $j\in\{1,2\}$. Let $\widetilde\Pi_a $ be the block-encoded version of $\Pi_a$, so that
\[
\widetilde\Pi_a = \alpha_\Pi(\bra{0^{a_\Pi}}\otimes \mathbb{I}_{2n})U_{\Pi_a}(\ket{0^{a_\Pi}}\otimes \mathbb{I}_{2n}),\qquad \|\widetilde \Pi_a - \Pi_a\|\leq \bar \varepsilon_\Pi. \]
It can be readily verified that, for each $j\in\{1,2\}$
\begin{equation}
\label{eq:rectangular-be-Pia}
 \alpha_\Pi (\bra{0^{a_\Pi}}\otimes \mathbb{I}_{2n})U_{\Pi_a}(\ket{0^{a_\Pi}}\otimes E_j) = \widetilde \Pi_a E_j = \widetilde{\Pi}_j.
\end{equation}
Therefore,
\begin{equation}
\label{eq:rectangular-be-error-Pia}
\|\widetilde{\Pi}_j-\Pi_j\| = \left\| \alpha_\Pi (\bra{0^{a_\Pi}}\otimes \mathbb{I}_{2n})U_{\Pi_a}(\ket{0^{a_\Pi}}\otimes E_j)  -\Pi_j\right\| = \left\| (\widetilde \Pi_a - \Pi_a) E_j\right\|  \leq \|\widetilde \Pi_a - \Pi_a\|   \leq \bar\varepsilon_\Pi,
\end{equation}
where we used $E_j^\dagger E_j = \mathbb{I}_n$ and $\|E_j\| = 1$. Hence, the unitaries $U_{\Pi_1}=U_{\Pi_a}$ and $U_{\Pi_2}=U_{\Pi_a}\left(\mathbb{I}_{2^{a_\Pi}}\otimes X\otimes\mathbb{I}_n\right)$ are $(\alpha_\Pi,a_\Pi,a_\Pi+1,\bar\varepsilon_\Pi)$-block-encodings of $\Pi_1$ and $\Pi_2$, respectively.
\end{proof}

The definitions of the singular value decomposition and functions of the singular values for rectangular matrices are analogous to the square case. Thus, we have an equivalent result to the QSVT inverse (\Cref{lem:inverse-qsvt}) but for the pseudoinverse and rectangular matrices, which we apply directly to $\Pi_2$.

\begin{corollary}[QSVT pseudoinverse of $\Pi_2$]\label{cor:qsvt-P2inv}
Let $U_{\Pi_2} = U_{\Pi_a}\left(\mathbb{I}_{2^{a_\Pi}}\otimes X\otimes\mathbb{I}_n\right)$ be an $(\alpha_\Pi, a_\Pi, a_\Pi+1, \bar\varepsilon_\Pi)$-block-encoding of $\Pi_2\in\mathbb{C}^{2n\times n}$. Assume that the singular values of $\Pi_2/\alpha_\Pi$ lie in $[1/\kappa_2,1]$, so that $\Pi_2$ has full column rank. Then there exists a circuit $U_{\Pi_2^+}$ that is an $(\alpha_+,a_\Pi+2,a_\Pi+1,\varepsilon_+)$-block-encoding of $\Pi_2^+$ where
\begin{equation}\label{eq:qsvt-Pi2}
\alpha_+=\frac{8\kappa_2}{3\alpha_\Pi}, \qquad
 \varepsilon_+=\alpha_+\left(\varepsilon_{\mathrm{pol}}^{(+)}+4m_+\sqrt{\frac{\bar\varepsilon_\Pi}{\alpha_\Pi}}\right), \qquad
 m_+=\mathcal{O}\mathopen{}\left(\kappa_2\log\frac{1}{\varepsilon_{\mathrm{pol}}^{(+)}}\right).
\end{equation}
$U_{\Pi_2^+}$ can be implemented using an additional ancilla qubit, $m_+$ uses of $U_{\Pi_a}$ and $U_{\Pi_a}^\dagger$, and $\mathcal{O}(m_+(a_\Pi+2))$ additional single- and two-qubit gates.
\end{corollary}

\begin{proof}
The adjoint $U_{\Pi_2}^\dagger$ is an $(\alpha_\Pi, a_\Pi+1, a_\Pi, \bar\varepsilon_\Pi)$-block-encoding of $\Pi_2^\dagger$. Let
$\Pi_2=W\Sigma V^\dagger$ be the singular value decomposition. Then $\Pi_2^\dagger=V\Sigma W^\dagger$, and for the inverse $f(x) = 1/x$, acting only on the nonzero singular values,
\[
f^{\mathrm{SV}}(\Pi_2^\dagger/\alpha_\Pi)=Vf(\Sigma/\alpha_\Pi)W^\dagger=V\alpha_\Pi\Sigma^{-1}W^\dagger
=\alpha_\Pi \Pi_2^+.
\]

By assumption, the singular values of $\Pi_2/\alpha_\Pi$ lie in $[1/\kappa_2,1]$. Hence by direct application of \Cref{lem:inverse-qsvt} to the rectangular block-encoding $U_{\Pi_2}$, we obtain a rectangular block-encoding of $\Pi_2^+$ with normalization and error given by \Cref{eq:qsvt-Pi2}. The number of queries to $U_{\Pi_2}$ and $U_{\Pi_2}^\dagger$ is given by $m_+ =\mathcal{O}(\kappa_2\log(1/\varepsilon_\mathrm{pol}^{(+)}))$, which directly translates into queries to $U_{\Pi_a}, U_{\Pi_a}^\dagger$. Taking the adjoint changes the right and left ancilla counts, and both increase by one due to the QSVT ancilla. The stated gate counts come directly from \Cref{lem:inverse-qsvt}.
\end{proof}
In the stabilizing CARE solution scenario, the block $\Pi_2$ is guaranteed to have full column rank, while the implemented block-encoding need not itself be certified to have full column rank. QSVT approximates the inverse of the exact $\Pi_2$, and the block-encoding error perturbation is absorbed into the robustness error term in \Cref{eq:qsvt-Pi2}. After building $\Pi_2^+$, we need to compute the product with $\Pi_1$. We provide a method with only one additional ancilla qubit based on the product of projected unitary encodings~\cite{sunderhauf_gqsvt_2023}.

\begin{lemma}[Product of rectangular block-encodings]
\label{lem:product-rect-be}
Let $A_1\in \mathbb{C}^{n_L\times m}$ and $A_2\in \mathbb{C}^{m\times m_R}$. Suppose that $U_1$ is an $(\alpha_1,a_{1,L},a_{1,R},\varepsilon_1)$-block-encoding of $A_1$, and that $U_2$ is an $(\alpha_2,a_{2,L},a_{2,R},\varepsilon_2)$-block-encoding of $A_2$. Assume $U_1$ and $U_2$ are of equal size (WLOG by padding with identities), so that $a_{1,R}=a_{2,L}=a$. There exists a circuit $U$ that is an $(\alpha_1\alpha_2, a_{1,L}+1, a_{2,R}+1, \varepsilon_{12})$-block-encoding of $A_1A_2$, with $\varepsilon_{12} \leq \alpha_1\varepsilon_2 + (\alpha_2+\varepsilon_2)\varepsilon_1$. In particular,
\[
U = (\mathbb{I}_2\otimes U_1)\left(\mathrm{C}_{a} \mathrm{NOT}\otimes \mathbb{I}_m\right)(\mathbb{I}_2\otimes U_2),
\]
where the construction introduces an ancilla qubit. The unitary $U$ requires one call to $U_1$, one call to $U_2$, and a $\mathrm{CNOT}$ gate on the additional ancilla, controlled by the $a$ inner ancilla qubits.
\end{lemma}
\begin{proof}
By definition, the block-encoded versions of $A_1$ and $A_2$,
\[
\widetilde{A}_1 = \alpha_1 (\bra{0^{a_{1,L}}} \otimes \mathbb{I}_{n_L}) U_1 (\ket{0^{a}} \otimes \mathbb{I}_m),  \qquad \widetilde{A}_2 = \alpha_2 (\bra{0^{a}} \otimes \mathbb{I}_m) U_2 (\ket{0^{a_{2,R}}} \otimes \mathbb{I}_{m_R})
\]
satisfy $\|\widetilde{A}_1 - A_1\| \leq \varepsilon_1$ and $\|\widetilde{A}_2 - A_2\| \leq \varepsilon_2$. Now we verify that $U$ block-encodes $A_1 A_2$ by applying the definition of an $(\alpha_1\alpha_2, a_{1,L}+1, a_{2,R}+1, \varepsilon_{12})$-block-encoding of $U$, which we denote by $\widetilde{A}_{12}$:
\[
\widetilde{A}_{12} =\alpha_1 \alpha_2 (\bra{0} \otimes \bra{0^{a_{1,L}}} \otimes \mathbb{I}_{n_L}) U (\ket{0} \otimes \ket{0^{a_{2,R}}} \otimes \mathbb{I}_{n_R}).
\]
Introducing the definition of $U$, the product on the $\mathrm{CNOT}$ $\ket{1}$ branch is zero, and we are left with
\[
\widetilde{A}_{12} =\alpha_1 \alpha_2 (\bra{0} \otimes \bra{0^{a_{1,L}}} \otimes \mathbb{I}_{n_L})(\mathbb{I}_2\otimes U_1)\left(\mathbb{I}_2\otimes \ket{0^a}\bra{0^a} \otimes \mathbb{I}_m\right)(\mathbb{I}_2\otimes U_2)(\ket{0} \otimes \ket{0^{a_{2,R}}} \otimes \mathbb{I}_{n_R}).
\]
The product on the remaining ancilla component is $1$, and using the definitions of the block-encodings of $A_1$ and $A_2$, we obtain
\[
\widetilde{A}_{12} =\alpha_1 \alpha_2 (\bra{0^{a_{1,L}}} \otimes \mathbb{I}_{n_L})U_1\left(\ket{0^a}\bra{0^a} \otimes \mathbb{I}_m\right)U_2(\ket{0^{a_{2,R}}} \otimes \mathbb{I}_{n_R}) = \alpha_1\alpha_2\widetilde{A}_1\widetilde{A}_2.
\]
Therefore, $U$ is indeed an $(\alpha_1\alpha_2, a_{1,L}+1, a_{2,R}+1, \varepsilon_{12})$-block-encoding of $\widetilde{A}_1\widetilde{A}_2$. It remains to bound $\varepsilon_{12}$ by error propagation:
\begin{align*}
\|\widetilde A_1 \widetilde A_2 - A_1 A_2\|
&= \|\widetilde A_1(\widetilde A_2 - A_2) + (\widetilde A_1 - A_1)A_2\| \\
&\leq \|\widetilde A_1(\widetilde A_2 - A_2)\| + \|(\widetilde A_1 - A_1)A_2\|
\leq \alpha_1\varepsilon_2 + (\alpha_2+\varepsilon_2)\varepsilon_1,
\end{align*}
where we used $\|\widetilde A_1\|\leq \alpha_1$ and $\| A_2\|\leq \|\widetilde A_2\|+\varepsilon_2\leq \alpha_2+\varepsilon_2$.
\end{proof}

By combining all the ingredients, we are ready to provide a block-encoding of the CARE solution. Note that we need to pad the block-encoding of $\Pi_1$ with one ancilla qubit since $\Pi_2$ acquires an ancilla after the pseudoinverse.

\begin{theorem}[Block-encoding of the stabilizing CARE solution]\label{thm:care-be}
Assume the setting of \Cref{cor:riesz-care-solution}, so that the stabilizing solution of the CARE exists and satisfies $X_s=-\Pi_2^+\Pi_1$ for $\Pi_a=\begin{pmatrix} \Pi_1 & \Pi_2\end{pmatrix}$, and the singular values of $\Pi_2/\alpha_\Pi$ lie in $[1/\kappa_2,1]$. Let $U_{\Pi_a}$ be the $(\alpha_\Pi, a_\Pi, \bar\varepsilon_\Pi)$-block-encoding of $\Pi_a$ over $M$ smoothed semicircle nodes from \Cref{cor:block-encoding-riesz-smooth}.

Then there exists a unitary $U_{X_s}$ which is an $(\alpha_X, a_X, \varepsilon_X)$-block-encoding of $X_s$, with
\begin{equation}
\label{eq:params-X+}
\alpha_X=\alpha_{+}\alpha_\Pi=\frac{8\kappa_2}{3}, \qquad a_X = a_\mH+\lceil\log M\rceil+5, \qquad \varepsilon_X \leq \alpha_{+}\bar \varepsilon_\Pi+(\alpha_\Pi+\bar\varepsilon_\Pi)\varepsilon_{+},
\end{equation}
where $\alpha_+, \varepsilon_+$ and $m_+$ are given by \Cref{eq:qsvt-Pi2}. $U_{X_s}$ can be implemented with $\mathcal{O}((m_++1)m_\Pi)$ queries to controlled $U_\mH,U_\mH^\dagger$, and $\mathcal{O}((m_++1)(m_\Pi(a_\mH+M)+M)+m_+(a_\Pi+2)+a_\Pi)$ additional one- and two-qubit gates.
\end{theorem}

\begin{proof}
Apply \Cref{lem:rectangular-be-Pia} to obtain rectangular block-encodings of $\Pi_1$, $\Pi_2$ with corresponding encoded blocks $\widetilde{\Pi}_1$ and $\widetilde{\Pi}_2$. Then, \Cref{eq:qsvt-Pi2} gives an $(\alpha_+,a_\Pi+2,a_\Pi+1,\varepsilon_+)$-block-encoding of $\Pi_2^+$.

We now multiply $\Pi_2^+$ by $-\Pi_1$. Due to the size difference of the unitaries, we pad with one ancilla
    \[
    U_{-\Pi_1}^{\mathrm{pad}}=-\mathbb{I}_2\otimes U_{\Pi_1}.
    \]
The right ancilla subspace of $U_{\Pi_2^+}$ and the left ancilla subspace of $U_{-\Pi_1}^{\mathrm{pad}}$ have the same size, so we can multiply block-encodings by \Cref{lem:product-rect-be}, resulting in the circuit
\[
    U_{X_s} =(\mathbb{I}_2\otimes U_{\Pi_2}^{+})(\mathrm{C}_{a_\Pi+1} \mathrm{NOT}\otimes \mathbb{I} )(\mathbb{I}_2\otimes U_{-\Pi_1}^{\mathrm{pad}}),
 \]
where the middle gate is a $\mathrm{CNOT}$ controlled by $a_\Pi+1$ ancilla qubits. Furthermore, $U_{X_s}$ is an $(\alpha_X, a_\Pi+3, \varepsilon_X)$-block-encoding of $X_s = -\Pi_2^+\Pi_1$. From the product lemma, $\alpha_X = \alpha_+\alpha_\Pi = 8\kappa_2/3$ and $\varepsilon_X\leq(\alpha_\Pi+\bar\varepsilon_\Pi)\varepsilon_+ +\alpha_+\bar\varepsilon_\Pi$.

The product uses one call to $U_{\Pi_2^+}$ and one to $U_{-\Pi_1}^{\mathrm{pad}}$. The pseudoinverse uses $m_+$ queries to $U_{\Pi_a}$ or $U_{\Pi_a}^\dagger$, and the multiplication contributes an additional use of $U_{\Pi_a}$. Each implementation of $U_{\Pi_a}$ itself requires $m_\Pi$ uses of $U_\mH$ and $U_\mH^\dagger$ and additional $\mathcal{O}(m_\Pi(a_\mH +M)+M)$ gates by \Cref{cor:block-encoding-riesz-smooth}. Adding the pseudoinverse QSVT overhead $\mathcal{O}(m_+(a_\Pi+2))$ and the $\mathcal{O}(a_\Pi)$ gates for $\mathrm{C}_{a_\Pi+1}\mathrm{NOT}$ yields the total gate count.
\end{proof}

The circuit $U_{X_s}$ for block-encoding $X_s$ is shown in \Cref{fig:UX-circuits-stacked}. We observe that most of the overhead comes from the repeated calls to $U_{\Pi_a}$ within QSVT.

\begin{figure}[!htbp]
\centering
\resizebox{0.55\textwidth}{!}{
\begin{tabular}{c@{\hspace{1.2em}}c@{\hspace{1.2em}}c}
\begin{quantikz}[column sep=0.5cm, row sep=0.32cm]
\lstick{$\ket{0}^{\,a_\mH+\lceil\log M\rceil+5}$} & \gate[2]{U_{X_s}} & \qw \\
\lstick{$\ket{\psi}$} & & \qw
\end{quantikz}
& $\displaystyle =$ &
\begin{quantikz}[column sep=0.36cm, row sep=0.25cm]
\lstick{$\ket{0}$} & \qw & \targ{} & \qw & \qw & \qw \\
\lstick{$\ket{0}_+$} & \qw & \octrl{-1} & \gate[4]{U_{\Pi_2^{+}}} & \qw & \qw \\
\lstick{$\ket{0}^{\,a_\Pi}$} & \gate[3]{U_{-\Pi_1}} & \octrl{-2} & & \qw & \qw \\
\lstick{$\ket{0}_c$} & & \qw & & \gate{X} & \qw \\
\lstick{$\ket{\psi}$} & & \qw & & \qw & \qw
\end{quantikz}
\end{tabular}}

\vspace{0.6em}
\makebox[\textwidth][c]{\(\displaystyle =\)\hspace{0.25em}
\resizebox{0.75\textwidth}{!}{
\begin{quantikz}[column sep=0.105cm, row sep=0.18cm]
\lstick{$\ket{0}$} & \qw & \targ{} & \qw & \qw & \qw & \qw & \qw 
& \qw & \qw & \qw & \qw & \qw & \qw & \qw & \qw & \qw & \qw & \qw & 
\cdots & \qw & \qw & \qw & \qw & \qw & \qw & \qw \\
\lstick{$\ket{0}_+$} & \qw & \octrl{-1} & \qw & \gate{H}\gategroup[4,steps=21,style={dashed,draw=green!50!black,fill=green!10,rounded corners=2pt,inner xsep=0pt,inner ysep=1pt},background,label style={label position=below,anchor=south,yshift=-0.8cm}]{$U_{\Pi_2^{+}}$} & \targ{} & \gate{e^{-i\phi_0^{(+)} Z}} & \targ{} & \qw & \qw & \qw & \qw & \targ{} & \gate{e^{-i\phi_1^{(+)} Z}} & \targ{} & \qw & \qw & \qw & \qw & \cdots & \qw & \targ{} & \gate{e^{-i\phi_{m_+}^{(+)} Z}} & \targ{} & \gate{H} & \qw & \qw \\
\lstick{$\ket{0}^{\,a_\Pi}$} & \gate[3]{U_{-\Pi_1}} & \octrl{-2} & \qw & \qw & \octrl{-1} & \qw & \octrl{-1} & \qw & \qw & \gate[3]{U_{\Pi_a}} & \qw & \octrl{-1} & \qw & \octrl{-1} & \qw & \gate[3]{U_{\Pi_a}^{\dagger}} & \qw & \qw & \cdots & \qw & \octrl{-1} & \qw & \octrl{-1} & \qw & \qw & \qw \\
\lstick{$\ket{0}_c$} & & \qw & \qw & \qw & \octrl{-2} & \qw & \octrl{-2} & \qw & \gate{X} & & \qw & \octrl{-2} & \qw & \octrl{-2} & \qw & & \gate{X} & \qw & \cdots & \qw & \octrl{-2} & \qw & \octrl{-2} & \qw & \gate{X} & \qw \\
\lstick{$\ket{\psi}$} & & \qw & \qw & \qw & \qw & \qw & \qw & \qw & \qw & \qw & & \qw & \qw & \qw & \qw & \qw & & \qw & \qw & \cdots & \qw & \qw & \qw & \qw & \qw
\end{quantikz}}}
\caption{Circuit representations of $U_{X_s}$, the $(\alpha_X, a_X=a_\mH+\lceil\log M\rceil+5, \varepsilon_X)$-block-encoding of the CARE solution $X_s$. \textbf{Top (right):} The block-encoding product yielding $\Pi_2^+ \Pi_1$. \textbf{Bottom:} The pseudoinverse QSVT block for building $U_{\Pi_2^+}$ unfolded. The $X$ gates implement column extraction as $U_{\Pi_2}=U_{\Pi_a}
\left(\mathbb{I}_{2^{a_\Pi}}\otimes X\otimes \mathbb{I}_n\right)$, and each $U_{\Pi_a}$ call expands into \Cref{fig:UPia-circuits-stacked}.}
    \label{fig:UX-circuits-stacked}
\end{figure}
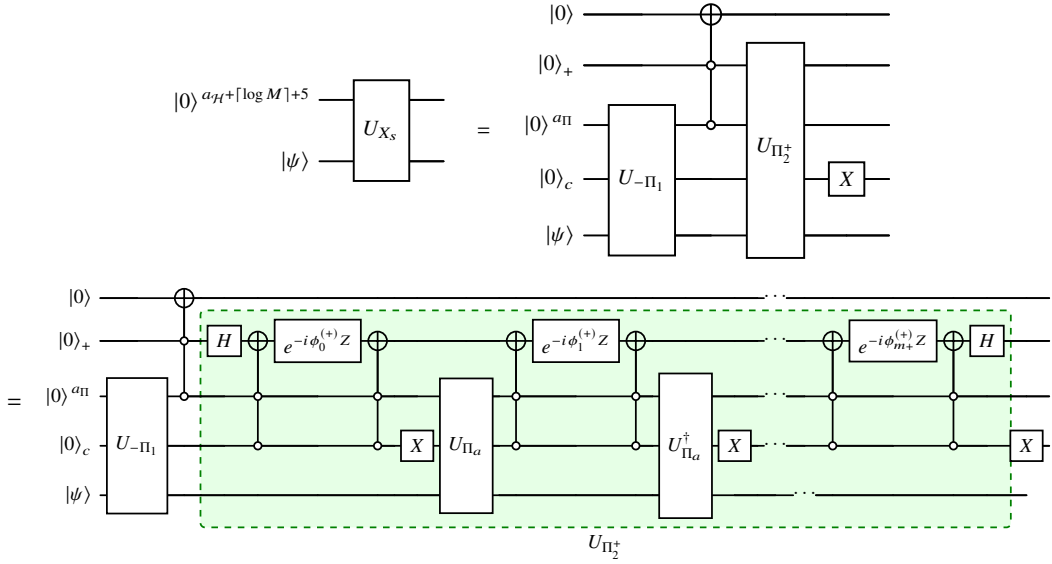

We close this section with a bound on $\kappa_2$ that suffices to obey $\kappa_2 \geq \alpha_\Pi / \sigma_{\min}(\Pi_2)$, in terms of the size of the stabilizing CARE solution $X_s$. This is particularly useful when we know properties of the solution \emph{a priori}, for example, from physical principles.

\begin{theorem}\label{thm:kappa2_bound}
    Let $X_s \in \mathbb{C}^{n \times n}$ be the solution to $\Pi_1 + \Pi_2 X_s = 0$, where $\Pi_a=\begin{pmatrix} \Pi_1 & \Pi_2\end{pmatrix}$ is a projector and $\Pi_2$ has full column rank. Then there exists a choice of $\kappa_2$ such that
    \begin{equation}
    \frac{\alpha_\Pi}{\sigma_{\min}(\Pi_2)} \leq \kappa_2 \leq \alpha_\Pi \sqrt{1 + \|X_s\|^2}.
    \label{eq:kappa2_bound}
    \end{equation}
\end{theorem}

\begin{proof}
    First, write
    \[
    \Pi_a = \begin{pmatrix}
        A & B\\
        C & D
    \end{pmatrix}, \qquad \text{where } \Pi_2 = \begin{pmatrix}
        B\\
        D
    \end{pmatrix}.
    \]
    The stabilizing CARE solution obeys $A + BX_s = C + DX_s = 0$. Now, using the fact that $\Pi_a$ is a projector, $\Pi_a^2 = \Pi_a$ implies
    \[
    \begin{pmatrix}
        AB + BD\\
        CB + D^2
    \end{pmatrix} =
    \begin{pmatrix}
        B\\
        D
    \end{pmatrix}.
    \]
    Substituting $A=-BX_s$ and $C=-DX_s$ yields
     \[
     B(D-X_sB-\mathbb{I})=0,\qquad
       D(D-X_sB-\mathbb{I})=0.
     \]
     Therefore,
     \[
     \begin{pmatrix}
     B\\
     D
     \end{pmatrix}
     (D-X_sB-\mathbb{I})=0.
     \]
     Since $\Pi_2=\begin{pmatrix}B\\D\end{pmatrix}$ has full column rank, its kernel is the zero vector, and the equation above has unique solution $D = \mathbb{I} + X_s B$. Equivalently, we can write
    \[
    \begin{pmatrix}
        -X_s & \mathbb{I}
    \end{pmatrix} \Pi_2 = \begin{pmatrix}
        -X_s & \mathbb{I}
    \end{pmatrix}
    \begin{pmatrix}
        B\\
        D
    \end{pmatrix} = \mathbb{I}.
    \]
    For any nonzero vector $u \in \mathbb{C}^n$, observe that
    \[
    \|u\|_2 = \|\mathbb{I} \cdot u\|_2 = \left\| \begin{pmatrix}
        -X_s & \mathbb{I}
    \end{pmatrix} \Pi_2 u \right\|_2 \leq \left\| \begin{pmatrix}
        -X_s & \mathbb{I}
    \end{pmatrix} \right\|  \|\Pi_2 u\|_2.
    \]
    By the variational characterization of singular values,
    \[
    \sigma_{\min}(\Pi_2) = \min_{u \in \mathbb{C}^n \setminus \{0\}} \frac{\|\Pi_2 u\|_2}{\|u\|_2} \geq \frac{1}{\left\| \begin{pmatrix}
        -X_s & \mathbb{I}
    \end{pmatrix} \right\|} = \frac{1}{\sqrt{1 + \|X_s\|^2}}.
    \]
    Thus ${\alpha_\Pi}/{\sigma_{\min}(\Pi_2)}\leq
\alpha_\Pi\sqrt{1+\|X_s\|^2}$, and there exists a number  $\kappa_2$ satisfying \Cref{eq:kappa2_bound}.  \qedhere
\end{proof}

\section{Ordinary RPA and \texorpdfstring{$m$-RPA}{m-RPA}}
\label{app:appD}
This appendix fixes notation and derives the $m$-RPA equations from the main text via the equation-of-motion (EOM) method.

\subsection{Ordinary RPA/rCCD}
We use indices $i,j,k,l$ for occupied orbitals and $a,b,c,d$ for virtual orbitals, obtained from the Hartree--Fock (HF) reference $\ket{\mathrm{HF}}$. Indices $p,q,r,s$ denote arbitrary orbitals.

The molecular electronic Hamiltonian in second quantization is given by
\begin{equation}
\hat{H} = \sum_{pq} \braket{p|h|q} \, a_p^{\dagger} a_q
+ \frac{1}{4} \sum_{pqrs} \braket{pq\|rs}
a_p^{\dagger} a_q^{\dagger} a_s a_r,
\label{eq:electronic-hamiltonian}
\end{equation}
where $\braket{p|h|q}$ are the one-electron integrals and $\braket{pq\|rs}$ the antisymmetrized two-electron integrals.

The coupled-cluster (CC) ground-state ansatz is
\[
  \ket{\Psi} = e^{\hat{T}} \ket{\mathrm{HF}},
  \qquad
  \hat{T} = \hat{T}_1 + \hat{T}_2 + \hat{T}_3 + \cdots ,
\qquad \hat{T}_n = \frac{1}{(n!)^2} \sum_{\substack{ij\ldots \\ ab\ldots}} t_{ij\ldots}^{ab\ldots}
a_a^\dagger a_b^\dagger \cdots a_j a_i ,
\]
where $\hat{T}_n$ is the operator of $n$-fold excitations. For example, the operator $\hat{T}_1 = \sum_{i,a}t_i^a a_a^\dagger a_i$ excites an electron from orbital $i$ to $a$ with amplitude $t_i^a$. Let $\ket{\Phi_i^{a}}$, $\ket{\Phi_{ij}^{ab}}$, and so on denote single, double, up to $n$ excitation states, generated from the reference $\ket{\mathrm{HF}}$. Denote each unique excitation string by $\ket{\mu}$, defined as:
\[
\ket{\mu}\equiv \ket{\Phi_{ij\cdots}^{ab\cdots}}
=a_a^\dagger a_b^\dagger  \cdots a_j a_i \ket{\mathrm{HF}} ,\quad i>j>\cdots,a>b>\cdots.
\]
Due to the ordered labels, we are already considering an antisymmetrized fermionic space.
Projecting the Schrödinger equation for the CC ansatz onto the reference and excited determinants gives the CC energy and amplitude equations:
\[
\braket{\mathrm{HF}|e^{-\hat{T}}\hat{H}e^{\hat{T}}|\mathrm{HF}} = E,
\qquad
\braket{\mu|e^{-\hat{T}}\hat{H}e^{\hat{T}}|\mathrm{HF}} = 0 \quad\forall\,\mu.
\]

The simplest coupled-cluster approximation is coupled-cluster doubles (CCD), in which $\hat{T} = \hat{T}_2$ and the rest of $\hat{T}_i$ are zero. The CCD equations can be grouped in different ``channels''~\cite{shepherd_ccch_2014}. The ring CCD approximation (rCCD), which keeps the particle-hole ring contraction terms, is particularly relevant due to its equivalence to particle-hole RPA and its relationship with the Riccati equation~\cite{scuseria_drccd_2008}. Defining composite indices $(i,a)$ and $(j,b)$, the matrices appearing in the ordinary rCCD/RPA Riccati equation are
\begin{equation}
A_{ia,jb}=(\epsilon_a-\epsilon_i)\delta_{ij}\delta_{ab}
+\braket{ib||aj},\qquad B_{ia,jb}
=\braket{ab||ij}, \qquad T_{ia, jb} = t^{ab}_{ij}.
\end{equation}
Assuming real orbitals, the matrices are real and the equation reduces to the algebraic Riccati equation
\begin{equation}
B + TA + AT +TBT = 0.
    \label{eq:ringccd-are}
\end{equation}
The matrices are of dimension $n \times n$ with $n = n_o n_v$, where $n_o$ is the number of occupied orbitals and $n_v$ the number of virtual orbitals. The invariant subspace formulation of \Cref{eq:ringccd-are} gives the RPA equation.

\subsection{Truncation of the excitation space}
Standard RPA considers 1 particle and 1 hole configurations ($1p1h$). The RPA equation can be generalized to extend the excitation space. We denote these $m$-particle--$m$-hole ($mpmh$) generalizations by $m$-RPA or higher RPA. Here we build a higher-RPA scheme via the usual equation-of-motion (EOM) approach~\cite{rowe_eom_1968, yannouleas_srpa_1987}. Let $\ket{\Psi}$ be the many-body ground state, $\hat{H}\ket{\Psi}=E_0 \ket{\Psi}$, and $\ket{\nu}$ an excited state, $\hat{H}\ket{\nu}=E_\nu \ket{\nu}$. In this framework, excitation operators $Q_\nu^\dagger$ act on $\ket{\Psi}$ to generate the excited states, $Q_\nu^\dagger \ket{\Psi} = \ket{\nu}$ and $Q_\nu \ket{\Psi}= 0$. The EOM method starts from the equation~\cite{rowe_eom_1968}:
\[
\bra{\Psi}[O,[\hat{H},Q_\nu^\dagger]]\ket{\Psi}=\omega_\nu \bra{\Psi}[O,Q_\nu^\dagger]\ket{\Psi}.
\]
In the first approximation, the excitation operator $Q_\nu^\dagger$ is expanded in a basis of particle-hole excitations up to rank $m$, limiting the problem to the $m$-excitation manifold. We define this basis with respect to a Hartree--Fock reference state, with $i_k$ denoting hole indices and $a_k$ particle indices. For each excitation rank $\alpha\in\{1,\dots,m\}$ and collective index $\mu_\alpha =(i_1...i_\alpha; a_1...a_\alpha)$,
\begin{equation}
    Q_\nu^\dagger = \sum_{\alpha = 1}^m\sum_{\mu_\alpha} X^{(\alpha)}_{\mu_\alpha}(\nu) K^{(\alpha)\dagger}_{\mu_\alpha}-Y^{(\alpha)}_{\mu_\alpha}(\nu)K^{(\alpha)}_{\mu_\alpha},
    \label{eq:Qnu-truncated}
\end{equation}
where $K^{(\alpha)\dagger}_{\mu_\alpha} =a^\dagger_{a_1}... a^\dagger_{a_\alpha} a_{i_\alpha}...a_{i_1}$ creates $\alpha$ particle-hole pairs, and $K^{(\alpha)}_{\mu_\alpha} =a^\dagger_{i_1}... a^\dagger_{i_\alpha} a_{a_\alpha}...a_{a_1}$ annihilates them.

Substituting \Cref{eq:Qnu-truncated} into the EOM equation and choosing $O=K^{(\alpha)}_{\mu_\alpha}$ and $O=K^{(\alpha)\dagger}_{\mu_\alpha}$, we obtain
\begin{align*}
\bra{\Psi}[K^{(\alpha)}_{\mu_\alpha},[\hat{H},Q_\nu^\dagger]]\ket{\Psi}
&=
\sum_{\beta=1}^m \sum_{\nu_\beta}
\bra{\Psi}[K^{(\alpha)}_{\mu_\alpha},[\hat{H},K^{(\beta)\dagger}_{\nu_\beta}]]\ket{\Psi}
X^{(\beta)}_{\nu_\beta}(\nu) \\
&\quad
-\sum_{\beta=1}^m \sum_{\nu_\beta}
\bra{\Psi}[K^{(\alpha)}_{\mu_\alpha},[\hat{H},K^{(\beta)}_{\nu_\beta}]]\ket{\Psi}
Y^{(\beta)}_{\nu_\beta}(\nu), \\
\bra{\Psi}[K^{(\alpha)\dagger}_{\mu_\alpha},[\hat{H},Q_\nu^\dagger]]\ket{\Psi}
&=
\sum_{\beta=1}^m \sum_{\nu_\beta}
\bra{\Psi}[K^{(\alpha)\dagger}_{\mu_\alpha},[\hat{H},K^{(\beta)\dagger}_{\nu_\beta}]]\ket{\Psi}
X^{(\beta)}_{\nu_\beta}(\nu) \\
&\quad
-\sum_{\beta=1}^m \sum_{\nu_\beta}
\bra{\Psi}[K^{(\alpha)\dagger}_{\mu_\alpha},[\hat{H},K^{(\beta)}_{\nu_\beta}]]\ket{\Psi}
Y^{(\beta)}_{\nu_\beta}(\nu).
\end{align*}
Hence, we can define
\begin{align*}
A_{\mu_\alpha,\nu_\beta}^{\alpha,\beta}
&=
\bra{\Psi}[K^{(\alpha)}_{\mu_\alpha},[\hat{H},K^{(\beta)\dagger}_{\nu_\beta}]]\ket{\Psi},
\qquad
G_{\mu_\alpha,\nu_\beta}^{\alpha,\beta}
=
\bra{\Psi}[K^{(\alpha)}_{\mu_\alpha},K^{(\beta)\dagger}_{\nu_\beta}]\ket{\Psi}, \\
B_{\mu_\alpha,\nu_\beta}^{\alpha,\beta}
&=
-\bra{\Psi}[K^{(\alpha)}_{\mu_\alpha},[\hat{H},K^{(\beta)}_{\nu_\beta}]]\ket{\Psi},
\qquad
D_{\mu_\alpha,\nu_\beta}^{\alpha,\beta}
=
-\bra{\Psi}[K^{(\alpha)}_{\mu_\alpha},K^{(\beta)}_{\nu_\beta}]\ket{\Psi}.
\end{align*}
By computing the complex conjugates of these elements using $[X,Y]^\dagger = [Y^\dagger,X^\dagger]$ and the fact that $\hat{H}$ is Hermitian, we obtain the systems of equations
\begin{align*}
        \bra{\Psi} [K^{(\alpha)\dagger}_{\mu_\alpha}, [\hat{H}, Q_{\nu}^{\dagger}]] \ket{\Psi} &= \sum_{\beta = 1}^m\sum_{\nu_\beta}-(B_{\mu_\alpha,\nu_\beta}^{\alpha,\beta})^* X^{(\beta)}_{\nu_\beta}(\nu) -  (A_{\mu_\alpha,\nu_\beta}^{\alpha,\beta})^* Y^{(\beta)}_{\nu_\beta}(\nu), \\
        \bra{\Psi} [K^{(\alpha)}_{\mu_\alpha}, [\hat{H}, Q_{\nu}^{\dagger}]] \ket{\Psi} &= \sum_{\beta = 1}^m\sum_{\nu_\beta} A_{\mu_\alpha,\nu_\beta}^{\alpha,\beta} X^{(\beta)}_{\nu_\beta}(\nu) +  B_{\mu_\alpha,\nu_\beta}^{\alpha,\beta}Y^{(\beta)}_{\nu_\beta}(\nu).
\end{align*}
And for the right side of the EOM,
\begin{align*}
        \omega_{\nu} \bra{\Psi} [K^{(\alpha)\dagger}_{\mu_\alpha}, Q_{\nu}^{\dagger}] \ket{\Psi} &=   \omega_{\nu} \sum_{\beta = 1}^m\sum_{\nu_\beta}(D_{\mu_\alpha,\nu_\beta}^{\alpha,\beta})^*  X^{(\beta)}_{\nu_\beta}(\nu)
        +(G_{\mu_\alpha,\nu_\beta}^{\alpha,\beta})^* Y^{(\beta)}_{\nu_\beta}(\nu), \\
        \omega_{\nu} \bra{\Psi} [K^{(\alpha)}_{\mu_\alpha}, Q_{\nu}^{\dagger}] \ket{\Psi} &=\omega_{\nu} \sum_{\beta = 1}^m\sum_{\nu_\beta}G_{\mu_\alpha,\nu_\beta}^{\alpha,\beta}  X^{(\beta)}_{\nu_\beta}(\nu) +D_{\mu_\alpha,\nu_\beta}^{\alpha,\beta} Y^{(\beta)}_{\nu_\beta}(\nu).
\end{align*}
As we have one such equation for every index $\alpha$, and it runs over indices $\beta$, the two systems can be stacked in the standard RPA form
\begin{equation}
\begin{pmatrix}
A & B \\
B^* & A^*
\end{pmatrix}
\begin{pmatrix}
\mathbf{X}(\nu) \\
\mathbf{Y}(\nu)
\end{pmatrix}=\omega_\nu
\begin{pmatrix}
G & D \\
-D^* & -G^*
\end{pmatrix}
\begin{pmatrix}
\mathbf{X}(\nu) \\
\mathbf{Y}(\nu)
\end{pmatrix},
\label{eq:RPA-general}
\end{equation}
where $A=[A^{\alpha,\beta}]_{\alpha,\beta=1}^m$, and similarly for $B,G,D$. This is the generalized RPA problem. The left matrix is the RPA Hamiltonian matrix and the right matrix is the metric. The amplitude vectors are stacked as
\[
\mathbf{X}(\nu)=
\begin{pmatrix}
X^{(1)}(\nu)\\
X^{(2)}(\nu)\\
\vdots\\
X^{(m)}(\nu)
\end{pmatrix},
\qquad
\mathbf{Y}(\nu)=
\begin{pmatrix}
Y^{(1)}(\nu)\\
Y^{(2)}(\nu)\\
\vdots\\
Y^{(m)}(\nu)
\end{pmatrix}.
\]

\subsection{Quasi-boson approximation}
We now evaluate the block-matrices under the QBA, in which the exact ground state is replaced by the Hartree--Fock determinant:
\[
\bra{\Psi}\cdots\ket{\Psi}\rightarrow \bra{\mathrm{HF}}\cdots\ket{\mathrm{HF}}.
\]
Under the QBA, the metric blocks are immediate. Since $K^{(\alpha)}_{\mu_\alpha}
=a^\dagger_{i_1}\cdots a^\dagger_{i_\alpha}
a_{a_\alpha}\cdots a_{a_1}$ with the orbitals chosen with respect to an HF reference, every $K^{(\alpha)}_{\mu_\alpha}$ is a de-excitation operator with respect to $\ket{\mathrm{HF}}$, and we have $K^{(\alpha)}_{\mu_\alpha}\ket{\mathrm{HF}}=0
\quad
\forall\,\alpha,\mu_\alpha$.
Therefore
\begin{align*}
D_{\mu_\alpha,\nu_\beta}^{\alpha,\beta}&=-\bra{\mathrm{HF}}[K^{(\alpha)}_{\mu_\alpha},K^{(\beta)}_{\nu_\beta}]\ket{\mathrm{HF}}
=0, \\
G_{\mu_\alpha,\nu_\beta}^{\alpha,\beta}
&=\bra{\mathrm{HF}}K^{(\alpha)}_{\mu_\alpha}K^{(\beta)\dagger}_{\nu_\beta}\ket{\mathrm{HF}}
=\braket{\mu_\alpha|\nu_\beta}
=\delta_{\alpha\beta}\delta_{\mu_\alpha\nu_\beta}.
\end{align*}
The first identity is immediate. For the second, apply Wick's theorem for expectation values with respect to the HF reference. Only complete contractions between the de-excitation string and the excitation string remain, and since the strings are ordered, contractions give the overlap above. We conclude that under the QBA, $G = \mathbb{I}$. Multiplying both sides of the eigenvalue equation by the metric, the RPA equation becomes:
\begin{equation}
    \begin{pmatrix}
        A & B \\ -B^* & -A^*
    \end{pmatrix}
    \begin{pmatrix}
        \mathbf{X}(\nu) \\
        \mathbf{Y}(\nu)
    \end{pmatrix} =
        \omega_\nu
    \begin{pmatrix}
        \mathbf{X}(\nu) \\
        \mathbf{Y}(\nu)
    \end{pmatrix}.
    \label{eq:mRPA-equation}
\end{equation}
To simplify $A,B$, we write the electronic Hamiltonian from \Cref{eq:electronic-hamiltonian} as $\hat{H}=\hat{F}+\hat{W}$, with one-body part $\hat{F}$ and two-body part $\hat{W}$. In the case of $B_{\mu_\alpha,\nu_\beta}^{\alpha,\beta} = -\bra{\mathrm{HF}} [K^{(\alpha)}_{\mu_\alpha}, [\hat{H}, K^{(\beta)}_{\nu_\beta}]] \ket{\mathrm{HF}}$, three of the commutators vanish because the particle annihilation terms act on the HF reference. The remaining term is
\[\bra{\mathrm{HF}} K^{(\alpha)}_{\mu_\alpha}K^{(\beta)}_{\nu_\beta}\hat{H}\ket{\mathrm{HF}} =\bra{\mathrm{HF}} K^{(\alpha)}_{\mu_\alpha}K^{(\beta)}_{\nu_\beta}\hat{F}\ket{\mathrm{HF}} +\bra{\mathrm{HF}} K^{(\alpha)}_{\mu_\alpha}K^{(\beta)}_{\nu_\beta}\hat{W}\ket{\mathrm{HF}}.\]
The product $K^{(\alpha)}_{\mu_\alpha}K^{(\beta)}_{\nu_\beta}$ contains $\alpha+\beta$ hole creation and particle annihilation operators. Due to Wick's theorem, all contractions between $K^{(\alpha)}_{\mu_\alpha}$ and $K^{(\beta)}_{\nu_\beta}$ are zero. There are $\alpha+\beta$ operators that have to be contracted with only the hole annihilation and particle creation operators of $\hat{H}$. Therefore, the one-body part is unable to generate nonzero contractions and is always zero. The two-body part can only supply two hole annihilation and two particle creation operators. This means that for $\alpha,\beta> 1$, all terms vanish.
The only nonzero block is $B^{1,1}$. As we have shown previously, the $a_i^\dagger a_a a_j^\dagger a_b$ operators cannot be fully contracted with the one-body Hamiltonian $\hat{F}$ without contracting at least one operator of $a_i^\dagger a_a$ with one of $a_j^\dagger a_b$, and all these contractions vanish, so that term is zero. The two-body term can be obtained by Wick's theorem, yielding
\[
B_{\mu_\alpha,\nu_\beta}^{\alpha,\beta}
=\begin{cases}
\braket{ab||ij},
& \alpha=\beta=1,\quad\mu_1=(i,a),\quad\nu_1=(j,b),\\
0, & \text{otherwise}.
\end{cases}
\]

For $A_{\mu_\alpha,\nu_\beta}^{\alpha,\beta}
=\bra{\mathrm{HF}}[K^{(\alpha)}_{\mu_\alpha},[\hat{H},K^{(\beta)\dagger}_{\nu_\beta}]]\ket{\mathrm{HF}}$, we have a different structure. Again, the two terms containing $K^{(\alpha)}_{\mu_\alpha}$ on the right vanish. We have:
\[
A_{\mu_\alpha,\nu_\beta}^{\alpha,\beta} = \bra{\mathrm{HF}}K^{(\alpha)}_{\mu_\alpha}\hat{H}K^{(\beta)\dagger}_{\nu_\beta}\ket{\mathrm{HF}} -\bra{\mathrm{HF}}K^{(\alpha)}_{\mu_\alpha}K^{(\beta)\dagger}_{\nu_\beta}\hat{H}\ket{\mathrm{HF}}.
\]
In this case, there are possible contraction paths between the $\alpha$ and $\beta$ sectors, which increases the complexity. Note that if $\alpha = \beta$, there is always a contraction between the $\alpha$ and $\beta$ sectors which will be nonzero when $\mu_\alpha = \nu_\beta$, leading to nonzero $A^{\alpha\beta}$. Similarly, for $|\alpha-\beta|\leq 2$, the Hamiltonian terms $\hat{F},\hat{W}$ will be able to supply the missing creation and annihilation operators that lead to nonzero contractions. However, for $|\alpha-\beta|> 2$ it is not possible to generate nonzero contractions, because $K^\alpha$ supplies $2\alpha$ operators, $K^\beta$ supplies $2\beta$, and the Hamiltonian supplies up to four operators. Then, $|2\alpha - 2\beta |> 4$, and there is not a full contraction path with the Hamiltonian, so there will be contractions within the $\alpha$ or the $\beta$ sector, which are zero. We have:
\[
A^{\alpha\beta}=0, \qquad |\alpha-\beta|>2.
\]
This completes the derivation of the $m$-RPA matrix elements and their block structure used in the main text.

\subsection{Riccati reduction and correlation energy}
We denote the matrix on the left side of \Cref{eq:mRPA-equation} by $\mH_{m\text{-}\mathrm{RPA}}$; its size is $2N_m\times 2N_m$, with
\begin{equation}
N_m = \sum_{\alpha=1}^m \begin{pmatrix} {n_o} \\ \alpha \end{pmatrix}\begin{pmatrix} {n_v} \\ \alpha \end{pmatrix},
\end{equation}
for $n_o$ occupied and $n_v$ virtual orbitals. We now recover the same invariant-subspace structure that appeared for rCCD/RPA. If there exists an $N_m$-dimensional invariant subspace of $\mH_{m\text{-}\mathrm{RPA}}$ spanned by the columns of $\begin{psmallmatrix} \mathbf{X}\\ \mathbf{Y} \end{psmallmatrix}$ with $\mathbf{X}$ invertible, then there is an associated Riccati equation
\[
T B^*T+ A T+ T A^*+ B=0,
\qquad
T=\mathbf{Y}\mathbf{X}^{-1}.
\]
The RPA eigenvector blocks $\mathbf{X}(\nu)$ and $\mathbf{Y}(\nu)$ play the role of the invariant-subspace blocks $U$ and $V$ in the main text (recall \Cref{eq:linearization-CARE}), with $T=\mathbf{Y}\mathbf{X}^{-1}$ being the analogue of $X=VU^{-1}$ for the CARE.
In the real-orbital case frequently considered in particle-hole RPA, this reduces to $TBT+AT+TA+B=0$. For $m=1$, we recover the usual rCCD/RPA equation, and the matrices $A,B$ agree with the ones introduced above. Our projector construction assumes that the chosen $m$-RPA matrix has the required spectral split.

The Riccati amplitude $T$ also yields a compact expression for the $m$-RPA correlation energy, the central observable in our quantum protocol, which we derive following the standard RPA/rCCD argument~\cite{scuseria_drccd_2008}. Restricting \Cref{eq:mRPA-equation} to a positive energy invariant subspace, the upper block reads
\[
A\mathbf{X}+B\mathbf{Y}=\mathbf{X}\Omega,\qquad \Omega=\mathrm{diag}(\omega_\nu)>0.
\]
Multiplying on the right by $\mathbf{X}^{-1}$ and using $T=\mathbf{Y}\mathbf{X}^{-1}$ gives $A+BT=\mathbf{X}\Omega\mathbf{X}^{-1}$, and the cyclicity of the trace gives $\Tr(BT)=\Tr(\Omega-A)$. Considering the plasmon formula $E_c^{(m)}=\tfrac{1}{4}\Tr(\Omega-A)$~\cite{ring_nmbp_1980}, with prefactor $1/4$ to match coupled-cluster scaling~\cite{scuseria_drccd_2008}, we write the correlation-energy observable
\[
E_c^{(m)}=\frac{1}{4}\Tr\!\left(BT\right).
\]
Here, with some abuse of notation, $B\equiv B^{1,1}$ and $T\equiv T_{1,1}^{(m)}$ is the $1p1h$ block of the Riccati solution. The block-sparsity of $B$ restricts the relevant components of $T$ to its $1p1h$ block, which is what we exploit in the cost analysis.

\subsection{Localization and sparsity}
We now justify the sparsity estimate used in the main-text $m$-RPA cost analysis. The key to increasing sparsity is localized orbitals. Denote the canonical HF orbitals by $\{\chi_p\}$, which are split into occupied and virtual subspaces. Then, the change to localized molecular orbitals (LMOs) can be achieved via a unitary rotation within the occupied and the virtual subspaces~\cite{helgaker_mest_2000, hoyvik_plhf_2013}:
\[
\tilde{\chi}_i = \sum_{j\in\mathrm{occ}} \chi_j (U_{\mathrm{occ}})_{ji},
\qquad
\tilde{\chi}_a = \sum_{b\in\mathrm{vir}} \chi_b (U_{\mathrm{vir}})_{ba},
\]
The fermionic operators associated with these orbitals transform as~\cite{hoyvik_plhf_2013},
\[
\tilde a_i^\dagger = \sum_{j} a_j^\dagger (U_\mathrm{occ})_{ji}, \qquad
\tilde a_a^\dagger = \sum_{b} a_b^\dagger (U_\mathrm{vir})_{ba},
\qquad
\tilde a_i = \sum_{j}  a_j (U_\mathrm{occ})_{ji}^* , \qquad
\tilde a_a = \sum_{b} a_b (U_\mathrm{vir})_{ba}^* .
\]

The unitary $U=\mathrm{diag}(U_\mathrm{occ},U_\mathrm{vir})$ can be obtained, for example, via the Pipek-Mezey or Foster-Boys localization procedures~\cite{pipek_filp_1989, foster_ccip_1960}. Under this basis change, it can be shown that the Hartree--Fock determinant remains invariant up to a global phase:
\[
\ket{\widetilde{\mathrm{HF}}}=U\ket{\mathrm{HF}}=e^{i\theta}\ket{\mathrm{HF}}.
\]
Therefore, we can see localization as a change of the one-particle basis, which will in turn affect the higher-particle basis as well. The localized excitation operators are:
\[
\tilde K^{(\alpha)\dagger}_{\mu_\alpha}
  =\sum_{\mu'_\alpha} K^{(\alpha)\dagger}_{\mu'_\alpha}\,
\mathcal{U}^{(\alpha)}_{\mu'_\alpha\mu_\alpha}=
\tilde a^\dagger_{a_1}... \tilde a^\dagger_{a_\alpha}
\tilde a_{i_\alpha}... \tilde a_{i_1},
\]
where $\mathcal{U}^{(\alpha)}$ is the induced unitary on the $\alpha p$-$\alpha h$ sector. The electronic Hamiltonian has the same expression in the new basis, but
with the new localized operators $\tilde a_q^\dagger, \tilde a_p$ and one- and two-electron integrals depending on the localized orbitals. Therefore, by this change of basis we can rewrite the excitation operators of the true ground state as:
\[
\tilde Q_\nu^\dagger
=\sum_{\alpha=1}^m \sum_{\mu_\alpha}
\tilde X_{\mu_\alpha}^{(\alpha)}(\nu)\tilde K_{\mu_\alpha}^{(\alpha)\dagger}
-\tilde Y_{\mu_\alpha}^{(\alpha)}(\nu)\tilde K_{\mu_\alpha}^{(\alpha)}.
\]
The matrices we found via the EOM formalism remain unchanged because all the calculations are equivalent in the new basis:
\begin{align*}
\tilde A_{\mu_\alpha,\nu_\beta}^{\alpha,\beta}
&=\bra{\mathrm{HF}}\left[\tilde K^{(\alpha)}_{\mu_\alpha},\left[\hat{H},\tilde K^{(\beta)\dagger}_{\nu_\beta}\right]
\right]\ket{\mathrm{HF}},
\\
\tilde B_{\mu_\alpha,\nu_\beta}^{\alpha,\beta}
&=-\bra{\mathrm{HF}}\left[\tilde K^{(\alpha)}_{\mu_\alpha},\left[\hat{H},\tilde K^{(\beta)}_{\nu_\beta}\right]
\right]
\ket{\mathrm{HF}}.
\end{align*}
Most importantly, under the QBA, the block-sparse structure of these matrices is maintained, because the equations are fully equivalent to the ones we had but in the new basis, and we can apply Wick's theorem as in the previous subsection. For similar reasons, the metric matrices remain invariant and we obtain exactly the same form for the $m$-RPA equations but with new matrix elements and solution vectors.

Following the setup in the main text, consider a system of physical volume $V$ in $D$ spatial dimensions, with $n_o,n_v=\mathcal{O}(V)$, and assume exponentially localized orbitals with correlation radius $R_c$. The two-electron integrals can be written in terms of transition densities $\rho_{pr}(\mathbf r)=\chi_p^*(\mathbf r)\chi_r(\mathbf r)$,
\[
\braket{pq|rs} = \int d\mathbf{r}_1\,d\mathbf{r}_2\,\rho_{pr}(\mathbf{r}_1)\,\frac{1}{|\mathbf{r}_1-\mathbf{r}_2|}\,\rho_{qs}(\mathbf{r}_2).
\]
For densities centered at well-separated points $\mathbf{R}_A,\mathbf{R}_B$ with $R=|\mathbf{R}_A-\mathbf{R}_B|$, the integrals admit a multipole expansion in terms of $1/R$~\cite{jackson_ce_1998}. By orthonormality, the monopole moments $Q_{pr}=\braket{p|r}=\delta_{pr}$ vanish in the off-diagonal terms, so only diagonal integrals $\braket{pq|pq}$ carry the $1/R$ Coulomb tail; off-diagonal integrals begin at charge--dipole ($1/R^2$) or dipole--dipole ($1/R^3$) contributions and decay faster. The $\mathcal{O}(V^2)$ diagonal long-range terms form the diagonal of $A$, while the remaining off-diagonal integrals are thresholded to orbital pairs within volume $R_c^D$.

\subsubsection{Sparsity via the commutator structure}
Without exploiting the commutator structure of $A$, disconnected matrix elements appear. To see why this happens, take
\[
\ket{\mu_\alpha}=a_a^\dagger a_b^\dagger a_c^\dagger a_k a_j a_i\ket{\mathrm{HF}},\qquad
 \ket{\nu_\beta}=a_a^\dagger a_i\ket{\mathrm{HF}},
\]
which differ by the two-label change $jk\mapsto bc$. Then $\braket{\mu_\alpha|\hat{F}|\nu_\beta}=0$, while the two-body term gives $\braket{\mu_\alpha|\hat{W}|\nu_\beta}=\braket{bc\|jk}$. This element is disconnected: none of the indices $j,k,b,c$ appears in $\ket{\nu_\beta}$, so each $\ket{\nu_\beta}$ couples to $\mathcal{O}(VR_c^D)$ different $\ket{\mu_\alpha}$ (with $j$ unrestricted and $k,b,c$ within the correlation volume), and the row sparsity of $A$ grows linearly with $V$.

Consider now the $A$ matrix in the localized basis:
\[
\tilde A^{\alpha,\beta}_{\mu_\alpha,\nu_\beta} =\bra{\mu_\alpha}[\hat{H},\tilde K^{(\beta)\dagger}_{\nu_\beta}]\ket{\mathrm{HF}}, \qquad
 \ket{\mu_\alpha}=\tilde K^{(\alpha)\dagger}_{\mu_\alpha}\ket{\mathrm{HF}}.
\]
Write the Hamiltonian as $\hat{H} = \hat{F} + \hat{W}_\mathrm{diag}+\hat{W}_\mathrm{off}$, where $\hat{F}$ is the one-body term, $\hat{W}_\mathrm{diag}$ contains the terms $\braket{pq|pq}$ and $\hat{W}_\mathrm{off}$ the off-diagonal terms. First, the one-body term is
 \[
    [\tilde a_p^\dagger \tilde a_q, \tilde K^{(\beta)\dagger}_{\nu_\beta}].
\]
If the indices $p,q$ are not present in $\nu_\beta$, then the fermionic operator pair commutes since $\tilde K^{(\beta)\dagger}_{\nu_\beta}$ has an even number of operators and the $2\beta$ commutations produce no net sign. In this case, the commutator is zero, and disconnected contributions vanish. When $p$ or $q$ are present in $\nu_\beta$, we can use $[AB,C] = A[B,C] + [A,C]B$ and the canonical anticommutation relation to obtain:
    \[
    [\tilde a_p^\dagger \tilde a_q, \tilde K^{(\beta)\dagger}_{\nu_\beta}] = \tilde a_p^\dagger[\tilde a_q, \tilde K^{(\beta)\dagger}_{\nu_\beta}] + [\tilde a_p^\dagger, \tilde K^{(\beta)\dagger}_{\nu_\beta}]\tilde a_q = \sum_{k=1}^\beta  \delta_{q a_k} \tilde K^{(\beta)\dagger}_{\nu_\beta, (a_k\mapsto p)} -\sum_{k=1}^\beta  \delta_{p i_k} \tilde K^{(\beta)\dagger}_{\nu_\beta, (i_k\mapsto q)},
    \]
    where $\tilde K^{(\beta)\dagger}_{\nu_\beta, (i_k\mapsto q)}$ indicates that one of the indices $i_k$ in $\nu_\beta$ changes to $q$. Note that this might change the excitation rank. Thus, the one-body term can change at most one of the indices in $\nu_\beta$. Therefore, $[\hat{F}, \tilde K^{(\beta)\dagger}_{\nu_\beta}]\ket{\mathrm{HF}}$ is a linear combination of determinants that differ by at most one index from $\ket{\nu_\beta}$. The contribution to $A$ is
    \[
    \bra{\mu_\alpha}[\hat{F},\tilde K^{(\beta)\dagger}_{\nu_\beta}]\ket{\mathrm{HF}} = \sum_{k=1}^{\beta}\sum_p \tilde h_{p,a_k}\braket{\mu_\alpha|\nu_\beta(a_k\to p)}
     -\sum_{k=1}^{\beta}\sum_q \tilde h_{i_k,q}\braket{\mu_\alpha|\nu_\beta(i_k\to q)}.
    \]
    For $D$-dimensional systems, the one-body integrals will be non-negligible in a neighborhood of $\mathcal{O}(R_c^D)$ orbitals due to localization, depending on the chosen threshold. Without localization, if we fix $\nu_\beta$, by changing a single label we can couple with $V$ different $\mu_\alpha$. As the sum goes over $\beta$ indices, possible replacements are $\beta V$. However, due to localization, the integrals are only non-negligible for $\mathcal{O}(R_c^D)$ orbitals. Therefore, in each column (index $\nu_\beta$) there are $\mathcal{O}(\beta R_c^D)$ nonzero elements.

     The diagonal part of $\hat{H}$ is $\hat{W}_{\mathrm{diag}} =\frac{1}{4}\sum_{p\ne q}\braket{\tilde p\tilde q|\tilde p\tilde q}\tilde a_p^\dagger \tilde a_q^\dagger \tilde a_q \tilde a_p$ and contributes the Coulomb tail. For this term, the relevant commutators will be $[\tilde a_p^\dagger \tilde a_q^\dagger \tilde a_q \tilde a_p, \tilde K^{(\beta)\dagger}_{\nu_\beta}]$. Define the occupation $n_p(\Phi)$ so that the occupation is one if $p$ is an occupied orbital in the determinant $\ket{\Phi}$ and zero otherwise. Then, $\tilde K^{(\beta)\dagger}_{\nu_\beta}\tilde a_p^\dagger \tilde a_q^\dagger \tilde a_q \tilde a_p\ket{\mathrm{HF}} = \tilde K^{(\beta)\dagger}_{\nu_\beta} n_p(\mathrm{HF})n_q(\mathrm{HF})\ket{\mathrm{HF}} = n_p(\mathrm{HF})n_q(\mathrm{HF}) \ket{\nu_\beta}$. On the other hand, the first term will be nonzero when $p,q$ are occupied in $\ket{\nu_\beta}$, $\tilde a_p^\dagger \tilde a_q^\dagger \tilde a_q \tilde a_p \tilde K^{(\beta)\dagger}_{\nu_\beta}\ket{\mathrm{HF}} = n_p(\nu_\beta)n_q(\nu_\beta) \ket{\nu_\beta}$. Thus, we have the following contribution to $A$,
    \[
    \bra{\mu_\alpha}[\hat{W}_\mathrm{diag},\tilde K^{(\beta)\dagger}_{\nu_\beta}]\ket{\mathrm{HF}} = \frac{1}{4}\sum_{p\ne q}\braket{\tilde p\tilde q|\tilde p\tilde q} \bra{\mu_\alpha}(n_p(\nu_\beta)n_q(\nu_\beta) - n_p(\mathrm{HF})n_q(\mathrm{HF})) \ket{\nu_\beta},
    \]
    which, besides the conditions imposed by the occupation numbers, is only nonzero when $\mu_\alpha = \nu_\beta$, i.e., the diagonal elements $\tilde A^{\alpha,\alpha}_{\mu_\alpha,\mu_\alpha}$. Therefore, this term does not increase off-diagonal sparsity.

    For the two-body off-diagonal term we can follow the same reasoning as for the one-body terms. A more involved expansion shows that each term can replace at most two labels of $\nu_\beta$. Without localization this gives $\mathcal{O}(\beta ^2 V^2)$ potential replacements. As this expansion contains the off-diagonal, fast-decaying terms in the multipole expansion, we can truncate the two-electron integrals to orbitals that lie within $\mathcal{O}(R_c^D)$ of each other. The overall column sparsity is $\mathcal{O}(\beta ^2 R_c^{2D})$. Combining the contributions and bounding $\beta\le m$, the row sparsity of $A$ after thresholding is
    \[
    s_A = \mathcal{O}(m^2R_c^{2D} + mR_c^D).
    \]
    If one further restricts the localized Hamiltonian to single-label scattering, only the linear term survives and $s_A=\mathcal{O}(m R_c^D)$, as in~\cite{chen_frqs_2025}.
    
    The $B$ block has only one nonzero sector, $\tilde B^{1,1}_{ia,jb}=\braket{\tilde a\tilde b\|\tilde i\tilde j}$. Because particle and hole orbitals are distinct, there is no monopole term, and these matrix elements decay as off-diagonal multipole terms. Truncation to local pairs gives
    \[
    s_B=\mathcal{O}(R_c^D),
    \]
    independent of $m$. Therefore the sparsity of the full localized $m$-RPA matrix is controlled by
    \begin{equation}
    s_{\mH} = \mathcal{O}(m^2 R_c^{2D} + m R_c^D).
    \label{eq:sparsity-mRPA}
    \end{equation}

\section{Sparse trace estimation}
\label{app:appE}

This appendix presents the readout procedure for obtaining the correlation energy. Using \Cref{thm:care-be}, we can block-encode the amplitudes $T$ from $m$-RPA. The goal is to estimate the trace $\Tr(BT)$, where $B$ is a sparse matrix. We first present the protocol in a problem-agnostic form, applicable for generic uses of our quantum CARE solver algorithm ($T$ being the solution to the CARE). Then, we specialize it to the correlation energy calculation for $m$-RPA.

\subsection{Sparse state encoding}

Throughout this section, let $B,T\in\mathbb{C}^{n\times n}$. We assume access to an $(\alpha_T,a,\varepsilon_T)$-block-encoding $U_T$ of
$T$, and write
\[
\widetilde T=\alpha_T(\bra{0^a}\otimes \mathbb{I}_n)U_T(\ket{0^a}\otimes \mathbb{I}_n), \qquad \|\widetilde T-T\|\leq \varepsilon_T.
\]
In our approach, the coefficient matrix $B$ is accessed through state preparation, rather than a block-encoding. This will allow us to heavily exploit the sparsity of $B$ and therefore avoid large overheads attendant to the dimension of $B$, which would be incurred in a naive strategy.

\begin{definition}[State encodings for a sparse matrix]\label{lem:B-sparse-states}
Let $B\in\mathbb{C}^{n\times n}$ be nonzero. For each row $\mu$ of $B$, define the normalization factors
\[
    \beta_\mu\coloneqq\sqrt{\sum_\nu |B_{\mu\nu}|^2}, \qquad \Lambda_B\coloneqq\sum_\mu \beta_\mu,
\]
and for each row with $\beta_\mu \neq 0$,
\[
    \ket{b_\mu}\coloneqq\frac{1}{\beta_\mu}\sum_\nu B_{\mu\nu}^{*}\ket{\nu}.
\]
Then, we define the state encodings for $B$ as the vectors on the space $\mathbb{C}^n \otimes \mathbb{C}^n$:
\begin{equation}
\ket{\chi_B}\coloneqq\sum_\mu\sqrt{\frac{\beta_\mu}{\Lambda_B}}\ket{b_\mu}\ket{\mu}, \qquad \ket{\eta_B}\coloneqq\sum_\mu\sqrt{\frac{\beta_\mu}{\Lambda_B}}\ket{\mu}\ket{\mu}.
\end{equation}
\end{definition}

Indeed, it can be checked that $\ket{\chi_B}$ and $\ket{\eta_B}$ are normalized states. The desired trace inner product is precisely encoded in the matrix element given by these states.

\begin{lemma}[Sparse trace identity]
\label{lem:sparse-trace}
Let $B,T\in\mathbb{C}^{n \times n}$ be nonzero. For the states $\ket{\chi_B}, \ket{\eta_B}$ defined in \Cref{lem:B-sparse-states},
\begin{equation}
\bra{\chi_B}(T \otimes \mathbb{I}_n)\ket{\eta_B} =\frac{\Tr(BT)}{\Lambda_B}.
 \label{eq:sparse-trace-identity}
\end{equation}
\end{lemma}
\begin{proof}
$T$ acts on basis states as $T\ket{\mu}=\sum_\nu T_{\nu\mu}\ket{\nu}$. Then, we obtain
    \begin{align*}
    \bra{\chi_B}(T \otimes \mathbb{I}_n)\ket{\eta_B}
    &=\sum_\mu \frac{\beta_\mu}{\Lambda_B}\bra{b_\mu}T\ket{\mu}  \\
    &=\sum_{\mu:\beta_\mu\neq 0}\frac{\beta_\mu}{\Lambda_B}\left(\frac{1}{\beta_\mu}\sum_\nu B_{\mu\nu}T_{\nu\mu}\right)   \\
    &=\frac{1}{\Lambda_B}\sum_{\mu,\nu} B_{\mu\nu}T_{\nu\mu} =\frac{\Tr(BT)}{\Lambda_B}. \qedhere
    \end{align*}
\end{proof}

\subsection{Base overlap algorithm via Hadamard test}

An approximation to the trace can therefore be recovered from the block-encoding overlap,
\begin{equation}
        A_{B,T}\coloneqq(\bra{0^a}\otimes\bra{\chi_B})(U_T \otimes \mathbb{I}_n)(\ket{0^a}\otimes\ket{\eta_B}) =\frac{\Tr(B\widetilde T)}{\alpha_T\Lambda_B},
        \label{eq:AB-overlap}
\end{equation}
up to the normalization factor $\alpha_T\Lambda_B$. Hence, estimating $A_{B,T}$ to additive error $\frac{\epsilon}{\alpha_T\Lambda_B}$ gives an $\epsilon$-precise estimate of $\Tr(BT)$.

We can estimate this overlap using the Hadamard test, which estimates quantities of the form $\Re[\bra{\phi}U\ket{\psi}]$. Since the trace may have real and imaginary components, a standard phase shift allows us to access the imaginary part as well.

\begin{algorithm}[H]
\caption{Phase-shifted Hadamard test}
\label{alg:hadamard-test}
\KwIn{State preparation circuits $\ket{\phi}=P_\phi\ket{0}, \ket{\psi}=P_\psi\ket{0}$; a unitary $U$; a phase $\theta\in\{0,-\pi/2\}$; number of samples $N$.}
\KwOut{Estimate $\widehat{\mu}$ for $\Re[e^{i\theta}\bra{\phi}U\ket{\psi}]$.}
\BlankLine
\For{$t=1$ \KwTo $N$}{
Prepare the control register in $\frac{1}{\sqrt{2}}(\ket{0}+e^{i\theta}\ket{1})$.\\
Apply $P_\phi$ to the work register controlled on the control qubit being in state $\ket{0}$.\\
Apply $P_\psi$ and then $U$ to the work register controlled on the control qubit being in state $\ket{1}$.\\
Apply a Hadamard gate to the control qubit and measure it in the computational basis. \\
Assign values $X_t = +1, -1$ to the outcomes $\ket{0}, \ket{1}$, respectively.}
\Return{$\widehat{\mu}\gets \frac{1}{N}\sum_{t=1}^{N} X_t$.}
\end{algorithm}

\begin{lemma}
    Let $\gamma, p_f \in (0, 1)$. The output of \Cref{alg:hadamard-test} obeys
    \[
    \Pr\mathopen{}\left(|\widehat{\mu} - \Re[e^{i\theta} \bra{\phi}U\ket{\psi}]| \geq \gamma\right) \leq p_f,
    \]
    provided that the number of repetitions $N \geq \frac{C}{\gamma^2} \log\frac{1}{p_f}$ ($C > 0$ is some universal constant).
\end{lemma}

\begin{proof}
Before the measurement in \Cref{alg:hadamard-test}, the probabilities of measuring the control qubit in $\ket{0}$ and $\ket{1}$ are, respectively,
\[
p_0(\theta) = \frac{1+\Re[e^{i\theta}\bra{\phi}U\ket{\psi}]}{2}, \qquad p_1(\theta) = \frac{1-\Re[e^{i\theta}\bra{\phi}U\ket{\psi}]}{2}.
\]
Then, the random variable $X_t\in\{-1,1\}$ satisfies
\[
\mathbb{E}[X_t] = 2p_0(\theta) - 1 = \Re[e^{i\theta}\bra{\phi}U\ket{\psi}].
\]
Choosing $\theta = 0$, we estimate $\Re[\bra{\phi}U\ket{\psi}]$, and choosing $\theta = -\pi/2$, we estimate $\Re[-i\bra{\phi}U\ket{\psi}]= \Im[\bra{\phi}U\ket{\psi}]$. Therefore, we can recover an estimator for $\bra{\phi}U\ket{\psi}$ by running the Hadamard test twice and summing the components. Hoeffding's inequality~\cite{hoeffding1963probability} furnishes the advertised sample complexity.
\end{proof}

This scaling with $\epsilon$ can be quadratically improved using amplitude estimation~\cite{brassard_qaae_2002, aaronson_qacs_2020}. We now give a full description of this process.

\subsection{Enhanced overlap algorithm via amplitude estimation}

\begin{theorem}[Sparse trace estimation from a block-encoding]
\label{thm:sparse-trace-be}
Let $B,T\in\mathbb{C}^{n\times n}$ be nonzero, $\epsilon > 0$, and let $U_T$ be an $(\alpha_T, a, \varepsilon_T)$-block-encoding of $T$ where $\varepsilon_T\leq \frac{\epsilon}{2\Lambda_B}$. Suppose the states $\ket{\eta_B}$ and $\ket{\chi_B}$ from \Cref{lem:B-sparse-states} are prepared by quantum circuits $P_{\eta_B}$ and $P_{\chi_B}$, respectively. Then there is an algorithm that produces an estimate $\widehat{\tau}$ such that, with probability at least $1-p_f$,
\begin{equation}
\left|\widehat{\tau}-\Tr(BT) \right|\leq\epsilon.
\label{eq:trace-error}
\end{equation}
This algorithm uses
\begin{equation}
Q_{\mathrm{AE}}=\mathcal{O}\mathopen{}\left(\frac{\alpha_T\Lambda_B}{\epsilon}\log\frac{1}{p_f}\right)
\label{eq:energy-AE-query}
\end{equation}
calls to controlled $U_T$, $P_{\eta_B}$, $P_{\chi_B}$, and their inverses.
\end{theorem}

\begin{proof}
The algorithm applies amplitude estimation to the circuit that implements the Hadamard test. First, let us bound the systematic error arising due to the block-encoding approximation. \Cref{alg:hadamard-test} with $\ket{\phi} = \ket{0^a}\ket{\chi_B}$, $\ket{\psi} = \ket{0^a}\ket{\eta_B}$, and $U = U_T \otimes \mathbb{I}_n$ estimates the overlap $A_{B,T}$ as in \Cref{eq:AB-overlap}. Since $A_{B,T} = \frac{\Tr(B\widetilde{T})}{\alpha_T\Lambda_B}$, the block-encoding error gives
\begin{align*}
\left|\alpha_T\Lambda_B A_{B,T}-\Tr(BT)\right|
&=\left|\Tr\bigl(B(\widetilde T-T)\bigr)\right|\\
&=\Lambda_B\left|
\bra{\chi_B}((\widetilde T-T) \otimes \mathbb{I}_n)\ket{\eta_B}
\right|\\
&\leq \Lambda_B\|\widetilde T-T\|
\leq \Lambda_B\varepsilon_T .
\end{align*}

Next, we establish the correctness of amplitude estimation. For a phase $\theta$, let $W_\theta$ denote the circuit of \Cref{alg:hadamard-test} (omitting the final measurement). This is a unitary satisfying
\[
W_\theta \ket{0}\ket{0^a}\ket{0^{\lceil \log n \rceil}} = \lambda_\theta \ket{0} \ket{\phi_{\mathrm{good}}} + \sqrt{1 - \lambda_\theta^2} \ket{1} \ket{\phi_{\mathrm{bad}}},
\]
where
\[
\ket{\phi_{\mathrm{good}}} \propto \ket{0^a} \ket{\chi_B} + e^{i\theta} \ket{0^a} \frac{\widetilde{T}\otimes\mathbb{I}_n}{\alpha_T} \ket{\eta_B} + \ket{{\perp}}, \qquad \lambda_\theta = \sqrt{\frac{1 + \Re[e^{i\theta} A_{B,T}]}{2}} \equiv \sqrt{p_0(\theta)},
\]
and $\ket{{\perp}}$ is some unnormalized state orthogonal to $\ket{0^a}$. Thus we can apply quantum amplitude estimation~\cite{brassard_qaae_2002,aaronson_qacs_2020} to estimate $\lambda_\theta$, and hence recover an estimate of $\Re[e^{i\theta} A_{B,T}]$. We follow the protocol of Aaronson and Rall~\cite{aaronson_qacs_2020}, which does not require the controlled versions of $W_\theta$ or the large auxiliary QFT circuits that Ref.~\cite{brassard_qaae_2002} uses. Theorem 3 of Ref.~\cite{aaronson_qacs_2020} states that an estimate $\widehat{\lambda}_\theta \in [0, 1]$ obeying $\Pr(|\widehat{\lambda}_\theta - \lambda_\theta| \leq \gamma) \geq 1 - p_f$ can be obtained from $Q_{\mathrm{AE}} = \mathcal{O}(\frac{1}{\gamma} \log\frac{1}{p_f})$ queries to $W_\theta$ and $W_\theta^\dagger$. Then, defining $\widehat{A}_{B,T}^{R} \coloneqq 2\widehat{\lambda}_0^2-1$ as an estimate for the real part of the overlap $A_{B,T}$ (similarly for the imaginary part $\widehat{A}_{B,T}^{I} \coloneqq 2\widehat{\lambda}_{-\pi/2}^2-1$), we have
\[
|\widehat{A}_{B,T}^{R} - \Re[A_{B,T}] | = 2 |\widehat{\lambda}_0^2 - \lambda_0^2| = 2 |\widehat{\lambda}_0 - \lambda_0| \cdot |\widehat{\lambda}_0 + \lambda_0| \leq 4\gamma.
\]
Combining the real and imaginary components $\widehat{A}_{B,T} \coloneqq \widehat{A}_{B,T}^R + i\widehat{A}_{B,T}^I$ gives an estimate of $A_{B,T}$ with additive error $8\gamma$.

Finally, we determine the amplitude-estimation error $\gamma$ sufficient to guarantee the desired trace error $\epsilon$. Define $\widehat{\tau} \coloneqq \alpha_T \Lambda_B \widehat{A}_{B,T}$, which estimates the non-normalized trace $\Tr(BT)$. The error propagates straightforwardly:
\[
\begin{aligned}
|\widehat{\tau} -\Tr(BT)| &= |\alpha_T \Lambda_B \widehat{A}_{B,T} - \Tr(BT)|\\
&\leq |\alpha_T \Lambda_B (\widehat{A}_{B,T} - A_{B,T})| + |\alpha_T \Lambda_B A_{B,T} - \Tr(BT)|\\
&\leq \alpha_T \Lambda_B \cdot 8\gamma + \Lambda_B \varepsilon_T.
\end{aligned}
\]
Take $\varepsilon_T\leq \frac{\epsilon}{2\Lambda_B}$ as in the theorem statement, and choose $\gamma = \frac{\epsilon}{16 \alpha_T \Lambda_B}$. This bounds the right-hand side by at most $\epsilon$, as desired. The cost of amplitude estimation is therefore
\begin{equation}
    Q_{\mathrm{AE}}= \mathcal{O}\mathopen{}\left(\frac{1}{\gamma}\log\frac{1}{p_f}\right)
    = \mathcal{O}\mathopen{}\left(\frac{\alpha_T\Lambda_B}{\epsilon}\log\frac{1}{p_f}\right)
    \label{eq:AE-queries}
\end{equation}
queries to $W_\theta$ and $W_\theta^\dagger$. Consequently, the circuit for $W_\theta$ is constructed from one controlled application of $U_T$, $P_{\eta_B}$, and $P_{\chi_B}$ each.
\end{proof}
The circuit for the full estimation protocol is provided in \Cref{fig:trace-est-stacked}.
\begin{figure}[!ht]
    \centering
    \resizebox{0.78\textwidth}{!}{%
    \begin{tabular}{c@{\hspace{0.4em}}c@{\hspace{0.4em}}c}
    \begin{quantikz}[column sep=0.42cm, row sep={0.75cm,between origins}]
        \lstick{$\ket{0}^{\,a_T+2\log n+2}$}
            & \gate{\textsc{TrEst}(U_T)}
            & \meter{}
    \end{quantikz}
    &
    $\displaystyle =$
    &
    \begin{quantikz}[column sep=0.32cm, row sep={0.75cm,between origins}]
        \lstick{$\ket{0}$}
            & \gate{R}
            & \gate[5]{G^{(r-1)/2}}
            & \meter{}\rstick{$\frac{1}{\alpha_T\Lambda_B}\Tr(BT)$} \\
        \lstick{$\ket{0}$}
            & \gate[4]{W_\theta}
            &
            & \meter{} \\
        \lstick{$\ket{0}^{\,a_T}$} & & & \qw \\
        \lstick{$\ket{0}^{\,\log n}$} & & & \qw \\
        \lstick{$\ket{0}^{\,\log n}$} & & & \qw
    \end{quantikz}
    \end{tabular}
    }

    \vspace{0.8em}
    \resizebox{\textwidth}{!}{%
    \begin{tabular}{c@{\hspace{0.3em}}c@{\hspace{1.6em}}c@{\hspace{0.3em}}c}
    $\displaystyle G\;\equiv$
    &
    \begin{quantikz}[column sep=0.22cm, row sep={0.75cm,between origins}]
        \lstick{$\ket{0}$}
            & \gate{X}\gategroup[2,steps=3,style={dashed,draw=blue!60!black,fill=blue!8,rounded corners=2pt,inner xsep=1pt,inner ysep=1pt},background,label style={label position=above,anchor=south,yshift=0.04cm}]{$\mathbb{I}-2\mathbb{I}_n\otimes\ket{00}\!\bra{00}$}
            & \gate{Z}
            & \gate{X}
            & \gate{R^\dagger}
            & \gate{X}\gategroup[5,steps=3,style={dashed,draw=red!60!black,fill=red!8,rounded corners=2pt,inner xsep=1pt,inner ysep=1pt},background,label style={label position=above,anchor=south,yshift=0.04cm}]{$\mathbb{I}-2\ket{0^{n+2}}\!\bra{0^{n+2}}$}
            & \gate{Z}
            & \gate{X}
            & \gate{R}
            & \qw \\
        \lstick{$\ket{0}$}
            & \gate{X}
            & \ctrl{-1}
            & \gate{X}
            & \gate[4]{W_\theta^\dagger}
            & \gate{X}
            & \ctrl{-1}
            & \gate{X}
            & \gate[4]{W_\theta}
            & \qw \\
        \lstick{$\ket{0}^{\,a_T}$}
            & \qw & \qw & \qw &
            & \gate{X}
            & \ctrl{-1}
            & \gate{X}
            &
            & \qw \\
        \lstick{$\ket{0}^{\,\log n}$}
            & \qw & \qw & \qw &
            & \gate{X}
            & \ctrl{-1}
            & \gate{X}
            &
            & \qw \\
        \lstick{$\ket{0}^{\,\log n}$}
            & \qw & \qw & \qw &
            & \gate{X}
            & \ctrl{-1}
            & \gate{X}
            &
            & \qw
    \end{quantikz}
    &
    $\displaystyle W_\theta\;\equiv$
    &
    \begin{quantikz}[column sep=0.24cm, row sep={0.75cm,between origins}]
        \lstick{$\ket{0}$}
            & \gate{H} & \gate{S_\theta} & \octrl{2} & \ctrl{2} & \ctrl{1} & \gate{H}
            & \qw \\
        \lstick{$\ket{0}^{\,a_T}$}
            & \qw & \qw & \qw & \qw & \gate[2]{U_T} & \qw
            & \qw \\
        \lstick{$\ket{0}^{\,\log n}$}
            & \qw & \qw & \gate[2]{P_{\chi_B}} & \gate[2]{P_{\eta_B}} & & \qw
            & \qw \\
        \lstick{$\ket{0}^{\,\log n}$}
            & \qw & \qw & & & \qw & \qw
            & \qw
    \end{quantikz}
    \end{tabular}
    }
    \caption{Trace-estimation pipeline for $\Tr(BT)$ using the QFT-free amplitude estimation protocol of Aaronson and Rall~\cite{aaronson_qacs_2020} with the Hadamard test circuit. \textbf{Top:} One shot of the trace estimator circuit. The algorithm runs the circuit repeatedly for some odd integers $r$ that are chosen adaptively. The outcomes are classically postprocessed to obtain the estimate for $\Tr(BT)/(\alpha_T\Lambda_B)$ (see \cite{aaronson_qacs_2020} for details). \textbf{Bottom (left):} The Grover operator $G=(W_\theta\otimes R)\,(\mathbb{I}-2\ket{0^{n+2}}\!\bra{0^{n+2}})\,(W_\theta\otimes R)^\dagger\,(\mathbb{I}-2\mathbb{I}_n\otimes\ket{00}\!\bra{00})$, where $R$ is a $Y$-rotation gate about a particular angle. \textbf{Bottom (right):} The Hadamard-test unitary $W_\theta$. The matrix $B$ enters only through the sparse state-preparation circuits $P_{\chi_B}$ and $P_{\eta_B}$ (\Cref{lem:B-sparse-states}); the phase gate $S_\theta=\mathrm{diag}(1,e^{i\theta})$ with $\theta\in\{0,-\pi/2\}$ selects whether the real or imaginary part of the overlap is accessed.}
    \label{fig:trace-est-stacked}
\end{figure}
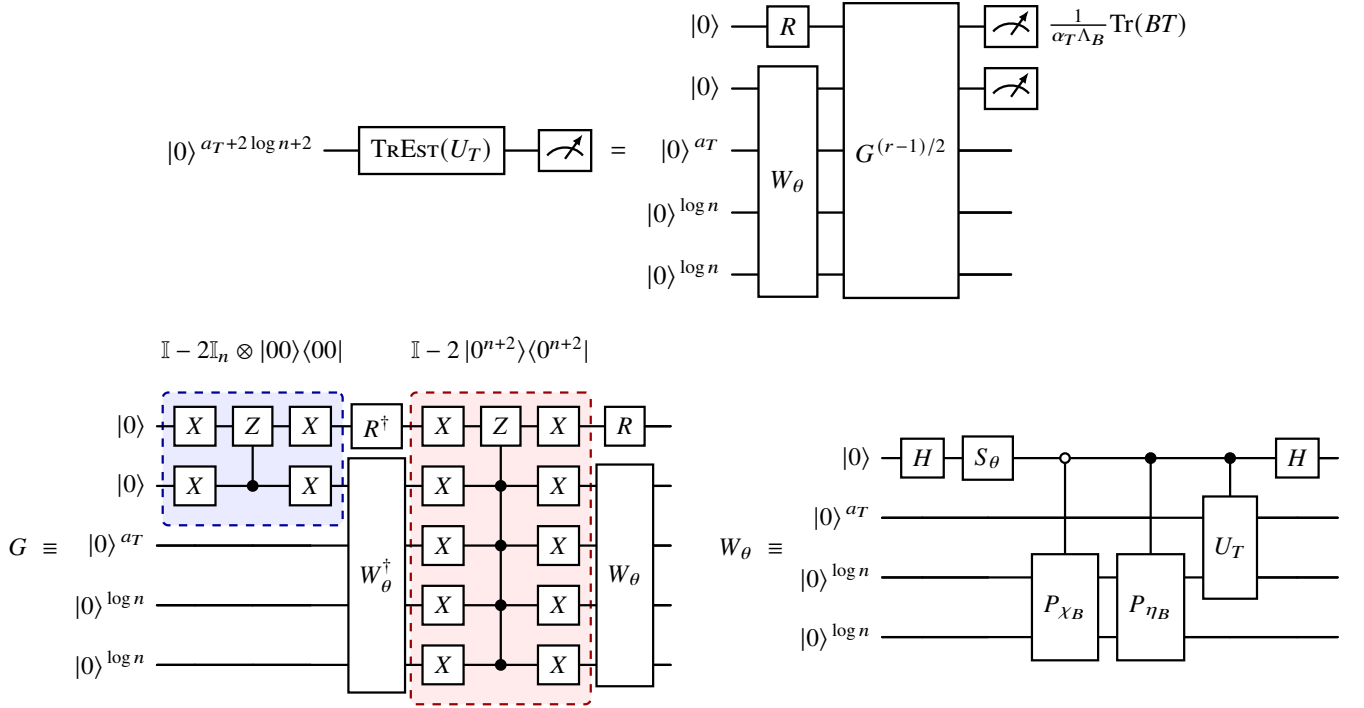

\subsection{Estimation cost under generic sparsity}

The block-encoding normalization $\alpha_T$ depends on the complexity of solving the CARE (see $\alpha_{X}$ in \Cref{thm:care-be}). Here we assess the size of the other factor, $\Lambda_B$. Let $s_B$ be the row sparsity of $B$, $r_B$ the number of nonzero rows, and $N_B$ the total number of nonzero entries. Then $N_B\leq r_B s_B$, and $\Lambda_B$ can be bounded as
\begin{equation}
\label{eq:lambda-B-bound}
\Lambda_B = \sum_\mu \beta_\mu \leq \|B\|_\mathrm{max} \sum_{\mu:\beta_\mu>0}\sqrt{s_\mu}\leq \|B\|_\mathrm{max}r_B\sqrt{s_B}.
\end{equation}
Above, $s_\mu$ is the sparsity of row $\mu$.

Now we analyze the complexity of the state-preparation circuits $P_{\chi_B}$ and $P_{\eta_B}$. The state $\ket{\chi_B}$ inherits the sparsity pattern of $B$:
\[
    \ket{\chi_B} = \sum_{\mu:\beta_\mu\neq 0}\sum_{\nu:B_{\mu\nu}\neq 0} {\frac{B_{\mu\nu}^*}{\sqrt{\beta_\mu \Lambda_B}}} \,  \ket{\mu}\ket{\nu},
\]
and therefore has vector-sparsity $N_B$. On the other hand, $\ket{\eta_B}$ is $r_B$-sparse, because its sparsity depends on how many rows $\mu$ are nonzero. Using explicit sparse state preparation, a $k$-sparse state on a $\lceil\log n\rceil$-qubit register can be prepared with $\mathcal{O}(k \log n)$ gates and zero ancillas~\cite{gleinig2021efficient,ramacciotti2024simple,li2025nearly}. Therefore, $P_{\chi_B}\ket{0}=\ket{\chi_B}$ can be implemented with $C_\chi^{\mathrm{prep}} = \mathcal{O}(N_B \log n)$ gates, while $P_{\eta_B}\ket{0}=\ket{\eta_B}$ has gate complexity $C_\eta^{\mathrm{prep}} = \mathcal{O}(r_B \log n)$.

The total gate complexity of applying the algorithm of \Cref{thm:sparse-trace-be}, substituting \Cref{eq:lambda-B-bound} and including the cost of sparse state preparation, is therefore
\begin{equation}
\label{eq:cost-agnostic}
    Q_{\mathrm{AE}} \cdot (C_T(\varepsilon_T) + C_\chi^{\mathrm{prep}} + C_\eta^{\mathrm{prep}}) =
    \mathcal{O}\mathopen{}\left( \frac{\alpha_T\|B\|_{\max}r_B\sqrt{s_B}}{\epsilon} \left[ C_T(\varepsilon_T) + s_B r_B\log n + r_B\log n \right] \log\frac{1}{p_f} \right),
\end{equation}
where $C_T(\varepsilon_T)$ is the gate cost of one block-encoding for $T$ with error $\varepsilon_T$. We recall that this error must be controlled as $\varepsilon_T\leq \frac{\epsilon}{2\Lambda_B}$, and furthermore that $C_T(\varepsilon_T)$ only depends polylogarithmically on $1/\varepsilon_T$.

\subsection{Estimation cost for \texorpdfstring{$m$-RPA}{m-RPA}}

For the $m$-RPA application, we let $B\equiv B^{1,1}$ be the two-body-integral observable and $T\equiv T_{1,1}^{(m)}$ be the block-encoded amplitudes. Recall that $m=1$ gives ordinary RPA, while $m>1$ gives the extended $m$-RPA model. Since the $m$-RPA matrix $B$ only has a nonzero $1p1h$ block, the correlation-energy density is
\[
\frac{E_c^{(m)}}{V}=\frac{1}{4V}\Tr(BT),
\]
where $V$ is the physical volume. Therefore, an additive error $\epsilon_{c}$ in ${E_c^{(m)}}/{V}$ is obtained by applying \Cref{thm:sparse-trace-be} with tolerance $\epsilon = 4V\epsilon_c$. This yields an estimator $\widehat e_c^{(m)}=\frac{\alpha_T\Lambda_B}{4V}\widehat A_{B,T}$ that satisfies
\begin{equation*}
\left| \widehat e_c^{(m)}-\frac{E_c^{(m)}}{V}\right|\leq\epsilon_c
\label{eq:energy-estimate-error}
\end{equation*}
with probability at least $1-p_f$. In particular, the target error for block-encoding $T$ can be chosen as $\varepsilon_T\leq 2V\epsilon_c/\Lambda_B$, and the number of queries in amplitude estimation becomes
\[
Q_{\mathrm{AE}}= \mathcal{O}\mathopen{}\left(\frac{\alpha_T\Lambda_B}{V\epsilon_c}\log\frac{1}{p_f}\right).
\]
Under the localization model analyzed in \Cref{app:appD}, we have:
\begin{equation}
s_B = \mathcal{O}(R_c^D), \qquad r_B = \mathcal{O}(VR_c^D), \qquad N_B = \mathcal{O}(VR_c^{2D}).
\end{equation}
Additionally, $\|B\|_\mathrm{max} = \mathcal{O}(1)$. Adapting the gate cost in \Cref{eq:cost-agnostic}, we have a total gate cost of
\begin{equation}\label{eq:C_total_prelim}
    C_E = \mathcal{O}\mathopen{}\left( \frac{\alpha_T R_c^{3D/2}}{\epsilon_c}[C_T(\varepsilon_T) + V R_c^{2D}\log n + V R_c^{D}\log n] \log\frac{1}{p_f} \right) \text{ gates to estimate } \frac{E_c^{(m)}}{V} \pm \epsilon_c.
\end{equation}

To analyze this expression further, we unpack the physical scalings of $\alpha_T$ and $C_T(\varepsilon_T)$. First, from the quantum CARE solver algorithm (\Cref{thm:care-be}), we have $\alpha_T = \mathcal{O}(\kappa_2)$ where $\kappa_2 = \mathcal{O}(M_\gamma \alpha_{\mH}\|T\|)$ by combining \Cref{thm:kappa2_bound} with $\alpha_\Pi = \mathcal{O}(M_\gamma\alpha_\mH)$ from \Cref{cor:block-encoding-riesz-smooth}. Assuming physicality of the solution, $\|T\| = \mathcal{O}(1)$, we get $\alpha_T = \mathcal{O}(\kappa_2) = \mathcal{O}(M_\gamma \alpha_\mH)$.

The cost of block-encoding the stabilizing solution $T$ is, from \Cref{cor:qsvt-P2inv} and \Cref{thm:care-be},
\[C_T(\varepsilon_T) = \mathcal{O}\mathopen{}\left(M_\gamma \alpha_\mH \kappa_2 \log\mathopen{}\left(\frac{1}{\varepsilon_\mathrm{pol}^{(\Pi)}}\right) \log\mathopen{}\left(\frac{1}{\varepsilon_\mathrm{pol}^{(+)}}\right) C_\mH \right),
\]
where $\varepsilon_\mathrm{pol}^{(\Pi)}$ and $\varepsilon_\mathrm{pol}^{(+)}$ are chosen such that we achieve the target error $\varepsilon_T$. Because these dependencies are polylogarithmic, we suppress their contributions henceforth. Under the sparsity assumptions described in the main text, the cost of block-encoding $\mH$ is $C_\mH = \mathcal{O}(VR_c^{2D})$. Thus the cost of $T$ dominates the other two terms (from the state-preparation circuits for $\ket{\chi_B}$ and $\ket{\eta_B}$). Combined with the bound on $\alpha_T$, \Cref{eq:C_total_prelim} becomes (suppressing logarithmic factors)
\begin{equation}
C_E = \widetilde{\mathcal{O}}\mathopen{}\left( \frac{V M_\gamma^3 \alpha_\mH^3 R_c^{7D/2} }{\epsilon_c} \right).
\end{equation}

It remains to analyze $\alpha_\mH$. In general, $\alpha_\mH=\mathcal{O}(s_\mH \|\mH\|_{\max}) = \mathcal{O}(s_\mH)$. The row sparsity depends on the regime:~when single-electron scattering dominates, $s_{\mH}=\mathcal{O}(mR_c^D)$, so the end-to-end gate complexity is
\begin{equation}
C_E^{\mathrm{1e}} = \widetilde{\mathcal{O}}\mathopen{}\left(\frac{V M_\gamma^3 m^3 R_c^{13D/2}}{\epsilon_c}\right).
\end{equation}
On the other hand, in the two-electron scattering regime, $s_{\mH}=\mathcal{O}(m^2R_c^{2D})$, so we instead get
\begin{equation}\label{eq:C_quantum_double-elec}
C_E^{\mathrm{2e}} = \widetilde{\mathcal{O}}\mathopen{}\left(\frac{V M_\gamma^3 m^6 R_c^{19D/2}}{\epsilon_c}\right),
\end{equation}
which is the claim in the main text.
\end{document}